\documentclass[12pt]{article}
\usepackage[authoryear]{natbib}
\usepackage{titling}
\usepackage{graphicx}
\usepackage{setspace}
\usepackage{sectsty}
\usepackage{url}
\usepackage{amsfonts}
\usepackage{amsmath}
\usepackage{amsthm}
\usepackage{multirow}
\usepackage{booktabs}
\usepackage{caption}
\usepackage[skip=-10pt]{subcaption}
\usepackage{float}
\usepackage{longtable}
\usepackage{xcolor,colortbl}
\usepackage[english]{babel}
\usepackage{lscape}
\usepackage{multirow}
\usepackage{rotating}
\usepackage{pdflscape}

\newcommand{\EE}{{\mathord{I\kern -.33em E}}}
\newcommand{\PP}{{\mathord{I\kern -.33em P}}}

\newcommand{\utDelta}
{\renewcommand{\arraystretch}{0.5} \begin{array}[t]{c}{{\pmb\Delta}}\\
\vspace{-4.5mm} \\
  \widetilde{}
\end{array}\hspace{-2mm}
\renewcommand{\arraystretch}{1}}

\newcommand{\utR}
{\renewcommand{\arraystretch}{0.5}\hspace{-2mm}
\begin{array}[t]{c}{{{R}}}\\
\vspace{-4.5mm} \\
  \widetilde{}
\end{array}\hspace{-2mm}\vspace{-1.5mm}
\renewcommand{\arraystretch}{1}}

\newcommand{\utRRel}
{\renewcommand{\arraystretch}{0.5}\hspace{-2mm}
\begin{array}[t]{c}{{\mathbf{R}}}\\
\vspace{-4.5mm} \\
  \widetilde{}
\end{array}\hspace{-2mm}\vspace{-1.5mm}
\renewcommand{\arraystretch}{1}}

\newcommand{\utT}
{\renewcommand{\arraystretch}{0.5}\hspace{-2mm}
\begin{array}[t]{c}{{\mathbf{T}}}\\
\vspace{-4.5mm} \\
  \widetilde{}
\end{array}\hspace{-2mm}\vspace{-1.5mm}
\renewcommand{\arraystretch}{1}}

\DeclareRobustCommand{\utLbf}
{\renewcommand{\arraystretch}{0.5} \!\!\!\begin{array}[t]{c}{{\bf L} }\\
\vspace{-4.5mm} \\
   {\widetilde{}} \\ 
  \vspace{-10mm}
\end{array}\hspace{-2mm}
\renewcommand{\arraystretch}{1}}

\DeclareRobustCommand{\utLbfCaption}
{\renewcommand{\arraystretch}{0.5} \!\!\!\begin{array}[t]{c}{{\bf L} }\\
\vspace{-2.0mm} \\
  \widetilde{}
\end{array}\hspace{-2mm}
\renewcommand{\arraystretch}{1}}

\newcommand{\n}{^{(n)}}
\newcommand{\npr}{^{(n)\prime}}
\newcommand{\pr}{^{\prime}}
\newcommand{\cqfd}{\hfill $\square$}
\newcommand{\inv}{^{-1}}
\newcommand{\vecd}{ \mathrm{vecd}^{\circ} }

\newcommand{\fulan}{ \mathcal{F}_{\text{\tiny{ULAN}}}} 
\newcommand{\matd}{ \mathrm{matd}^{\circ}}
\newcommand{\half}{\frac{1}{2}}

\newcommand{\be}{\begin{equation}}
\newcommand{\en}{\end{equation}}

\newtheorem{lemma}
{Lemma}[section]
\newtheorem{proposition}
{Proposition}[section]
\newtheorem{corollary}
{Corollary}[section]
{Conjecture}[section]
\theoremstyle{definition}
{Definition}[section]
{Assumption}[section]

\title{\vspace{-25mm}
{\sc 
$R$-Estimation for Asymmetric\\ Independent Component Analysis
\vspace{-1mm}
}}
\author{ 
Marc {\sc Hallin}\thanks{Marc Hallin is Professor at ECARES, Universit\'{e} libre de Bruxelles, Ave.\ F.D.\  Roosevelt~50, CP 114/04, B-1050 Bruxelles, Belgium and ORFE, Princeton University, Sherrerd Hall,   Princeton, NJ 08544  (E-mail: mhallin@ulb.ac.be).  Research supported by the Sonderforschungsbereich ``Statistical modelling
of nonlinear dynamic processes" (SFB 823) of the Deutsche Forschungsgemeinschaft, and
the IAP research network grant~P7/06 of the Belgian Federal Government (Belgian Science
Policy).
} 
\  and
Chintan {\sc Mehta}\thanks{Chintan Mehta is Doctoral Student at ORFE, Princeton University, Sherrerd Hall,   Princeton, NJ 08544 (E-mail: cmehta@princeton.edu). Supported in part by the National Institutes of Health through Grant Number R01 GM072611 and National Science Foundation DMS-0704337.  } }
\date{\today}

\begin{document}
\maketitle

$\,$\vspace{-15mm}\\
\begin{abstract} Independent Component Analysis (ICA) recently has attracted much attention in the statistical literature as an appealing alternative to elliptical models. Whereas $k$-dimensional elliptical densities depend on  one single unspecified radial density, however, $k$-dimensional independent component distributions involve $k$ unspecified component densities. In practice, for  given sample size $n$ and  dimension $k$, this makes the statistical analysis much harder.  We focus here on the estimation, from an independent sample, of the mixing/demixing matrix of the model. Traditional methods (FOBI, Kernel-ICA,  FastICA) mainly originate from the engineering literature.  Their consistency requires moment conditions, they are poorly robust, and  do not achieve any type of asymptotic efficiency. When based on robust scatter matrices, the two-scatter methods developed by \citet{ose2006} and \citet{noo2008}  enjoy better robustness features, but their optimality properties remain unclear.  
The ``classical semiparametric" approach by \citet{cb2006}, quite on the contrary,  achieves semiparametric efficiency, but requires  the estimation of the  densities  of the $k$  unobserved independent components. As a reaction, an efficient (signed-)rank-based approach has been proposed by \citet{ip2011} for the case of symmetric component densities. The performance of theiir estimators is quite good, but they unfortunately fail to be root-$n$ consistent as soon as one of the component densities violates the symmetry assumption. In this paper, using ranks rather than signed ranks, we extend their approach to the asymmetric case and propose a one-step R-estimator for ICA mixing matrices. The finite-sample performances of those estimators are investigated and compared to those of existing methods under moderately large sample sizes. Particularly good performances are obtained from a version involving data-driven scores taking into account  the skewness and kurtosis of residuals. Finally, we show, by an empirical exercise,  that our methods also may provide excellent results in a context  such as image analysis, where   the basic assumptions of ICA are quite unlikely to hold. 
 \end{abstract}
 

\doublespace

\section{Introduction}
\subsection{Independent Component Analysis (ICA)} The traditional Gaussian model for noise, where a $k$-dimensional error term $\bf e$ is ${\cal N}({\bf 0}, {\pmb\Sigma})$ can be extended, mainly, into two directions. Either the elliptical density contours of the multinormal are preserved, and $\bf e$ is assumed to be {\it elliptically symmetric} with respect to the origin, with unspecified radial density $f$. Or, the independence of the marginals of~${\pmb\Sigma}^{-1/2}{\bf e}$ is preserved, but their densities $f_1,\ldots , f_k$ remain unspecified, yielding  the {\it independent component} model. In both cases, the distribution of $\bf e$ involves an unknown linear transformation: the~$k\times k$ symmetric positive definite {\it sphericizing matrix} ${\pmb\Sigma}^{-1/2}$ ($k(k+1)/2$ parameters) in the elliptical case; the $k\times k$ {\it mixing matrix} $\pmb\Lambda$ ($k^2$ parameters) in the independent component case. The main difference, however, is that, while elliptical noise only depends on {\it one} nonparametric nuisance, the radial density $f$, independent component noise involves $k$ nonparametric nuisances, the component densities~$f_1,\ldots , f_k$. This makes the statistical analysis of models based on independent component noise significantly harder than its elliptical counterpart.  

In this paper, we focus on the problem of estimating $\pmb\Lambda$. Many   solutions (FastICA,  FOBI, Kernel-ICA, ...)  have been proposed, mostly in the engineering  literature; see Section~\ref{Sec:EstimationMethod} for details. Their root-$n$ consistency requires finite moments---of order four (deflation-based Fast-ICA: see Olilla (2010), Nordhausen et al.~(2011) or Ilmonen~(2012)) or eight (FOBI: see Ilmonen et al.~(2010))---or remains an open question (Kernel-ICA). None of them is achieving efficiency nor enjoying any well-identified optimality property, and their robustness properties are  poor or likely to be poor (for FOBI, see Ilmonen et al.~(2010)).  An ingenious  method based on the availability of two scatter matrices has been developed by  \citet{ose2006} and \citet{noo2008}. Under appropriate assumptions on the component densities, the resulting estimators are root-$n$ consistent and, when based on robust scatter matrices, the method can be seen as a robustification of FOBI.  No particular efficiency property can be expected, though, and, as soon as one of the component densities is asymmetric, a  preliminary  symmetrization step may be  required, which is computationally demanding.  

In contrast with these approaches, a  rigorous asymptotic analysis of the problem, putting emphasis on asymptotic efficiency,  is performed by \citet{cb2006} in line with the classical  \citet{bkrw1993}  semiparametric methodology, based on tangent space projections. 
 In that approach,  the $k$ component densities $f_1,\ldots , f_k$ need to be estimated, which again is computationally very costly. 
 
As a reaction, an efficient rank-based method has been developed recently by  Ilmonen and Paindaveine~(2011). That method is exploiting the consistency properties of such estimators as FOBI or FastICA, or those based on the two-scatter method, and  taking into account the invariance and distribution-freeness features of ranks in order to bypass the costly step of estimating~$k$ densities. Their estimators---call them $R_+$-estimators---achieve semiparametric efficiency at some selected $k$-tuple of component densities, and yield very good finite-sample performances, even under moderately large samples. However, they are  based on marginal {\it signed ranks}, which requires the somewhat restrictive assumption that all component densities are symmetric. 

We show here how that unpleasant assumption can be avoided, and propose a one-step $R$-estimation procedure based on residual ranks rather than   residual {\it signed} ranks as in $R_+$-estimation. We establish the asymptotic root-$n$ consistency and  asymptotic normality of our $R$-estimators, and carefully study their finite-sample performances via simulations. In particular, we show how they improve on the traditional and two-scatter methods,  
  and outperform Ilmonen and Paindaveine's $R_+$-estimators as soon as the symmetry assumption is violated by one of the component densities.

$R$-estimation, as well as  $R_+$-estimation, requires choosing~$k$ score functions, a choice that, in this context,  may be somewhat difficult. We  therefore describe and recommend a version of our method based on data-driven scores, where the skewness and kurtosis of component residuals   can be  taken into account. That method is easily implementable, and  achieves particularly good results. 

Finally, with an application to image analysis, we  show that our method also provides good results in situations where the basic assumptions of ICA clearly do not hold. There, our $R$-estimators are shown to improve, quite substantially,  the demixing performances of such classical methods as FOBI, FastICA or Kernel-ICA.

\subsection{Notation, identifiability, and main assumptions}\label{subSec:Notation}

  Denote by $\mathbf{X}\n :=(\mathbf{X}\npr_1 ,\ldots , \mathbf{X}\npr_n)\pr$, $n\in\mathbb{N}$, with $\mathbf{X}\npr_i := (X\n_{i1}, \ldots , X\n_{ik})$, $i=1,\ldots ,n$, a triangular array of observed     $k$-dimensional random vectors   satisfying \vspace{-2mm}
\begin{equation}\label{ICmod}
\mathbf{X}\n_i ={\pmb\mu} + {\pmb\Lambda} \mathbf{Z}\n_i
\vspace{-2mm}\end{equation}
where $\mathbf{Z}\n :=(\mathbf{Z}\npr_1 ,\ldots , \mathbf{Z}\npr_n)\pr$ is an unobserved  $n$-tuple of i.i.d.\   $k$-dimensional \textit{latent vectors}~$\mathbf{Z}\npr_i := (Z\n_{i1}, \ldots , Z\n_{ik})$, $i=1,\ldots ,n$, with joint  and marginal densities  $f^{\mathbf{Z}}$ and~$f_1,\ldots , f_k$ such that  (the IC assumption)\vspace{-2mm}
\begin{equation}
\label{eq:indf}
f^{\mathbf{Z}}(\mathbf{z})=\prod_{j=1}^kf_j(z_j), \qquad \mathbf{z}=(z_1,\ldots , z_k)\in\mathbb{R}^k.
\vspace{-2mm}\end{equation}
 The $k\times 1$ vector $\pmb\mu$ and the $k\times k$ full-rank matrix $\pmb\Lambda$  are parameters; $\pmb\Lambda$ and its inverse~$\pmb{\Lambda}^{-1}$ are called the \textit{mixing} and \textit{demixing} (or \textit{unmixing}) matrices, respectively.  
 Under (\ref{eq:indf}), the~$k$ components $Z\n_{i1}, \ldots , Z\n_{ik}$ of the latent vectors $\mathbf{Z}\n_i$ are mutually independent: they are called the \textit{independent components}, and  their marginal probability densities~$f:=(f_1,\ldots , f_k)$   the \textit{component densities},  of the \textit{independent component model}~(\ref{ICmod})-(\ref{eq:indf}). 

  Identification constraints clearly are needed in order for $\pmb\mu$ and  $\pmb\Lambda$ to be identified. 
  Without any loss of generality, we throughout impose that $f\in\mathcal{F}_0$, where \vspace{-2mm}
$$\mathcal{F}_0:=\Big\{ f:=(f_1,\ldots , f_k) \  \vert  f_j(z)>0 \text{ for all } z\in\mathbb{R}, \text{ and }  \int_{-\infty}^0f_j(z)\mathrm{d}z=1/2=\int_{0}^\infty f_j(z)\mathrm{d}z \Big\}; \vspace{-2mm}$$
the vector  $\pmb\Lambda ^{-1}\pmb\mu$ then is identified as the componentwise median of the $\pmb\Lambda ^{-1}\mathbf{X}\n_i $'s. Identification issues  for~$\pmb\Lambda$ are more severe, due to the invariance of the IC assumption (\ref{ICmod}) and~(\ref{eq:indf})   under permutation, rescaling, and sign changes of the centered independent components $\mathbf{Z}\n_i -  \pmb\Lambda ^{-1}\pmb\mu$. 
Denoting by  $\mathbf{D}_1$ and $\mathbf{D}_2$  two arbitrary full-rank $k\times k$ diagonal matrices, and by $\mathbf P$  an arbitrary  $k\times k$ permutation matrix, we clearly have that  $\mathbf{\Lambda} \mathbf{Z} = \mathbf{\Lambda}^{\ast} \mathbf{Z}^{\ast}$ for~$\mathbf{\Lambda}^{\ast} = \mathbf{\Lambda} \mathbf{D}_1 \mathbf{P} \mathbf{D}_2$ and~$\mathbf{Z}^{\ast} = \mathbf{D}_2^{-1} \mathbf{P}^{-1} \mathbf{D}_1^{-1}\mathbf{Z}$, where $\mathbf{Z}^{\ast}$ still satisfies (\ref{ICmod}) and (\ref{eq:indf}). 
The mixing matrices ${\pmb\Lambda}$ and~$\mathbf{\Lambda}^{\ast}$ therefore are observationally equivalent.

  Several identification constraints have been proposed in the literature in order to tackle this identifiability issue. 
  Those  we are imposing here are borrowed from Ilmonen and Paindaveine~(2011). 
  Considering the equivalence classes of $k \times k$ nonsingular matrices associated with the equivalence relation  
$ {\mathbf{\Lambda}}^\ast \sim \mathbf{\Lambda}$ iff ${\mathbf{\Lambda}}^\ast = \mathbf{\Lambda} \mathbf{D}_1\mathbf{P} \mathbf{D}_2 $ 
for some permutation and full-rank diagonal matrices  $\mathbf{P}$, $\mathbf{D}_1$ and~$\mathbf{D}_2$, respectively, denote by $ \Pi$  the mapping \vspace{-3mm}
\begin{equation} 
\label{eq:PiMapping}
\mathbf{\Lambda} \mapsto {\Pi}(\mathbf{\Lambda}) := \mathbf{\Lambda} \mathbf{D}_1^{\mathbf{\Lambda}} \mathbf{P}^{\mathbf{\Lambda}} \mathbf{D}_2^{\mathbf{\Lambda}},\vspace{-3mm}
\end{equation}
where (a) $\mathbf{D}_1^{\mathbf{\Lambda}}$  is  the $k \times k$ positive 
 diagonal matrix whose $j^{th}$ diagonal element is the inverse of the Euclidean norm of $\pmb\Lambda$'s $j^{th}$ column ($j=1, \ldots,k$), (b) $\mathbf{P}^{\mathbf{\Lambda}}$ is a permutation matrix  that reorders the columns of $\pmb\Lambda\mathbf{D}_1^{\mathbf{\Lambda}}$  in such a way  that  $\left|( \mathbf{\Lambda} \mathbf{D}_1^{\mathbf{\Lambda}} \mathbf{P}^{\mathbf{\Lambda}} )_{ij} \right| <  \left| ( \mathbf{\Lambda}\mathbf{D}_1^{\mathbf{\Lambda}} \mathbf{P}^{\mathbf{\Lambda}})_{ii} \right|$ for all~$j > i$, and 
 (c) the (not necessarily positive) diagonal matrix $\mathbf{D}_2^{\mathbf{\Lambda}}$  normalizes $\mathbf{\Lambda} \mathbf{D}_1^{\mathbf{\Lambda}} \mathbf{P}^{\mathbf{\Lambda}} $ in such a way that  $( \mathbf{\Lambda} \mathbf{D}_1^{\mathbf{\Lambda}} \mathbf{P}^{\mathbf{\Lambda}} \mathbf{D}_2^{\mathbf{\Lambda}})_{jj}=1$, i.e. $(\mathbf{D}_2^{\mathbf{\Lambda}})_{jj} =  ( \mathbf{\Lambda}\mathbf{D}_1^{\mathbf{\Lambda}} \mathbf{P}^{\mathbf{\Lambda}} )_{jj}^{-1}$ for $j = 1, \ldots, k$.
Consider the set~$\mathcal{M}_k$ of nonsingular $k\times k$ matrices for which  no tie occurs in the definition of  $\mathbf{P}^{\mathbf{\Lambda}}$. 
Then, for~$\mathbf{\Lambda}_1, \mathbf{\Lambda}_2 \in \mathcal{M}_k$,  $\mathbf{\Lambda}_1 \sim \mathbf{\Lambda}_2$ if and only if $\Pi (\mathbf{\Lambda}_1) = \Pi (\mathbf{\Lambda}_2)$. 
Each class of equivalence thus contains a  unique  element $\mathbf{\Lambda}$ such that $\Pi(\mathbf{\Lambda}) = \mathbf{\Lambda}$, and inference for mixing matrices can be restricted to the set $\mathcal{M}_{k}^1:= \Pi(\mathcal{M}_k)$. 

  The matrices $\mathbf{\Lambda}$ for which ties occur in the construction  of $\mathbf{P}^{\mathbf{\Lambda}}$ have Lebesgue measure zero in $\mathbb{R}^{k \times k}$; neglecting them has little practical implications. 
While one could devise a systematic way to define a unique $\mathbf{P}^{\mathbf{\Lambda}}$ in the presence of such  ties, the resulting mapping~$\mathbf{\Lambda}\mapsto \mathbf{P}^{\mathbf{\Lambda}}$ would not be continuous, which disallows the use of the Delta method when constructing  root-$n$ consistent estimators for $\mathbf{\Lambda}$.  

  For $\mathbf{L} \in \mathcal{M}_{k}^1$, denote  by $\pmb \theta = (\pmb \mu, \text{vecd}^{\circ}(\mathbf{L}))$ the model parameter, where $\text{vecd}^{\circ}( \mathbf{L} )$ 
 stands for the vector of size $k(k-1)$ that stacks the columns of $\mathbf{L}$ on top of each with the diagonal elements omitted (since, by definition, they are set to one). 
 Write  $\Theta := \left( \mathbb{R}^k \times \text{vecd}^{\circ}(\mathcal{M}_{k}^1) \right)$ for  the parameter space. 
 Note that, by imposing scaling and some nonnegative asymmetry constraints on the component densities, one could add the (unique) diagonal matrix $\mathbf{D}^{\mathbf{\Lambda}}$ such that $\mathbf{D}_1^{\mathbf{\Lambda}} \mathbf{P}^{\mathbf{\Lambda}} \mathbf{D}_2^{\mathbf{\Lambda}}=\mathbf{P}^{\mathbf{\Lambda}} \mathbf{D}^{\mathbf{\Lambda}}$ to the list of (nuisance) parameters. 
 In the present context, it is more convenient to have it absorbed into the unspecified form of $f$.   
 The role of $\mathbf D^{\mathbf{\Lambda}}$ is quite similar, in that respect, to that of the   scale functional  in elliptical families, as discussed in \citet{hp2006}.  
 
  Another solution to those identification problems is adopted by \citet{cb2006}, who impose scaling restrictions of $f$, and then let their  PCFICA algorithm (\citet{cb2005}) make a choice between the various observationally equivalent values of   $\pmb\Lambda\inv$.
\vspace{-2mm}

\section{Local asymptotic normality and group invariance}
\subsection{Group Invariance and semiparametric efficiency}\label{subsec:groupInvariance}
  Denoting by $\mathrm{P}\n_{\pmb \theta ;f}$,  $\mathrm{P}\n_{\pmb \mu,  \mathbf{L};f}$ or $\mathrm{P}\n_{\pmb \mu, \text{vecd}^{\circ}(\mathbf{L});f}$  the joint distribution of $\mathbf{X}\n$ under location $\pmb\mu$, mixing matrix ${\pmb\Lambda}$ such that $\Pi(\pmb{\Lambda}) = \mathbf{L}$, and component densities $f=( f_1,\ldots , f_k)$, let 
 \vspace{-4mm}
$$ \mathcal{P}\n:=\big\{ \mathrm{P}\n_{\pmb \theta ;f} \ \vert \  \pmb\theta \in\Theta ,\ f\in\mathcal{F}_0\big\}, \qquad  \mathcal{P}\n_f:= \big\{ \mathrm{P}\n_{\pmb \theta ;f} \ \vert \  \pmb\theta \in\Theta \big\} \ \text{for fixed} \ f\in\mathcal{F}_0,$$
 \vspace{-12mm}
$$ \mathcal{P}\n_{\pmb\mu ; f}:= \big\{ \mathrm{P}\n_{\pmb \mu,  \mathbf{L};f}\ \vert \  \mathbf{L}  \in \mathcal{M}_{k}^1 \big\} \  \text{for fixed} \ \pmb\mu\in\mathbb{R}^k\ \text{and} \ f\in\mathcal{F}_0 , $$
 \vspace{-8mm}
$$\mathcal{P}\n_{\mathbf{L}} \text{ or } \mathcal{P}\n_{\pmb{\Lambda}} := \big\{  \mathrm{P}\n_{\pmb \mu, \mathbf{L};f} \ \vert \  \pmb\mu \in \mathbb{R},\ f\in\mathcal{F}_0 \big\} \ \text{for fixed} \ \Pi({\pmb\Lambda })=\mathbf{L} \in \mathcal{M}_{k}^1,  \ \ \ \text{and}  $$
 \vspace{-8mm}
$$ \mathcal{P}\n_{\pmb{\mu}, \mathbf{L}} \text{ or } \mathcal{P}\n_{\pmb{\mu}, \pmb{\Lambda}} := \big\{  \mathrm{P}\n_{\pmb \mu, \mathbf{L};f} \ \vert   \ f\in\mathcal{F}_0  \big\} \  \text{for fixed} \ \pmb{\mu}\in\mathbb{R}^k \ \text{and} \ \Pi({\pmb\Lambda })=\mathbf{L} \in \mathcal{M}_{k}^1.$$  
All those subfamilies will play a role in the sequel. 
 
  A semiparametric (in the spirit of \citet{bkrw1993}) approach to Independent Component Analysis (ICA) and, more particularly,  the estimation of $\pmb\Lambda$,  requires the  \textit{uniform local asymptotic normality} (ULAN) of $\mathcal{P}\n_{ f}$ at any~$f$ satisfying adequate regularity assumptions: see Section~\ref{ULANsec}. 
It is easy to see that ULAN of $\mathcal{P}\n_{ f}$ (with parameters $\pmb\mu$ and $\bf L$) implies that of~$\mathcal{P}\n_{\pmb\mu ; f}$ (with parameter  $\bf L$)  for any given $\pmb\mu\in\mathbb{R}^k$. 

  The model we are interested in involves the family $\mathcal{P}\n$. 
Depending on the context, several distinct semiparametric approaches to ICA are possible: either both the location~$\pmb\mu$ and the mixing matrix~$\pmb\Lambda$ are parameters of interest with the density $f$ being a nuisance; or the location~$\pmb\mu$ is a parameter of interest with nuisance~$(\pmb\Lambda , f)$; or the mixing matrix~$\pmb\Lambda$ (equivalently, $\mathbf{L}$) only is of interest and $(\pmb\mu , f)$ is a nuisance. 
\citet{hw2003} have shown that, under very general conditions, if the parametric submodels associated with fixed values of the nuisance are \textit{uniformly locally asymptotically normal} (ULAN), while the submodels associated with fixed values of the parameter of interest are generated by groups of transformations, then semiparametrically efficient inference  can be based on the maximal invariants of those groups. 
  
  In the present context,   $\pmb\Lambda$ is the parameter of interest, and $(\pmb\mu , f)$ is the nuisance. 
Consider~$f=(f_{*1},\ldots , f_{*k})$, and assume that 
\begin{enumerate}
\item[(A1)] 
  $f$ belongs to the subset ${\mathcal F}_{\text{\tiny{ULAN}}}$ of ${\cal{F}}_0$  such that the sequence of (parametric) subfamilies~$\mathcal{P}\n_{\pmb\mu ; f}$, with parameter $\mathbf{L}$, is  ULAN, with \textit{central sequence} ${\pmb\Delta}\n_{\pmb\mu ; f}(\mathbf{L})$ (actually, ULAN holds at any~$(\pmb\mu , f )$ iff it holds  at $(\mathbf{0}  , f)$), and 

\item[(A2)] for all $\mathbf{L}  \in \mathcal{M}_{k}^1$ and $n\in\mathbb{N}$, the (nonparametric) subfamily $\mathcal{P}\n_{\mathbf{L}}$ is generated by some group of transformations  $\mathcal{G}\n(\mathbf{L} ), {\scriptstyle\circ}$ acting on the observation space  $\mathbb{R}^{kn}$, with maximal invariant~$\mathbf{R}\n(\mathbf{L} )$.
\end{enumerate}
It follows from \citet{hw2003} that the  semiparametric efficiency bounds (at~$(\pmb\mu , f )$, if  $\mathbf L$ is the parameter of interest)  can be achieved  by basing inference on the maximal invariant~$\mathbf{R}\n(\mathbf{L} )$---more specifically, on the conditional expecta\-tion~$\mathrm{E}_{ \mathrm{P}\n_{\pmb \mu, \mathbf{L};f}}[{\pmb\Delta}\n_{\pmb\mu ; f}({\bf L})\vert \  \mathbf{R}\n(\mathbf{L} )]$; since  $ \mathbf{R}\n(\mathbf{L} )$ is invariant, that conditional expectation moreover  is distribution-free under~$\mathcal{P}\n_{\mathbf{L}}$ (hence, also under densities $f$ that do not necessarily belong to ${\mathcal F}_{\text{\tiny{ULAN}}}$).  

Section~\ref{ULANsec} establishes the  ULAN property (A1) of  $\mathcal{P}\n_{\pmb\mu ; f}$ for any $\pmb\mu$ and $f$ satisfying some mild regularity assumptions. 
Let us show here that (A2) holds for any $\mathbf{L}  \in \mathcal{M}_{k}^1$ and $n$, and that the maximal invariant is the vector $ \mathbf{R}\n(\mathbf{L} ) = (\mathbf{R}\npr_{1}(\mathbf{L} ),\ldots , \mathbf{R}\npr_{n}(\mathbf{L} ))\pr$, where~$\mathbf{R}\n_{i}(\mathbf{L} ) = ({R}\n_{i1}(\mathbf{L} ),\ldots , {R}\n_{ik}(\mathbf{L} ))\pr$ 
and ${R}\n_{ij}(\mathbf{L} )$ is the rank of $(\mathbf{L}^{-1}\mathbf{X}\n_{i})_j$ among $(\mathbf{L}^{-1}\mathbf{X}\n_{1})_j,\ldots , (\mathbf{L}^{-1}\mathbf{X}\n_{n})_j$. Letting \vspace{-2mm}
\begin{equation}
\label{eq:residuals}
\mathbf{Z}_i^{(n)}\left(\pmb \mu ,  \mathbf{L} \right) := \mathbf{L}^{-1} ( \mathbf{X}_i^{(n)} - \pmb \mu  ), \quad  i=1, \ldots, n, 
\vspace{-2mm}\end{equation}
${R}\n_{ij}(\mathbf{L} )$,  under $\mathcal{P}\n_{\mathbf{L}}$, is thus also  the rank of $\big(\mathbf{Z}\n_{i}(\pmb \mu ,  \mathbf{L} )\big)_j$ among $\big(\mathbf{Z}\n_{1}(\pmb \mu ,  \mathbf{L} )\big)_j,\ldots , \big(\mathbf{Z}\n_{n}(\pmb \mu ,  \mathbf{L} )\big)_j$.

The elements $ g_{\bf h} $ of the generating group $\mathcal{G}\n(\mathbf{L} ),{\scriptstyle\circ} $ are  indexed by the family $\mathcal H$ of $k$-tuples ${\bf h}=(h_1,\ldots , h_k)$ of monotone continuous and strictly increasing functions $h_j$ from $\mathbb{R}$ to  $\mathbb{R}$ such that $\lim_{z\to\pm\infty} h_j(z)= \pm\infty$, with  $g_{\bf h}\in \mathcal{G}\n(\mathbf{L} )$  defined as \vspace{-2mm}
$$g_{\bf h} : {\bf x}=({\bf x}\pr_1, \ldots , {\bf x}\pr_n)\pr = \big( (x_{11},\ldots , x_{1k}), \ldots , (x_{n1},\ldots , x_{nk})\big)\pr \in\mathbb{R}^{kn}\mapsto g_{\bf h} ({\bf x})
\vspace{-4mm}$$
where \vspace{-3mm}
$$g_{\bf h} ({\bf x})=  \Big({\bf L}\big(h_1(({\bf L}^{-1}{\bf x}_1)_{1}),\ldots ,  h_k(({\bf L}^{-1}{\bf x}_1)_{k})\big)\pr ,  \ldots , 
{\bf L}\big(h_1(({\mathbf{L}}^{-1}{\bf x}_{n})_1),\ldots , h_k(({\mathbf{L}}^{-1}{\bf x}_{n})_k\big)\pr \Big)\pr.\vspace{-2mm}$$
That is, $\mathcal{G}\n(\mathbf{L} ),{\scriptstyle\circ}$ is a transformation-retransformation form of the group   of continuous marginal order-preserving transformations acting componentwise on the $\mathbf{L}^{-1}\mathbf{X}\n_i$'s. 
Standard results on ranks entail that this group is generating $\mathcal{P}\n_{\mathbf{L}}$ and has maximal invariant~$ \mathbf{R}\n(\mathbf{L} )$.

  A similar situation holds when the parameter of interest is $({\pmb\mu}, {\bf L})$; similar ideas then lead to considering a smaller group $\mathcal{G}_0\n(\mathbf{L} )$, with maximal invariant the componentwise signs and ranks extending  the methods proposed in Hallin et al.~(2006, 2008).    This latter approach is not needed  here, where we focus on $R$-estimation of~$\bf L$, but it  is considered in \citet{hm2013}, who study  testing problems for location and regression. 

  The approach by \citet{ip2011} is quite parallel. 
However, although  addressing the problem of estimating the mixing matrix $\pmb\Lambda$, so that $\pmb\mu$ is a nuisance, these authors   do not consider the group $\mathcal{G}\n(\mathbf{L} )$, nor the group $\mathcal{G}_0\n(\mathbf{L} )$.  
They rather  make  the additional assumption that the $k$ component densities $f_j$ all are symmetric with respect to the origin. 
Under that assumption, they are using yet another group, which is the subgroup $\mathcal{G}_+\n(\mathbf{L} )$ of $\mathcal{G}\n(\mathbf{L} )$ corresponding to those   ${\bf h}\in\mathcal{H}$ such that $h_j(-z)= -h_j(z) $ for all $j=1,\ldots ,k$ and~$z\in\mathbb{R}$.  
The resulting maximal invariant is a vector of componentwise \emph{signed ranks}, that is,  the vector of componentwise residual signs, 
 along with the vector   $ \mathbf{R}_+\n(\pmb{\mu}, \mathbf{L} ) = (\mathbf{R}\npr_{+1}(\pmb{\mu}, \mathbf{L} ),\ldots , \mathbf{R}\npr_{+n}(\pmb{\mu}, \mathbf{L} ))\pr$, where $\mathbf{R}\n_{+i}(\pmb{\mu}, \mathbf{L} ) = ({R}\n_{+i1}(\pmb{\mu}, \mathbf{L} ),\ldots , {R}\n_{+ik}(\pmb{\mu}, \mathbf{L} ))\pr$,  with  
 ${R}\n_{+ij}(\pmb{\mu}, \mathbf{L} )$  the rank of $\big\vert \big(  \mathbf{Z}_i^{(n)}\left(\pmb \mu ,  \mathbf{L} \right)\big)_j\big\vert $ among $\big\vert  \big( \mathbf{Z}_1^{(n)}\left(\pmb \mu ,  \mathbf{L} \right)\big)_j\big\vert ,\ldots , \big\vert  \big( \mathbf{Z}_n^{(n)}\left(\pmb \mu ,  \mathbf{L} \right)\big)_j\big\vert $.  
 As a result, their estimators lose root-$n$ consistency as soon as one of the underlying $f_j$'s fails to be symmetric with respect to zero---an assumption that   hardly can be checked for. 

\subsection{Uniform local asymptotic normality (ULAN)}\label{ULANsec}
  Establishing ULAN 
   requires regularity conditions on $f$.  
  The following conditions are sufficient for $f=(f_1,\ldots , f_k)$ to belong to $\mathcal{F}_{ \text{\tiny{ULAN}} }$. \vspace{-2mm}
\begin{enumerate} 
\item[(A3)] The component densities $f_j$, $j=1,\ldots , k$, are \textit{absolutely continuous}, that is, there exist~$k$ real-valued functions $\dot{f}_j$ such that, for any $a<b$, $f_j(b) - f_j(a)=\int_a^b\dot{f}_j(z)\mathrm{d}z$.\vspace{-2mm}
\end{enumerate}
Letting $
\pmb{\varphi}_f (\mathbf{z}) := \left( \varphi_{f_1} ( z_{1}), \ldots, \varphi_{f_k}( z_{k} ) \right)\pr$,  $\mathbf{z}=(z_1,\ldots , z_k)\pr \in\mathbb{R}^k$,
with $\varphi_{f_j} := -f_j'/f_j$, assume moreover that \vspace{-2mm}
\begin{enumerate} 
\item[(A4)] all  component densities $f_j$ admit finite second-order moments, finite information for location, and finite information for scale; i.e. for $j =1, \ldots, k$, 
$\displaystyle{s^2_{f_j} := \int_{-\infty}^{\infty} z^2 f_j (z) \mathrm{d}z }$, 
$\displaystyle{\mathcal{I}_{f_j} := \int_{-\infty}^{\infty} \varphi_{f_j}^2(z) f_j(z) \mathrm{d}z}$, 
and $\displaystyle{\mathcal{J}_{f_j} := \int_{-\infty}^{\infty} z^2 \varphi_{f_j}^2 (z) f_j(z) \mathrm{d}z}$ are finite. \vspace{-2mm}
\end{enumerate}
For such $f$, it follows from the Cauchy-Schwarz inequality that 
$
\displaystyle{\alpha_{f_j} := \int_{-\infty}^{\infty} z f_j(z) \mathrm{d}z }$ \linebreak and $\displaystyle{
 \kappa_{f_j} := \int_{-\infty}^{\infty} \varphi_{f_j}^2( z) z f_j(z) \mathrm{d}z}$, 
 $j=1,\ldots,k$, 
also are finite.
Consequently,  the quantities 
$
\gamma_{pq}(f) := \mathcal{I}_{f_p} s^2_{f_q}$, 
$\varsigma_{pq}(f) := \alpha_{f_p} \kappa_{f_q}$,  and 
 $ \varrho_{jpq}(f) := \mathcal{I}_{f_j} \alpha_{f_p} \alpha_{f_q},
$ 
are bounded for every~$j, p, q\in \left\{1, \ldots, k \right\}$.
The information matrix for the ULAN result, in Proposition \ref{prop:ParametricULAN} below, depends on these quantities through\vspace{-2mm}
\begin{eqnarray}
{\bf G}_f &:=& \sum_{j=1}^k \left( \mathcal{J}_{f_j} - 1 \right) \left( \mathbf{e}_j \mathbf{e}_j^{\prime} \otimes  \mathbf{e}_j \mathbf{e}_j^{\prime} \right) +  \sum_{ \substack{p, q=1 \\ p \ne q} }^k \left\{ \gamma_{qp}(f)  \big( \mathbf{e}_p \mathbf{e}_p^{\prime} \otimes  \mathbf{e}_q \mathbf{e}_q^{\prime} \big)  +  \big( \mathbf{e}_p \mathbf{e}_q^{\prime} \otimes  \mathbf{e}_q \mathbf{e}_p^{\prime} \big) \right\} \nonumber \\
&&~~ + \sum_{ \substack{p, q=1 \\ p \ne q} }^k \mathbf{e}_p \mathbf{e}_q^{\prime} \otimes  \big( \varsigma_{pq}(f) \mathbf{e}_q \mathbf{e}_q^{\prime} + \varsigma_{qp}(f) \mathbf{e}_p \mathbf{e}_p^{\prime} \big)  +\!\!\!\!  \sum_{\substack{j, p,q =1 \\ j\ne p, j \ne q, p \ne q}}^k   \varrho_{jpq}(f)  \left( \mathbf{e}_p \mathbf{e}_q^{\prime} \otimes  \mathbf{e}_j \mathbf{e}_j^{\prime} \right)\! , \label{eq:Gf}
\end{eqnarray}
where $\mathbf{e}_j$ is the $j$th canonical basis vector of $\mathbb{R}^k$ and $\otimes$ denotes the Kronecker product.

Writing $\mathbf{I}_k$  for the  $k \times k$ identity matrix,  define 
$\mathbf{C} := \sum_{p=1}^k \sum_{q=1}^{k-1} \mathbf{e}_p \mathbf{e}_p\pr \otimes \mathbf{u}_q \mathbf{e}\pr_{q + \delta_{q \geq p}} ,$ 
 where~$\mathbf{u}_q$ is the $q$th canonical basis vector of $\mathbb{R}^{k-1}$ and $\mathbf{e}_{q + \delta_{q \geq p}} := \delta_{q \geq p} \mathbf{e}_{q+1} + (1 - \delta_{q \geq p}) \mathbf{e}_q$, with $ \delta_{q \geq p}$   the indicator for $q\geq p$.
Then, let $\text{odiag}( {\bf M} )$ replace the diagonal entries of a matrix ${\bf M}$ with zeros.
Finally, for any ${\bf m} \in \mathbb{R}^{k(k-1)}$, define $\text{matd}^{\circ}( {\bf m})$ as the unique  $k \times k$ matrix with a diagonal of  zeroes such that $\vecd ( \text{matd}^{\circ}( {\bf m})) = {\bf m}$.

\begin{proposition} \label{prop:ParametricULAN} 
Let $f \in \mathcal{F}_0$ satisfy (A3) and (A4). 
Then, $f \in \fulan$, and, for any fixed~$\pmb\mu \in \mathbb{R}^k$, the sequence of subfamilies $\mathcal{P}\n_{\pmb\mu ; f}$, with parameter ${\bf L} \in \mathcal{M}_k^1$, is ULAN  with central sequence\vspace{-4mm}
\begin{equation}
\label{eq:OPTcenSeq}
\pmb{\Delta}_{{\bf L};\ \pmb{\mu}, f}^{(n)} =  \mathbf{C} \big( \mathbf{I}_k  \otimes \mathbf{L}^{-1} \big)^{\prime} \mathrm{vec}\big[  {\bf T}\n_{{\bf L};\ {\pmb \mu}, f}  \big],~~\text{where}~~{\bf T}\n_{{\bf L};\ {\pmb \mu}, f} := n^{-\half} \sum_{i=1}^n \big( \pmb{\varphi}_f \big( {\bf Z}_i\n \big) {\bf Z}_i\npr - {\bf I}_k \big)
\vspace{-2mm}\end{equation}
where~${\bf Z}_i\n := \mathbf{Z}_i^{(n)}\big(\pmb \mu ,  \mathbf{L} \big)$ is~defined in (\ref{eq:residuals}), and full-rank information matrix\vspace{-2mm}
\begin{equation}
\label{eq:InfMatrix}
\mathbf{\Gamma}_{\mathbf{L}; f} := \mathbf{C} \big( \mathbf{I}_k \otimes \mathbf{L}^{-1} \big)^{\prime} {\bf G}_f  \big( \mathbf{I}_k \otimes \mathbf{L}^{-1} \big) \mathbf{C}^{\prime},
\vspace{-2mm}\end{equation}
with ${\bf G}_f$ defined in (\ref{eq:Gf}).
Specifically, for any sequence  ${\bf L}\n =  {\bf L}  + O(n^{-\half}) \in \mathcal{M}_k^1$ and any bounded sequence $\pmb \tau\n \in \mathbb{R}^{k(k-1)}$,  \vspace{-3mm}
\begin{equation}
\label{eq:ParametricULANRepresentation}
\log   \frac{ \mathrm{d}\mathrm{P}\n_{  \pmb{\mu}, {\bf L}\n  + n^{-\half} \matd ( {\pmb \tau}\n ); f} }{  \mathrm{d}\mathrm{P}\n_{ \pmb{\mu}, {\bf L}\n\!;  f} }  =  \pmb{\tau}^{(n) \prime} \pmb{\Delta}_{\mathbf{L}\n\! ;\ {\pmb \mu}, f}^{(n)} - \frac{1}{2}  \pmb \tau^{(n) \prime} \mathbf{\Gamma}_{\mathbf{L}; f}  \pmb \tau^{(n)} + o_{\mathrm{P}}(1)
\vspace{-2mm}\end{equation}
and  
$\displaystyle{\pmb{\Delta}_{{\mathbf{L}}^{(n)}\!;\  {\pmb \mu}, f}^{(n)} \overset{\mathcal{L}}{\longrightarrow} \mathcal{N}_{k(k-1)} \big( \mathbf{0}, \mathbf{\Gamma}_{\mathbf{L}; f} \big)}$,
 as $n \to \infty$ under $\mathrm{P}_{\pmb \mu, {\mathbf{L}}^{(n)};  f}^{(n)}$.
\end{proposition}

This ULAN property extends that established by  Oja et al.~(2010) under the additional assumption that each component density $f_j$ is symmetric. 
Symmetry for every $f_j$ implies that the quantities $\alpha_{f_j}$ and $\kappa_{f_j}$, hence also the quantities $\varsigma_{jp}$ and $\varrho_{jpq}$, all take value zero for~$j, p, q \in \{1, \ldots, k\}$; therefore, dropping this assumption of symmetry affects the information matrix~(\ref{eq:InfMatrix}) through ${\bf G}_f$ in (\ref{eq:Gf}), which explains why our $\mathbf{\Gamma}_{\mathbf{L}; f}$ differs from theirs.\vspace{-1mm}

%

\subsection{Rank-based versions of central sequences} 
  The ULAN result from Proposition~\ref{prop:ParametricULAN} allows the construction of parametrically efficient inference procedures  for  $\mathbf{L} \in \mathcal{M}_k^1$ at any given $f$ and $\pmb\mu$. 
  In practice, these are unspecified nuisances;  misspecifying either or both of them, in general,  leads to invalid inference---tests that fail to reach the nominal asymptotic level and estimators that do not achieve root-$n$ consistency. 
Therefore, the semiparametric approach under which both $f$ and $\pmb\mu$ are unspecified is the most sensible one. Instead of the standard semiparametric approach 
  of  \citet{cb2006}, which requires estimating the $k$ component density scores, we consider the result of \citet{hw2003} who show that, under very general conditions, the parametric central sequence conditioned  on the maximal invariant mentioned in (A2) is a version (central sequences are always defined up to $o_{\rm P}(1)$ quantities) of  the corresponding semiparametrically efficient central sequence based on the tangent space projection.  
  
  Let $\mathbf{F}: \mathbb{R}^k \longrightarrow [0,1]^k$ and  $\mathbf{J}_f: [0,1]^k \longrightarrow \mathbb{R}^k$  be defined so that, for $\mathbf{z} = \left(z_1, \ldots, z_k\right)\pr \in~\mathbb{R}^k$, ~$ \mathbf{F}(\mathbf{z}) := \left( F_1(z_1), \ldots, F_k(z_k) \right)\pr$, with 
 $ F_j(z_j) := \int_{-\infty}^{z_j} f_j(z) \mathrm{d}z$ for  
  $j=1, \ldots , k$,  
and,   for  \linebreak $ \mathbf{u} = (u_1, \ldots, u_k)\pr \in [0,1]^k$,
 $
\mathbf{J}_f (\mathbf{u}) := \pmb{\varphi}_f \left( \mathbf{F}^{-1} \left( \mathbf{u} \right) \right) = \left( \varphi_{f_1} \left( F_1\inv \left( u_1\right) \right), \ldots,  \varphi_{f_k} \left( F_k\inv \left( u_k\right) \right) \right)\pr,
$ 
with  $J_{f_j}(u_j) := \varphi_{f_j}\left( F_j\inv (u_j ) \right)$ for $j=1, \ldots, k$. 
Writing 
$\mathbf{U}_i\n := (U\n_{i,1},\ldots , U\n_{i,k})\pr$ \linebreak  for~$ \mathbf{U}_i\n \left(\pmb \mu, \mathbf{L}\right)   := \mathbf{F} \left( \mathbf{Z}_i^{(n)}\left(\pmb \mu, \mathbf{L}\right) \right)$,  $i=1, \ldots, n$, the parametric statistic ${\bf T}\n_{{\bf L}; \ {\pmb \mu}, f}$ defined in~(\ref{eq:OPTcenSeq}) takes the form 
$
{\bf T}\n_{{\bf L}; \ {\pmb \mu}, f}  =  n^{-\half} \sum_{i=1}^n  \left( \mathbf{J}_f  ( \mathbf{U}_i\n )  \mathbf{F}^{-1\prime}( \mathbf{U}_i\n )  - \mathbf{I}_k \right).
$

Assume moreover that 
\begin{enumerate}
\item[(A5)] for all $j=1,\ldots , k$, $z\mapsto  \varphi_{f_j}(z)$ is the difference of two monotone increasing functions.\vspace{-1mm}
\end{enumerate}
Assumption (A5) will be required whenever rank-based statistics with scores $\varphi_{f_j}\circ F_j^{-1}$ are considered. 
Conditioning $\pmb{\Delta}_{\mathbf{L}; \ {\pmb \mu}, f}\n$  on the sigma-field $\mathcal{B}(\mathbf{L})$ generated by the marginal  ranks of the  $\mathbf{Z}_i^{(n)}(\pmb \mu, \mathbf{L})$'s  yields \vspace{-4mm}
\begin{eqnarray}
\utDelta\n_ {\mathbf{L}; f: \mathrm{ex}}  &:=& \mathrm{E} \left[ \pmb{\Delta}_{\mathbf{L}; \ {\pmb \mu}, f}\n  | \mathcal{B}(\mathbf{L}) \right] \nonumber \\ 
&=& \mathbf{C} \left( \mathbf{I}_k  \otimes \mathbf{L}^{-1} \right)^{\prime} \mathrm{vec} \big[ \utT\n_{ {\bf L}; f:\mathrm{ex}} \big] \quad \text{where} \quad \utT\n_{ {\bf L}; f:\mathrm{ex}}  := \mathrm{E} \big[ {\bf T}\n_{{\bf L}; \ {\pmb \mu}, f} \big| \mathcal{B}(\mathbf{L}) \big]; \label{exdelta}
\end{eqnarray}
 \vspace{-10mm}
 
\noindent clearly, $\utDelta\n_ {\mathbf{L}; f: \mathrm{ex}}$ does not depend on $\pmb\mu$. 
Computing this conditional expectation requires evaluating, for each $j\in\{1,\ldots k\}$ and $r \in \{1, \ldots, n \}$,  \vspace{-2mm}
\begin{equation}
\mathrm{E} \left[ J_{f_{j}} ( U\n_{i,j}   ) F_{j}^{-1}(  U\n_{i,j} )  |  R\n_{i,j}(\mathbf{L}) =r \right]=\mathrm{E} \left[ J_{f_{j}} ( U_{(r)}\n  ) F_{j}^{-1}(  U_{(r)}\n )   \right]
\vspace{-3mm}\label{EXSCOR1}\end{equation}
 and, for each $j\pr \neq j^{\prime\prime}\in\{1,\ldots k\}$ and $r, s \in \{1, \ldots, n \}$, \vspace{-2mm}
\begin{equation}
\mathrm{E} \left[ J_{f_{j\pr}} ( U\n_{i,j\pr}   ) F_{j^{\prime\prime}}^{-1}(  U\n_{i,j^{\prime\prime}} )  |  R\n_{i,j^{\prime} }(\mathbf{L}) =r,\ R\n_{i,j^{\prime\prime} }(\mathbf{L}) =s \right]
=\mathrm{E} \left[ J_{f_{j\pr}} ( U_{(r)}\n  ) \right] \mathrm{E} \left[ F_{j^{\prime\prime}}^{-1}(  U_{(s)}\n )   \right],
\label{EXSCOR2}\vspace{-3mm}\end{equation}
 where $U_{(r)}\n $ and $U_{(s)}\n $ respectively denote,  in a sample~$U_1,\ldots , U_n$ of i.i.d.\  random variables uniform over $(0,1)$, the $r$th and $s$th order statistics.  
 As a function of $r$ and $s$, such  quantities are  called  \emph{exact} scores;   they  depend on $n$, and computing them via  numerical integration is somewhat tedious. 
 
 The so-called \textit{approximate scores},  in general,  are preferable: denoting~by \vspace{-1mm}
\begin{eqnarray*}
\utRRel _i^{(n)} \left( \mathbf{L} \right) :=  
  \Big(  
\utR \n_{i,1}\left( \mathbf{L}\right) , \ldots , \utR \n_{i,k}\left( \mathbf{L} \right)    \Big)\pr  :=
\bigg(  
\frac{ R_{i,1}^{(n)} \left( \mathbf{L} \right)}{n+1} , \ldots , \frac{ R_{i,k}^{(n)}\left( \mathbf{L} \right)  }{n+1}  \bigg)\pr 
\vspace{-3mm}
\end{eqnarray*}\vspace{-10mm}

\noindent the (marginal) \textit{normalized ranks}, 
  the approximate scores corresponding to (\ref{EXSCOR1}) and (\ref{EXSCOR2})~are \vspace{-10mm}
 \begin{eqnarray}
 J_{f_{j}} \big( \utR \n_{i,j}\left( \mathbf{L}\right) \big) F_{j}^{-1} \big(  \utR \n_{i,j}\left( \mathbf{L}\right) \big) - \frac{1}{n} \sum_{i=1}^n  J_{f_{j}} \Big(  \frac{i}{n+1} \Big) F_{j}^{-1}\Big(  \frac{i}{n+1} \Big)   \quad\text{and}  \nonumber \\
 J_{f_{j\pr}} \big( \utR \n_{i,j\pr}\left( \mathbf{L}\right) \big)  F_{j^{\prime\prime}}^{-1} \big(  \utR \n_{i,j^{\prime\prime}}\left( \mathbf{L}\right) \big) - \frac{1}{n} \sum_{i=1}^n  J_{f_{j\pr}} \Big(  \frac{i}{n+1} \Big) \frac{1}{n} \sum_{i=1}^n F_{j\prime\prime}^{-1}\Big(  \frac{i}{n+1} \Big),
 \label{APPRSCOR}\vspace{-2mm}
 \end{eqnarray}
 respectively. 
 Letting ${\bf 1}_k \in \mathbb{R}^k$ be the $k$-dimensional vector of ones, the approximate-score version of the central sequence is thus
\vspace{-6mm}
\begin{equation}\label{apprdelta}
\utDelta\n_ {\mathbf{L}; f} := \mathbf{C} \left( \mathbf{I}_k  \otimes \mathbf{L}^{-1} \right)^{\prime} \mathrm{vec} \Big[ \utT_{{\bf L}; f}\n \Big], \qquad \text{where} 
\end{equation}
\vspace{-10mm}
\begin{equation}
\utT_{{\bf L}; f}\n :=  \mathrm{odiag}\Big[ n^{-\half} \sum_{i=1}^n  \Big( \mathbf{J}_f \big(   \utRRel_i^{(n)} ( \mathbf{L} )  \big) \mathbf{F}^{-1\prime} \big( \utRRel_i^{(n)} ( \mathbf{L} ) \big) -  \overline{\mathbf{J}}_f\n \overline{{\bf F}\inv}\npr \Big) \Big]  \label{rankScoreMatrix}
\end{equation}
with ~$\overline{\mathbf{J}}_f\n := \frac{1}{n} \sum_{i=1}^n \mathbf{J}_f \Big( \frac{i}{n+1} {\bf 1}_k \Big)$ and $\overline{{\bf F}\inv}\n := \frac{1}{n} \sum_{i=1}^n {\bf F}\inv \Big( \frac{i}{n+1} {\bf 1}_k \Big)$.

The following  proposition, by establishing the asymptotic equivalence between the exact- and approximate-score forms~(\ref{exdelta}) and~(\ref{apprdelta}),   shows that (\ref{apprdelta}) indeed is a version of the corresponding semiparametrically efficient central sequence for the problem. \vspace{-1mm}

\begin{proposition} \label{prop:RankCenSeqEquivalence}
Fix $\pmb \mu \in \mathbb{R}^k$, $\mathbf{L} \in \mathcal{M}_{k}^1$, and $f \in \fulan$ satisfying (A5). 
Then,  under $\mathrm{P}^{(n)}_{\pmb \mu, \mathbf{L}; f}$, 
\vspace{-7mm}
\begin{equation*}
(i) \utDelta\n_ {\mathbf{L}; f} = \utDelta\n_ {\mathbf{L}; f: \mathrm{ex}} + o_{L^2}(1) \ \ \ \qquad \text{and} \ \ \ \qquad (ii) \utDelta\n_ {\mathbf{L}; f} = \mathbf{\Delta}^{(n) \ast}_{\mathbf{L}, \pmb \mu; f} + o_{L^2}(1),  \vspace{-2mm}
\end{equation*}
as $n \to \infty$, where $\mathbf{\Delta}^{(n) \ast}_{\mathbf{L}, \pmb \mu; f}$ is a semiparametrically efficient (at  $\mathbf{L}$, $\pmb \mu$, and $f$) central sequence. 
\end{proposition}

  Consequently, $\utDelta\n_ {\mathbf{L}; f}$ can be used to construct semiparametrically efficient (at  $f$, irrespective of $\pmb \mu$) estimation procedures for $\mathbf{L}$. 
Contrary to those based on $\mathbf{\Delta}^{(n) \ast}_{\mathbf{L}, \pmb \mu; f}$, the $R$-estimators derived from $\utDelta\n_ {\mathbf{L}; f}$ remain root-$n$ consistent, though, under most component densities $g \in \mathcal{F}_{\text{ULAN}}$, $g\neq f$. 
And, unlike those proposed by~Ilmonen and Paindaveine~(2011), they do not require $f$ nor $g$ to be symmetric. 

  The asymptotic representation for the rank-based central sequence $\utDelta\n_ {\mathbf{L}; f}$~ under  $\mathrm{P}\n_{\pmb \mu, \mathbf{L}; g}$   where~$g \in \mathcal{F}_0$ is not necessarily equal to $f \in \fulan$ is provided in the next proposition. 
  If, additionally, $g \in \fulan$, the asymptotic distribution for $\utDelta\n_ {\mathbf{L}; f}$ can be made explicit. For every $p \ne q \in \left\{1, \ldots, k \right\}$, let \vspace{-2mm}
  \begin{eqnarray*}
\gamma^{\ast}_{pq}(f, g) &:=&   \int_0^1 \varphi_{f_p}\big( F_p^{-1}(u) \big) \varphi_{g_p}\big( G_p^{-1}(u) \big) \mathrm{d}u \Big( \smallint_0^1 F_q^{-1}(u) G_q^{-1}(u) \mathrm{d}u - \alpha_{f_q} \alpha_{g_q} \Big) \\
\text{and} \quad \rho^{\ast}_{pq}(f, g) &:=&  \int_0^1 F_p^{-1}(u) \varphi_{f_p}\big( G_p^{-1}(u) \big) \mathrm{d}u  \smallint_0^1 \varphi_{f_q} \big( F_q^{-1}(u) \big) G_q^{-1}(u) \mathrm{d}u.  
\end{eqnarray*}
The quantities $\gamma^{\ast}_{pq}(f, g) $ and $\rho^{\ast}_{pq}(f, g)$ are referred to as \emph{cross-information} quantities; note that $\gamma^{\ast}_{pq} (f, f) = \gamma_{pq}(f) - \varrho_{pqq}(f)$ and $\rho^{\ast}_{pq}(f, f) = 1$.  
 Then, define\vspace{-2mm}
   \begin{equation}
\label{eq:CrossInfMatrix}
\mathbf{\Gamma}^{\ast}_{ \mathbf{L}; f, g} :=  \mathbf{C} \left( \mathbf{I}_k \otimes \mathbf{L}^{-1}\right)^{\prime} \tilde{\mathbf{G}}_{f, g} \left( \mathbf{I}_k \otimes \mathbf{L}^{-1}\right) \mathbf{C}^{\prime}\vspace{-3mm}
\end{equation}
where $ \tilde{\bf G}_{f,g} :=  \sum_{ {p\neq q=1 } }^k  \gamma^{\ast}_{sr} (f, g) \big(  {\bf e}_p {\bf e}_p\pr \otimes {\bf e}_q {\bf e}_q\pr  \big)+ \rho^{\ast}_{pq}(f,g) \big( {\bf e}_p {\bf e}_q\pr \otimes {\bf e}_q {\bf e}_p\pr \big)$, 
and write $\mathbf{\Gamma}^{\ast}_{ \mathbf{L}, f}$ for~$\mathbf{\Gamma}^{\ast}_{ \mathbf{L}, f, f}$.
Remark that $\mathbf{\Gamma}^{\ast}_{ \mathbf{L}, f, g} $ depends on $g$ only through 
  $\gamma^{\ast}_{pq}(f, g) $ and $\rho^{\ast}_{pq}(f, g)$.

%
\begin{proposition} 
\label{prop:AsymptoticRepresentation}
Fix $f \in \fulan$,  $\pmb{\mu} \in \mathbb{R}^k$, and ~$\mathbf{L} \in \mathcal{M}_{k}^1$; with ${\bf Z}_i\n\! := \mathbf{Z}_i^{(n)}\big( {\pmb \mu}, \mathbf{L} \big)$ defined in~(\ref{eq:residuals}), let $\widetilde{\mathbf{J}}_f\n\! := \frac{1}{n}  \sum_{i=1}^n  \mathbf{J}_f \big(  {\bf G} \big( {\bf Z}_i\n \big)  \big)$ and $\widetilde{{\bf F}}^{-1 (n)} \! := \frac{1}{n} \sum_{i=1}^n  \mathbf{F}\inv \big(  {\bf G} \big( {\bf Z}_i\n \big) \big)$. Then, 
\begin{enumerate}
\item[(i)]  If $g \in \mathcal{F}_0$,  
 $
\utDelta\n_ {\mathbf{L}; f} = \pmb{\Delta}^{\diamond (n)}_{ {\bf L}, {\pmb \mu}; f, g}  + o_{L^2}(1)
$ 
 as $n \to \infty$, under $\mathrm{P}^{(n)}_{\pmb \mu, {\bf L}, g}$, where 
  \vspace{-4mm}
  $$\pmb{\Delta}^{\diamond (n)}_{ {\bf L}, {\pmb \mu}; f, g} := \mathbf{C} \left( \mathbf{I}_k  \otimes \mathbf{L}^{-1} \right)^{\prime} \mathrm{vec} \Big[ {\bf T}^{\diamond (n)}_{ {\bf L}, {\pmb \mu}; f, g} \Big], \quad\text{and}
\vspace{-9mm}  $$
  \begin{equation}
\label{AsymptoticRankScoreMatrix}
{\bf T}^{\diamond (n)}_{ {\bf L}, {\pmb \mu}; f, g} := \mathrm{odiag}\big[ n^{-\half} \sum_{i=1}^n  \big( \mathbf{J}_f (  {\bf G}( {\bf Z}_i\n )  ) \mathbf{F}^{-1\prime}(  {\bf G}( {\bf Z}_i\n ) ) -  \widetilde{\mathbf{J}}_f\n \widetilde{{\bf F}}^{-1 (n) \prime} \big) \big].
  \vspace{-4mm}\end{equation}
 \item[(ii)] Suppose furthermore that $g \in \fulan$, and fix ~$\pmb{\tau} \in \mathbb{R}^{k(k-1)}$ so that ~$\mathbf{L}  + n^{-\half} \matd ( \pmb{\tau}) \in~\mathcal{M}_k^1$. Then,
$\utDelta\n_ { \mathbf{L}; f} \overset{\mathcal{L}}{\longrightarrow} \mathcal{N}_{k(k-1)} \left(  \mathbf{\Gamma}_{\mathbf{L}; f, g}^{\ast}  \pmb{ \tau}, \mathbf{\Gamma_{ \mathbf{L}; f}^{\ast}} \right)$ as $n \to \infty$, under $\mathrm{P}_{ \pmb{\mu},  \mathbf{L}  + n^{-\half} \matd ( \pmb{\tau} ); g}^{(n)}$ with~$\mathbf{\Gamma}_{\mathbf{L}; f, g}^{\ast}$ defined in (\ref{eq:CrossInfMatrix}). 
If ${\pmb \tau} = {\bf 0}_{k(k-1)}$,  $g \in \mathcal{F}_0$  is sufficient  for this convergence to hold. \vspace{-1mm}
\item[(iii)] If, again, $g \in \fulan$ and  $\pmb{\tau} \in \mathbb{R}^{k(k-1)}$ is as defined in (ii), then, as $n \to \infty$, under~$\mathrm{P}_{ \pmb{\mu}, \mathbf{L}; g}^{(n)}$, \vspace{-4mm}
\begin{equation}
\label{eq:AsymptoticLinearity}
\utDelta\n_ {  \mathbf{L}  + n^{-\half} \matd ( \pmb{\tau}) ; f} - \utDelta\n_ { \mathbf{L}; f} = - \mathbf{\Gamma}_{\mathbf{L}; f, g}^{\ast}  \pmb{ \tau} + o_{\mathrm{P}}(1) .
\end{equation}
\end{enumerate}
\end{proposition}
In Section \ref{sec:OneStepREstimation}, our $R$-estimation procedures require evaluating the $f$-score rank-based central sequence, for $f \in \fulan$, at a preliminary root-$n$ consistent estimator $\tilde{\bf L}\n$ of $\bf L$.
The asymptotic impact of substituting $\tilde{\bf L}\n$ for ${\bf L}$ does not directly follow from Proposition~\ref{prop:AsymptoticRepresentation}(iii) because the perturbation $\pmb{\tau}$   in (\ref{eq:AsymptoticLinearity}) is a deterministic quantity.   Lemma 4.4 in  \citet{k1987} provides sufficient conditions for Proposition \ref{prop:AsymptoticRepresentation}(iii) to hold when replacing $\pmb{\tau}$ with a sequence of random vectors, $\tilde{\pmb{\tau}}\n$, $n \in \mathbb{N}$. More precisely, if 
  \begin{enumerate}
  \item[(C1a)]    $\tilde{\pmb{\tau}}\n =  O_{\mathrm{P}}(1)$, as $n \to \infty$, and  
  \item[(C1b)] 
there exists an integer $N < \infty$ so that, for all $n \geq N$, $\tilde{\pmb{\tau}}\n$ can take, at most, a finite number of values within any  bounded ball    centered at the origin in $\mathbb{R}^{k(k-1)}$, \vspace{-2mm}
  \end{enumerate}
 hold, then (\ref{eq:AsymptoticLinearity}) is still valid with  $\pmb{\tau}$  replaced by $\tilde{\pmb{\tau}}\n$. 

Let $\tilde{\bf L}\n \in \mathcal{M}_k^1$ be an estimator for ${\bf L}$. We say that it is \emph{root-$n$ consistent} under $\mathrm{P}\n_{{\pmb \mu}, {\bf L}; g}$ and \emph{locally asymptotically discrete}  if  $n^{\half} \vecd \big( \tilde{\bf L}\n - {\bf L} \big)$ satisfies (C1a) under $\mathrm{P}\n_{{\pmb \mu}, {\bf L}; g}$ and (C1b). Proposition \ref{prop:AsymptoticRepresentation}(iii) and Lemma 4.4 from \citet{k1987}  then yield the following corollary. 
 
\begin{corollary} 
\label{cor:kreissresult}
Fix $\pmb{\mu} \in \mathbb{R}^k$, ${\bf L} \in \mathcal{M}_k^1$ and $f, g \in \fulan$. 
Suppose that $\tilde{\bf L}\n$ is root-$n$ consistent under $\mathrm{P}\n_{\pmb{\mu}, {\bf L}; g}$ and locally asymptotically discrete. 
Then, under $\mathrm{P}\n_{\pmb{\mu}, {\bf L}; g}$, as $n \to \infty$,\vspace{-3mm}
\begin{equation}
\label{eq:Kreiss}
\utDelta\n_ { \tilde{\bf L}\n, f} - \utDelta\n_ { \mathbf{L}, f} = - \mathbf{\Gamma}_{\mathbf{L}, f, g}^{\ast}  \vecd \big( \tilde{\bf L}\n  - {\bf L} \big) + o_{\mathrm{P}}(1).
\vspace{-2mm}\end{equation}
\end{corollary}
The asymptotic discreteness requirement for the preliminary estimator is not overly restrictive. 
Any root-$n$ consistent sequence 
  $\tilde{\bf L}\n := (\tilde{L}_{rs}\n ) \in \mathcal{M}_k^1$ indeed can be discretized as $\tilde{\bf L}\n_{\#} := (\tilde{L}_{rs; \#}\n )$, with 
 $
\tilde{L}\n_{rs; \#} := \big(c n^{\half} \big)^{-1} \mathrm{sign} \big( \tilde{L}_{rs}\n \big)  \Big \lceil c n^{\half} \big| \tilde{L}_{rs}\n \big| \Big \rceil $, for   $r \ne s \in \{1, \ldots, k\}$, where~$c > 0$ is an arbitrary constant and $\lceil x \rceil$ denotes the smallest integer greater than or equal to $x$.  
The root-$n$ consistency properties of $\tilde{\bf L}\n$ carry over to $\tilde{\bf L}\n_{\#}$ which by construction is locally asymptotically discrete and, because  $\mathcal{M}_k^1$ is a compact subset of $\mathbb{R}^{k(k-1)}$, still takes values in  $\mathcal{M}_k^1$. \vspace{-2mm}

\section{$R$-estimation of the mixing matrix}\label{sec:OneStepREstimation}
Assume that a rank test rejects $\text{H}_0: \pmb{\theta} =  \pmb{\theta}_0$ against the alternative $\text{H}_1: \pmb{\theta} \ne \pmb{\theta}_0$ for large values of  some test statistic $Q_{{\pmb \theta}_0} \big(  {\bf R}\n \big( \pmb{\theta}_0 \big) \big)$ measurable with respect to the ranks $ {\bf R}\n \big( \pmb{\theta}_0)$ of residuals ${\bf Z}\n (\pmb{\theta}_0 ) := \big( {\bf Z}\n_1 (\pmb{\theta}_0), \ldots , {\bf Z}\n_n (\pmb{\theta}_0) \big)\pr$, which are~i.i.d. if and only if ${\pmb \theta} = {\pmb \theta}_0$.  The original $R$-estimator  for   ${\pmb \theta} \in \Theta$, as proposed by~Hodges and Lehmann~(1963), is  defined as~$\hat{\pmb{\theta} }\n_{\text{\tiny{HL}}} := \arg \min_{ \pmb{\theta} \in \Theta } Q_{{\pmb \theta}}\n \left(  {\bf R}\n \left( \pmb{\theta} \right) \right) $.

Even for   simple problems such as location, regression, etc. involving a low-dimensional parameter ${\pmb \theta}$, minimizing   $Q\n_{{\pmb \theta}} \big(  {\bf R}\n \big( \pmb{\theta} \big) \big)$ is wrought with difficulty---as a function of $\pmb{\theta}$, it is piecewise constant, discontinuous, and   non-convex. 
In the present case of a~$k(k-1)$-dimensional parameter space $\mathcal{M}_k^1$, solving this problem typically would require an infeasible grid-search in relatively high dimension. 

As an alternative, we consider the one-step $R$-estimators described in Hallin et al.~(2006) and \citet{hp2013}, that lead to expedient computation, and provide a consistent estimator for the asymptotic covariance matrix as a by-product. 
Those one-step estimators are computed from a preliminary root-$n$ consistent estimator $\tilde{\bf L}\n$ and the resulting value $\utDelta\n_ {\tilde{\bf L}\n; f}$ of the rank-based central sequence  associated  with some   reference density~$f \in \fulan$ satisfying~(A5). \vspace{-2mm}

\subsection{One-step $R$-estimation}\label{subsec:onestepRest}
For fixed $f \in \fulan$, assume that\vspace{-1mm}
\begin{enumerate}
\item[(C1)] there exists a sequence of estimators $\tilde{\bf L}\n \in \mathcal{M}_k^1$ of the parameter ${\bf L} \in \mathcal{M}_k^1$ that are both root-$n$ consistent and locally asymptotically discrete, under $\mathrm{P}\n_{{\pmb \mu}, {\bf L}; g}$ for any ${\pmb \mu} \in \mathbb{R}^k$, ${\bf L} \in \mathcal{M}_k^1$, and $g \in \fulan$, and, furthermore,\vspace{-2mm}
\item[(C2)] 
for all $p \ne q \in \left\{ 1, \ldots, k \right\}$, there exist consistent (under   $\mathcal{P}_g\n$ for every $g \in \mathcal{F}_{\text{\tiny{ULAN}}}$) and locally asymptotically discrete sequences $\hat{\gamma}^{\ast}_{pq }(f) $ and $\hat{\rho}^{\ast}_{pq }(f) $ of estimators for the cross-information quantities $\gamma^{\ast}_{pq}(f,g)$ and $\rho^{\ast}_{pq}(f,g)$.\vspace{-1mm}
\end{enumerate}

For any $f \in \mathcal{F}_{\text{\tiny{ULAN}}}$,  the one-step $R$-estimator for $\mathbf{L} \in \mathcal{M}_k^1$ based on $f$-scores is the $k \times k$ matrix $\utLbf_f\n \in \mathcal{M}_k^1$ defined by\vspace{-3mm}
\begin{equation}
\label{MM:onestep}
\text{vecd}^{\circ} \big( \utLbf_f\n  \big) = \text{vecd}^{\circ} \big( \tilde{\mathbf{L}}\n \big) + n^{-\half} \big( \hat{\mathbf{\Gamma}}_{\tilde{\mathbf{L}}\n; f}^{\ast} \big)^{-1}  \utDelta_{ \tilde{\mathbf{L}}^{(n)}; f},\vspace{-4mm}
\end{equation}
where $\hat{\mathbf{\Gamma}}_{\tilde{\mathbf{L}}\n; f}^{\ast}$ is a consistent estimate of $\mathbf{\Gamma}_{ \mathbf{L}; f, g}^{\ast}$. 
This estimator is constructed by plugging~$\hat{\gamma}^{\ast}_{pq}(f)$ and $\hat{\rho}^{\ast}_{pq}(f)$ into     (\ref{eq:CrossInfMatrix}). 
Under Assumptions (C1) and (C2),  $\hat{\mathbf{\Gamma}}_{\tilde{\mathbf{L}}\n; f}^{\ast}$ is a consistent estimate of $\mathbf{\Gamma}_{ \mathbf{L}; f, g}^{\ast}$. 
The procedure for obtaining each estimate $\hat{\gamma}^{\ast}_{pq}(f)$ and~$\hat{\rho}^{\ast}_{pq}(f)$ satisfying~(C2) is discussed in Section~\ref{sec:CIQ}.
The next proposition establishes the asymptotic distribution of~$\,\utLbf\n_f$; its proof parallels that of Theorem 5.1 in \citet{ip2011}.
\begin{proposition}
\label{prop:REstAsymptoticDistribution}
Fix a reference density $f \in \mathcal{F}_{\text{\tiny{ULAN}}}$. Then,\vspace{-2mm}
\begin{enumerate} 
\item[(i)] for any $\pmb{\mu} \in \mathbb{R}^k$, $\mathbf{L} \in \mathcal{M}_k^1$, and $g \in \mathcal{F}_{\text{\tiny{ULAN}}} $, the one-step $R$-estimator (\ref{MM:onestep}) is such that\vspace{-3mm}
\begin{equation}
\label{eq:REstAsymptoticDistribution}
n^{\half} \mathrm{ vecd}^{\circ}\left( \utLbf\n_f - \mathbf{L} \right)  \overset{\mathcal{L}}{\longrightarrow} \mathcal{N}_{k(k-1)} \left( \mathbf{0}, \left( \mathbf{\Gamma}_{{\bf L}; f,g }^{\ast} \right)\inv \mathbf{\Gamma}_{\mathbf{L}; f}^* \left( \mathbf{\Gamma}_{{\bf L}; f,g }^{\ast} \right)\inv \right)
\vspace{-3mm}\end{equation}
as $n \to \infty$, under $\mathrm{P}^{(n)}_{\pmb{\mu}, \mathbf{L}, g}$;
\item[(ii)] if, moreover, $f = g$, then $\left( \mathbf{\Gamma}_{{\bf L}; f,g }^{\ast} \right)\inv \mathbf{\Gamma}_{\mathbf{L}; f}^* \left( \mathbf{\Gamma}_{{\bf L}; f,g }^{\ast} \right)\inv = \big( \mathbf{\Gamma}^*_{\mathbf{L}; f} \big)\inv$, and $\, \utLbf\n_f $ is a semiparametrically efficient (at $f$) estimate of ${\bf L}$. 
\end{enumerate}
\end{proposition}

The $R$-estimator $\utLbf\n_f$ can be written in a form that avoids inverting $\hat{\mathbf{\Gamma}}_{\tilde{\mathbf{L}}\n; f}^{\ast}$, which can be numerically singular when estimated in practice.
Define therefore the $k \times k$ matrices~$\hat{\mathcal{A}}\n_{\tilde{\mathbf{L}}\n, f} := ( \hat{\alpha}_{pq}(f) )_{p,q=1}^k$ and $\hat{\mathcal{B}}\n_{\tilde{\mathbf{L}}\n, f} := ( \hat{\beta}_{pq}(f) )_{p,q=1}^k$ with zeroes on the diagonal 
  and, for every $p \ne q \in \{1, \ldots, k\}$,\vspace{-2mm}
$$
\hat{\alpha}\n_{pq} (f)  := \frac{ \hat{\gamma}^{\ast}_{pq}(f) }{ \hat{\gamma}^{\ast}_{pq}(f)  \hat{\gamma}^{\ast}_{qp}(f) - \hat{\rho}^{\ast}_{pq}(f) \hat{\rho}^{\ast}_{qp}(f) } \quad \text{and} \quad \hat{\beta}\n_{pq} (f)  := \frac{ -\hat{\rho}^{\ast}_{pq}(f)  }{ \hat{\gamma}^{\ast}_{pq}(f)  \hat{\gamma}^{\ast}_{qp}(f) - \hat{\rho}^{\ast}_{pq}(f) \hat{\rho}^{\ast}_{qp}(f)  }. 
$$
Letting ${\bf A} \odot {\bf B} = (a_{pq} b_{pq})$ denote the Hadamard product between two matrices ${\bf A} = (a_{pq})$ and ${\bf B} = (b_{pq})$ of the same size, define\vspace{-3mm}
\begin{equation}
\label{eq:onestepv2-intermediary}
\hat{ {\bf N}}\n_{\tilde{\mathbf{L}}\n, f} := \big( \hat{ \mathcal{A}}_{\tilde{\mathbf{L}}\n, f}^{(n) \prime}   \odot \utT_{\tilde{\bf L}\n; f}\n \big) + \big(  \hat{ \mathcal{B}}_{\tilde{\mathbf{L}}\n, f}^{(n) \prime} \odot \utT_{\tilde{\bf L}\n; f}\npr \big),
\end{equation} 
with $\utT_{{\bf L}; f}\n$ defined in (\ref{rankScoreMatrix}).
Theorem 5.2 in \citet{ip2011} then implies that $\utLbf\n_f$ can be expressed as \vspace{-3mm}
\begin{equation}
\label{MM:onestepv2}
\utLbf\n_f  = \tilde{\bf L}\n + n^{-\half} \tilde{\bf L}\n \Big[ \hat{\mathbf{N}}_{\tilde{\mathbf{L}}\n, f}\n - \text{diag}\big( \tilde{\bf L}\n \hat{\mathbf{N}}_{\tilde{\mathbf{L}}\n, f}\n \big) \Big].\vspace{-4mm}
\end{equation}

\subsection{Consistent estimation of cross-information quantities}\label{sec:CIQ}
A critical point in computing $\utLbf\n_f$ (\ref{MM:onestepv2}) is the consistent estimation of the cross-information quantities   in $\mathbf{\Gamma}^{\ast}_{\mathbf{L}, f; g}$. 
To tackle this issue, we exploit the asymptotic linearity  (\ref{eq:AsymptoticLinearity})  of $\utDelta_{ \tilde{\mathbf{L}}\n; f}$  using a method first proposed by Hallin et al.~(2006) in the context of the $R$-estimation of a scatter matrix in an elliptical model, and further developed by Cassart et al.~(2010) and \citet{hp2013}.  
In the present case, we have to consistently estimate a total of $2k(k-1)$ cross-information quantities appearing in  $\mathbf{\Gamma^{\ast}}_{ \mathbf{L}; f, g}$. 

Fixing $f \in \mathcal{F}_{\text{\tiny{ULAN}}}$, define, for  $\lambda \in {\mathbb R}$ and $r \ne s \in \{1, \ldots, k\}$,   the mappings\vspace{-3mm}
\begin{eqnarray}
\lambda \mapsto h^{\gamma^{\ast}_{rs}}(\lambda) := \big(   \utT\n_{\tilde{\mathbf{L}}^{(n)}; f} \big)_{rs} \big(  \utT\n_{ \tilde{\mathbf{L}}^{\gamma^{\ast}_{rs}}_{\lambda}; f}  \big)_{rs}  \quad \text{and} \quad  \lambda \mapsto h^{\rho^{\ast}_{rs}}(\lambda) := \big(   \utT\n_{\tilde{\mathbf{L}}^{(n)}; f} \big)_{sr} \big(  \utT\n_{ \tilde{\mathbf{L}}^{\rho_{rs}}_{\lambda}; f}  \big)_{sr} \label{eq:LALinearity} \vspace{-3mm}
\end{eqnarray}\vspace{-12mm}

\noindent 
(from $\mathbb{R} _+$ to $\mathbb{R}$), where \vspace{-3mm}
\begin{eqnarray*} 
\tilde{\mathbf{L}}^{\gamma^{\ast}_{rs}}_{\lambda} &:=& \tilde{\bf L}\n+ n^{-\half} \lambda \big(  \utT\n_{\tilde{\bf L}\n , f}  \big)_{rs} \tilde{\mathbf{L}}^{(n)} \big( \mathbf{e}_r \mathbf{e}_s^{\prime} - \mathrm{diag}\big( \tilde{\bf L}\n \mathbf{e}_r \mathbf{e}_s^{\prime} \big) \big) \quad \text{and}  \\
\tilde{\mathbf{L}}^{\rho^{\ast}_{rs}}_{\lambda} &:=& \tilde{\bf L}\n + n^{-\half} \lambda \big( \utT\n_{\tilde{\bf L}\n , f}   \big)_{sr} \tilde{\bf L}\n  \big( \mathbf{e}_r \mathbf{e}_s^{\prime} - \mathrm{diag}\big( \tilde{\bf L}\n \mathbf{e}_r \mathbf{e}_s^{\prime} \big) \big), 
\end{eqnarray*}
\vspace{-12mm}

\noindent with $\utT\n_{{\bf L}, f} $ defined in (\ref{rankScoreMatrix}).\vspace{-1mm}
Assume, additionally, that 
\begin{enumerate}
\item[(C3)] for fixed $f, g \in \mathcal{F}_{\text{\tiny{ULAN}}}$, $\pmb{\mu} \in \mathbb{R}^k$, and ${\bf L} \in \mathcal{M}_k^1$, the sequence $\tilde{\mathbf{L}}^{(n)}$ of preliminary estimators (satisfying (C1)) is such that each element in $\utT\n_{\tilde{\mathbf{L}}^{(n)}; f}$ is bounded from below by a positive constant with probability tending to one under $\mathrm{P}\n_{\pmb{\mu}, \mathbf{L}; g}$.
More precisely, for all~$\epsilon > 0$, there exist $\delta_{\epsilon} > 0$ and an integer $N_{\epsilon}$ such that 
$\mathrm{P}_{\pmb{\mu}, {\bf L}, g}\n \left[ \left(   \mathbf{T}_{\tilde{\mathbf{L}}\n; f} \right)_{rs} > \delta_{\epsilon} \right] \geq 1 - \epsilon $ for all~$n \geq N_{\epsilon}$ and $r \ne s \in \{1, \ldots, k \}$.\vspace{-1mm}

\end{enumerate}
This assumption is satisfied by most root-$n$ consistent estimators for the mixing matrix; see Section 4 for a discussion.   \vspace{-1mm}

The following lemma is adapted from 
  \citet{hp2013}.\vspace{-1mm}
\begin{lemma} 
\label{fScoreAsymLinearity}
Fix $f,g \in \mathcal{F}_{\text{\tiny{ULAN}}}$, $\pmb{\mu} \in \mathbb{R}^k$, and $\mathbf{L} \in \mathcal{M}_k^1$. Let $\tilde{\mathbf{L}}^{(n)}$ be a sequence of preliminary estimators for ${\bf L}$ satisfying   (C1) and (C3). 
For every $r \ne s \in \{1, \ldots, k\}$, the mappings $h^{\gamma^{\ast}_{rs}}$ and $h^{\rho^{\ast}_{rs}}$ defined in (\ref{eq:LALinearity}) satisfy, for all  $\lambda \in {\mathbb R}$,\vspace{-4mm}
$$
h^{\gamma^{\ast}_{rs}}(\lambda) = (1 - \lambda \gamma^{\ast}_{rs}(f, g) ) \left(   \mathbf{T}_{\tilde{\mathbf{L}}^{(n)}; f} \right)_{rs}^2 \! + o_{\mathrm{P}}(1) ~~ \text{and} ~~ h^{\rho^{\ast}_{rs}}(\lambda) = (1 - \lambda \rho^{\ast}_{rs}(f, g) ) \left(   \mathbf{T}_{\tilde{\mathbf{L}}^{(n)}; f} \right)_{sr}^2 \! + o_{\mathrm{P}}(1)
\vspace{-3mm}
$$
as $n \to \infty$, under $\mathrm{P}\n_{\pmb{\mu}, \mathbf{L}; g}$. Furthermore, each mapping is almost surely positive for $\lambda = 0$. \vspace{-1mm}
\end{lemma}

By Lemma \ref{fScoreAsymLinearity}, the mappings $h^{\gamma^{\ast}_{rs}}$ and $h^{\rho_{rs}}$ are both positive at $\lambda = 0$ and, up to $o_{\mathrm{P}}(1)$'s under $\mathrm{P}\n_{\pmb{\mu}, \mathbf{L}; g}$, are linear with a negative slope. 
Therefore, intuitively appealing estimators for $\gamma^{\ast}_{rs}(f, g)$ and $\rho^{\ast}_{rs}(f, g)$ would be, respectively, $\left( \hat{\gamma^{\ast}}_{rs}(f,g) \right)^{-1} :=  \inf_{\lambda} \left\{ \lambda \in \mathbb{R}: h^{\gamma^{\ast}_{rs}}(\lambda) < 0\right\}$ and  $\left( \hat{\rho^{\ast}}_{rs}(f,g) \right)^{-1} :=  \inf_{\lambda} \left\{ \lambda \in \mathbb{R}: h^{\rho^{\ast}_{rs}}(\lambda) < 0\right\}$; estimators for  $\rho_{rs}(f, g)$ would be defined in an analogous manner.
However, these estimators are not asymptotically discrete.
Instead, taking $\lambda_j = j/c$ for some large $c >0$ and $j \in \mathbb{Z}$, let\vspace{-4mm}
\begin{equation}
\label{crossInfoGamma}
\left( \hat \gamma^{\ast}_{rs}(f) \right)^{-1} := \lambda^-_{\gamma^{\ast}_{rs}}  +  { c^{-1} h^{\gamma^{\ast}_{rs}}(\lambda^-_{\gamma^{\ast}_{rs}} )  }/\big({  h^{\gamma^{\ast}_{rs}}(\lambda^-_{\gamma^{\ast}_{rs}} ) - h^{\gamma^{\ast}_{rs}}(\lambda^+_{\gamma^{\ast}_{rs}} ) \big)} , \vspace{-3mm}
\end{equation}
with $\lambda^-_{\gamma^{\ast}_{rs}}\! := \max_{j \in \mathbb{Z}} \!\left\{ \lambda_j\! : h^{\gamma^{\ast}_{rs}}(\lambda_j ) > 0 \right\}$ and $\lambda^+_{\gamma^{\ast}_{rs}(f)}\! := \min_{j \in \mathbb{Z}}\! \left\{ \lambda_j\! : h^{\gamma^{\ast}_{rs}}(\lambda_j ) < 0 \right\}$. 
Similarly~put\vspace{-4mm}
\begin{equation}
\label{crossInfoRho}
\left( \hat{\rho}^{\ast}_{rs}(f) \right)^{-1} := \lambda^-_{\rho^{\ast}_{rs}}  + { c^{-1} h^{\rho^{\ast}_{rs}}(\lambda^-_{\rho^{\ast}_{rs}} )  }/{\big(  h^{\rho^{\ast}_{rs}}(\lambda^-_{\rho^{\ast}_{rs}} ) - h^{\rho^{\ast}_{rs}}(\lambda^+_{\rho^{\ast}_{rs}} ) \big)}, \vspace{-2mm}
\end{equation}
with $\lambda^-_{\rho^{\ast}_{rs}} := \max_{j \in \mathbb{Z}} \left\{ \lambda_j\! : h^{\rho^{\ast}_{rs}}(\lambda_j ) > 0 \right\}$ and $\lambda^+_{\rho^{\ast}_{rs}} := \min_{j \in \mathbb{Z}} \left\{ \lambda_j\! : h^{\rho^{\ast}_{rs}}(\lambda_j ) < 0 \right\}$.
The estimators (\ref{crossInfoGamma}) and (\ref{crossInfoRho}) can be shown, under assumptions (C1) and (C3), to satisfy (C2) along the same lines as in Theorem 5.3 of \citet{ip2011}. \vspace{-2mm}

\subsection{Data-driven specification of reference density}\label{Sec:DataDriven}
While the choice of the reference density $f$ has no impact on the consistency properties of the corresponding $R$-estimator $\utLbf\n_f$, it has a direct influence on its performances for both finite $n$ and as $n \to \infty$;
the ``closer" $f$ is to the actual density $g$, the better the performance for $\utLbf\n_f$.
The efficiency loss due to a misspecified reference density $f$ is revealed though an inspection of the cross-information quantities.

Many mixing matrix estimators of ${\bf L}$, including those proposed by \citet{cb2006} and \citet{bj2003}, rely on nonparametric estimates of the underlying component densities or scores. 
However, such nonparametric estimates require large sample sizes to be effective, and are  sensitive to tuning parameters such as bandwidth or the choice of a basis functions. 
For instance, \citet{cb2006}  propose estimating score functions using a basis of $t$ $B$-spline functions; the exact choice of $t$ has a significant impact on the resulting estimator. 
Furthermore, nonparametric methods tend to be sensitive to outliers.

The purpose of using the $R$-estimators based on $f$-scores is precisely to increase robustness against outliers while avoiding nonparametric density estimation.  
A distinctive feature of ranks is that they are independent, under the null hypothesis and hence also under contiguous alternatives, of the corresponding order statistics. 
That property can be exploited, in the spirit of \citet{dj2000}, to select a reference density $f$ that accounts for features (skewness, kurtosis, etc.) of the actual underlying $g$: as long as such a selection is based on order statistics, it has no impact on the validity of $R$-estimation procedures.

We propose selecting $f := (f_1, \ldots, f_k)$ by fitting, componentwise,  a   parametric density to the (order statistic of the) residuals associated with the preliminary estimator $\tilde{\bf L}\n$.  
If skewness and kurtosis are to be accounted for, a convenient family of densities is  the family of skew-$t$ distribution (Azzalini and Capitanio,~2003) with densities of the form\vspace{-2mm}
\begin{equation}
\label{eq:skewt}
h_{\pmb{\omega} }( x ) = \frac{2}{\sigma} t_{\nu} (z) T_{\nu + 1}\Big( \alpha z \Big( \frac{\nu + 1}{ \nu + z^2 } \Big)^{1/2} \Big) \qquad \text{for~$x \in \mathbb{R}$~and~$z := \sigma^{-1}\left( x - \mu \right)$, }\vspace{-2mm}
\end{equation}
indexed by $\pmb{\omega} := (\mu, \sigma, \alpha, \nu)$, where $\mu \in \mathbb{R}$ is a location, $\sigma \in \mathbb{R}_0^+$ a scale, $\alpha \in \mathbb{R}$ a skewness parameter, and $\nu > 0$  the number of degrees of freedom governing the tails; $t_{\nu}(z)$ and~$T_{\nu}( z )$ 
 are the density and cumulative distribution functions, respectively,  of  the Student-$t$ distribution with $\nu$ degrees of freedom. 
For each $j=1, \ldots, k$, an estimator $(\hat{\mu}_j, \hat{\sigma}_j, \hat{\alpha}_j, \hat{\nu}_j)$    is obtained from the   residuals  $Z_{1,j}\n( \tilde{\bf L}\n), \ldots, Z_{n,j}\n (\tilde{\bf L}\n )$ using a routine maximum likelihood method.  
Then, the $f$-score functions used in the $R$-estimation procedure are those associated with the skew-$t$ density $h_{\hat{\pmb{\omega}}_j}$, with $\hat{\pmb{\omega}}_j = (\hat{\mu}_j, \hat{\sigma}_j, \hat{\alpha}_j, \hat{\nu}_j)$, thus taking  into account the skewness, kurtosis and tails of the residuals (for $\nu\in (4,\infty)$, those  kurtoses range between 3 and $\infty$). 

  Data-driven scores, however, clearly need not be restricted to the family of skew-$t$ densities, and can be selected from other univariate parametric families as well; in Section~\ref{Sec:Simulations}, we also consider, for instance, the family of stable distributions, indexed by $\pmb{\omega} := (\mu, \sigma, \beta, \gamma)$, where $\mu$ and $\sigma$ are location and scale, $\beta$ is a skewness parameter ($\beta=0$ means symmetry), and $\gamma\in(0,2]$ characterizes  tail behavior ($\gamma=2$ means Gaussian tails,~$\gamma =1$ Cauchy tails).  \vspace{-3mm}
  
\section{Simulations}\label{Sec:Simulations} \vspace{-1mm}
Simulation experiments are conducted to examine finite-sample performances of the proposed $R$-estimation procedure.
In the simulations, we evaluate $R$-estimators $\utLbf\n_f$ based on various  preliminary estimators from the literature and a data-driven reference density $f$, as described in Section~\ref{Sec:DataDriven}. 
In this section, we describe the precise construction of  the four preliminary estimators to be used,  the $R$-estimator $\utLbf\n_f$,  and, for the sake of comparison,  the $R_+$-estimator of Ilmonen and Paindaveine (2011). 
Then we describe the simulation experiment setups and conclude with a discussion of the simulation results. \vspace{-2mm}

\subsection{Preliminary, $R$-, and $R_+$-estimators}\label{Sec:EstimationMethod}
\subsubsection{The preliminary estimators}\label{Sec:PreliminaryEstimators}
 \citet{ose2006} propose estimating a mixing matrix using two distinct \emph{scatter matrices} with the \emph{independent components property}.
A \emph{scatter matrix} is a $k \times k$ symmetric positive definite and affine-equivariant   function of a sample of $n$ random $k$-vectors; it  is said to possess the \emph{independent components property} if, when the sample of random $k$-vectors 
  ${\bf X}_1\n, \ldots, {\bf X}_n\n$  
at which it is evaluated is i.i.d.\ with mutually independent components, all of its off-diagonal elements are $o_{\mathrm{P}}(1)$ as $n\to\infty$.  Examples include the sample covariance matrix \vspace{-2mm}
 $$\mathbf{S}_{\text{\tiny{COV}}} := \frac{1}{n} \sum_{i=1}^n \big( {\bf X}_i\n - \bar{\bf X}\n \big) \big( {\bf X}_i\n - \bar{\bf X}\n \big)\pr \qquad \text{where} \qquad \bar{\bf X}\n := \frac{1}{n} \sum_{i=1}^n {\bf X}_i\n;\vspace{-3mm} $$ 
and the fourth-order scatter matrix\vspace{-3mm}
$$
\mathbf{S}_{\text{\tiny{COV4}}} := \frac{1}{n} \sum_{i=1}^n \big(\mathbf{X}_i^{(n)} - \bar{\mathbf{X}}^{(n)} \big)^{\prime} \big( \mathbf{S}_{\text{\tiny{COV}}} \big)\inv \big(\mathbf{X}_i^{(n)} - \bar{\mathbf{X}}^{(n)} \big) \big(\mathbf{X}_i^{(n)} - \bar{\mathbf{X}}^{(n)} \big) \big(\mathbf{X}_i^{(n)} - \bar{\mathbf{X}}^{(n)} \big)^{\prime},
\vspace{-3mm}$$
leading to the popular FOBI estimator (Cardoso~(1989)).

Not all scatter matrices possess the independent components property---certainly in the presence of asymmetric densities.   
As a remedy, \citet{noo2008} propose using \emph{symmetrized} versions of the scatter matrices involved, which  entails their evaluation  at the~$n(n-1)/2$ distinct pairwise differences of the original observations, which  for large~$n$ is computationally heavy. 

The asymptotic properties and robustness of the estimator $\hat{\bf \Lambda} \big({\bf S}_{\text{\tiny{A}}}, {\bf S}_{\text{\tiny{B}}} \big)$ associated with the scatter matrices ${\bf S}_{\text{\tiny{A}}}$ and $ {\bf S}_{\text{\tiny{B}}}$ follow  from those of   ${\bf S}_{\text{\tiny{A}}}$ and ${\bf S}_{\text{\tiny{B}}}$ themselves (see Ilmonen et al.~(2012) for details).  
Since root-$n$  consistency of 
  ${\bf S}_{\text{\tiny{COV4}}}$ requires    finite  eight moments, so does FOBI. More robust estimates of scatter such as the van der Waerden rank-based estimator~$\mathbf{S}_{\text{\tiny{HOP}}}$ (Hallin et al.~2006) or Tyler's  estimator of shape  $\mathbf{S}_{\text{\tiny{Tyl}}}$ \citep{t1987} maintain root-$n$ consistency without any moment assumptions. Irrespective of moments, though,~$\hat{\bf \Lambda} \big({\bf S}_{\text{\tiny{A}}}, {\bf S}_{\text{\tiny{B}}} \big)$  loses consistency as soon as two component densities yield identical {\it ``generalized kurtoses"} (defined as the diagonal elements of ${\rm p}\! \lim _{n\to\infty}  {\bf S}_{\text{\tiny{B}}}^{-1} {\bf S}_{\text{\tiny{A}}}$).

 In the simulations below, we consider three preliminary estimators based on  the  two-scatter  method:  the {FOBI} estimator $\tilde{\bf \Lambda}_{\text{\tiny{Fobi}}}$, the estimator $\tilde{\bf \Lambda}_{\text{\tiny{HOPCov}}} := \hat{\bf \Lambda} \big({\bf S}^{\ast}_{\text{\tiny{HOP}}}, {\bf S}_{\text{\tiny{COV}}} \big)$ based on the symmetrized version ${\bf S}^{\ast}_{\text{\tiny{HOP}}}$ of  $\mathbf{S}_{\text{\tiny{HOP}}}$  and the sample covariance~${\bf S}_{\text{\tiny{COV}}}$, and~$\tilde{\bf \Lambda}_{\text{\tiny{TylHub}}} := \hat{\bf \Lambda} \big({\bf S}^{\ast}_{\text{\tiny{Tyl}}}, {\bf S}^{\ast}_{\text{\tiny{Hub}}} \big)$ based on  the symmetrized versions $ {\bf S}^{\ast}_{\text{\tiny{Tyl}}}$ and~$ {\bf S}^{\ast}_{\text{\tiny{Hub}}}$ of the Tyler estimator~$ {\bf S}_{\text{\tiny{Tyl}}}$   \citep{d1998} and Huber's~$M$-estimator of scatter (with   weights~$W(t)=\min(1,\chi^2_{k,0.9}/t^2)$, where~$\chi^2_{k,0.9}$ is the upper $0.1$ quantile of the $\chi^2_k$-distribution with $k$ degrees of freedom: see page 194 of Maronna et al.~(2006)).

The asymptotic properties of the {FastICA} estimator $\tilde{\bf \Lambda}_{\text{\tiny{FIca}}}$  \citep{ho1997} have been studied by \citet{o2010}, Nordhausen et al.~(2011), and Ilmonen et al.~(2012), who give sufficient conditions for  root-$n$ consistency.  
 In the simulations, we used the symmetric fastICA  R package by Marchini et al.~(2012)   with $\log \cosh$ scores and initial demixing matrix set to  identity.

Finally, the  {Kernel-ICA} algorithm \citep{bj2003} seeks a demixing matrix that minimizes the mutual information between 
  independent components via a  \emph{generalized variance}, a construction implicitly measuring  non-Gaussianity. 
Of all preliminary estimators we considered, $\tilde{\bf \Lambda}_{\text{\tiny{KIca}}}$ (computed from the kernel-ica Matlab package \citep{kicaCode} with default~settings) yields the best performances in the simulations;  its asymptotic properties so far have not been well studied, though, and, to the best of our knowledge,   root-$n$ consistency conditions have not been obtained yet.

After evaluating each preliminary estimator  ($\tilde{\bf \Lambda}_{\text{\tiny{PE}}} $,   for {PE} = {Fobi}, {HOPCov}, {TylHub}, {FIca}, and {KIca})
from each replication, one-step $R$-estimators are computed from the observationally equivalent\vspace{-2mm}
\begin{equation}
\label{eq:PreliminaryEstimators}
\tilde{\bf L}_{\text{\tiny{PE}}}\! := \Pi ( \tilde{\bf \Lambda}_{\text{\tiny{PE}}})  
\vspace{-2mm}
\end{equation}
which belong to  $\mathcal{M}_k^1$  (see~(\ref{eq:PiMapping}) for the definition of the mapping $\Pi$).

\subsubsection{The $R$-estimators}\label{Sec:ImplementationREstimator}
As described in Section \ref{Sec:DataDriven}, we used data-driven scores from the skew-$t$ family in the construction of our $R$-estimators. 
For each replication of  ${\bf X}_1\n, \ldots, {\bf X}_n\n$ and preliminary estimator~$\tilde{\mathbf{L}} \in \mathcal{M}_k^1$,  we compute the  residuals     $\hat{\bf Z}_i\n \big(  \tilde{\bf L}  \big) := \tilde{\bf L}\inv {\bf X}_i\n$ for $i=1, \ldots, n$. For  \linebreak each~$j=1,\ldots, k$, a skew-$t$ density $h_{\,\hat{\pmb{\omega}}_j}$ (see (\ref{eq:skewt}))  is fit to the $n$-tuple~$\hat{ Z}_{1,j}\n \big(  \tilde{\bf L}  \big), \ldots, \hat{ Z}_{n,j}\n \big(  \tilde{\bf L}  \big)$     of~$j$th components via maximum likelihood (MLE). For numerical stability reasons, the estimator $\omega _j$  was limited to the interval $\hat{\alpha}_j \in [-15,15]$ and~$\hat{\nu}_j \in [3,\infty)$.

The resulting one-step $R$-estimate then is, with $f := ( h_{\hat{\pmb{\omega}}_1}, \ldots, h_{\hat{\pmb{\omega}}_k} )$, \vspace{-2mm}
\begin{equation}
\label{eq:OneStepRDataDriven}
\utLbf^{\ast}\big( \tilde{\bf L} \big) := \tilde{\bf L} + n^{-\half}~\tilde{\bf L} \Big[ \hat{\mathbf{N}}_{\tilde{\bf L}, f }\n - \text{diag}\Big( \tilde{\bf L} \ \hat{\mathbf{N}}_{\tilde{\bf L}, f}\n \Big) \Big],\vspace{-3mm}
\end{equation} 
where $\hat{\mathbf{N}}_{\tilde{\bf L}, f}\n $ is defined in (\ref{eq:onestepv2-intermediary}) 
(because $\utLbf^{\ast}\big( \tilde{\bf L} \big)$ is based on data-driven scores, no reference density is used in the notation).

In the simulations, we also explore the performance of a multistep version of the same $R$-estimator. Taking  $\utLbf^{\ast}\big( \tilde{\bf L} \big)$ as a preliminary, (\ref{eq:OneStepRDataDriven})  indeed is easily iterated, letting   \vspace{-3mm}
\begin{equation}
\label{eq:MultiStepRDataDriven}
\utLbf^{\ast}_{\text{\tiny{(t)}}}\big( \tilde{\bf L} \big) :=  \utLbf^{\ast}\Big( \utLbf^{\ast}_{\text{\tiny{(t-1)}}}\big( \tilde{\bf L} \big)  \Big)\quad\text{with }\ \ \utLbf^{\ast}_{(0)} \big( \tilde{\bf L} \big) := \tilde{\bf L},\quad t=1,\ldots,T.\vspace{-2mm}
\end{equation} 
\subsubsection{The Ilmonen-Paindaveine $R_+$-estimators}\label{Sec:SignedRankEstimator}
We also computed the   Ilmonen and Paindaveine~(2011) signed-rank $R_+$-estimators, the validity of which only holds under  symmetric component densities. 
 This   not only requires a root-$n$ consistent preliminary estimator $\tilde{\bf L}\n \in \mathcal{M}_k^1$, but also an estimate for the location~$\pmb{\mu} \in \mathbb{R}^k$. The preliminary estimators we used are those described in Section~\ref{Sec:EstimationMethod}; for location, we adopted the same componentwise median estimator  as in Ilmonen and Paindaveine~(2011). To make the comparison a fair one, however, we also implemented the signed-rank procedure on the basis of data-driven scores, as explained in Section~\ref{Sec:ImplementationREstimator}---restricting  the fit, of course, to symmetric Student or stable densities. The resulting $R_+$-estimators are denoted as $\utLbf^{\ast}_+\big( \tilde{\bf L} \big)$.  Finally, parallel to (\ref{eq:MultiStepRDataDriven}), multistep versions of $\utLbf^{\ast}_+\big( \tilde{\bf L} \big)$ are easily constructed; the notation $\utLbf^{\ast}_{+\text{\tiny{(t)}}}\big( \tilde{\bf L} \big)$ is used in an obvious way.

\subsection{Simulation settings}\label{Sec:SimulationExperiments}
In each simulation experiment, 
 three-dimensional  observations ($k$=3) 
were generated from various generating processes. Each 
generating process
 is characterized by a sample size~$n$ and a triple $g^{\text{\tiny{(S)}}} := (g_1^{\text{\tiny{(S)}}}, g_2^{\text{\tiny{(S)}}}, g_3^{\text{\tiny{(S)}}})\vspace{-0mm}$ of  component  densities, labeled  $(S)=(A),\ldots , (I)$,  
the list of which is provided  in Table~\ref{table:Setup}. Those densities present   various  skewness levels and tail behaviors, with  (A,~C) skew-$t$ and   stable, but also (B)  skew Laplace densities, of the form  (for location $\mu \in {\mathbb R}$, scale $\sigma > 0$, and shape parameter $\eta \in (0,\infty)$, where $\eta =0$ yields symmetry
$$ h(z; \mu, \sigma, \eta) =\left\{\begin{array}{lc}  \frac{\sqrt{2}}{\sigma}\frac{\eta}{1 + \eta^2} \exp \big( -\frac{\sqrt{2}}{\sigma \eta} | z - \mu| \big) \quad &z \leq \mu\\ 
 \frac{\sqrt{2}}{\sigma}\frac{\eta}{1 + \eta^2} \exp \big( - \frac{\sqrt{2}}{\sigma} \eta | z - \mu| \big) \quad &z > \mu 
\end{array}
\right.
$$
see Kotz et al.~(2001)). 
 We also considered variations of component distributions in (D,~E)  with an asymmetric bimodal mixture distribution ({mix-}${ t}_3$) included in (E), and, for the purpose of a comparison of  $R$- and $R_+$-estimators, two   symmetric triples (F,~G).  Finally,  for the sake of a very rough  robustness investigation, two contaminated settings (H,~I) were  included in the study. There, at each replication, three component densities were selected at random (uniformly, without replacement)  from  skew $t_5\,(\alpha=4)$, skew~Laplace $(\eta=2)$,~$\cal N$,  stable $(\beta=1, \gamma=1.5)$, and {mix-}${ t}_3$; the resulting observation then, with probability 2\% (H) or 5\% (I)  is multiplied by a factor drawn from a uniform distribution over $[-5, 5]$.

\begin{table}[h] 
\centering
\begin{tabular}{clll} 
\toprule
\midrule 
 & \multicolumn{3}{c}{Component densities } \\ \cmidrule{2-4}
($S$)  & \multicolumn{1}{c}{$g_1^{\text{\tiny{(S)}}}$} & \multicolumn{1}{c}{$g_2^{\text{\tiny{(S)}}}$}  & \multicolumn{1}{c}{$g_3^{\text{\tiny{(S)}}}$} \\ 
\midrule
(A) & skew-$t_5\,(\alpha=-4)$		& skew-$t_5\,(\alpha=2)$			& Student $t_{5}$					\\ [1ex] 
(B) & skew Laplace$(\eta=2)$			& skew Laplace$(\eta=1/3)$		& Laplace								\\ [1ex] 	
(C) & stable$(\beta=-1, \gamma=1.5)$	& stable$(\beta=1, \gamma=1.5)$	& stable$(\beta=0, \gamma=1.5)$	\\ [1ex] 
& & & \\ [-2ex]
(D) & skew $t_5\,(\alpha=-4)$		& skew Laplace$(\eta=2)$			& stable$(\beta=-1, \gamma=1.5)$	\\ [1ex] 
(E) & stable$(\beta=-1, \gamma=1.5)$	& skew Laplace($\eta=2$)			& {mix-}${ t}_3$					\\ [1ex]
(F) & Cauchy $t_1$						& Student $t_2$						& Student $t_3$	 					\\ [1ex] 	
(G) & Student $t_{3}$					& Student $t_{5}$					& Normal								\\ [1.5ex] 
\midrule
(H)& contaminated data& (2\% contamination)&\\
(I) & contaminated data& (5\% contamination)&\vspace{2mm}\\
\hline
\end{tabular}
\captionsetup{font=small}
\caption{Component densities used in the simulation experiment, all with median zero and unit scale: (a)   skew-$t_{\nu}$ densities with shape (asymmetry)  parameter $\alpha$ and $\nu$ degrees of freedom;  (b)~{skew Laplace}  densities with shape parameter $\eta$;  (c) {stable}  densities  with skewness parameter $\beta$ and tail index $\gamma$; (d)  a mixture of two non-standardized Student densities  with~$3$ degrees of freedom; and (e)   the classical Student~$t_\nu$ ($\nu$ degrees of freedom) and Laplace (double exponential) densities.  In setups (H) and (I), the source component densities are contaminated, with contamination rate 2\% and 5\%, respectively, according to the mechanism explained above. 
\label{table:Setup} } \vspace{-3mm}
\end{table}

\vspace{-2mm}

\doublespacing

\normalsize

 Each marginal distribution in (A-G) has median  zero  and unit scale (recall that location and scale  here play no role).  
 Throughout, the same $3 \times 3$ mixing matrix \vspace{-2mm}
$$ {\bf L} := \left(\begin{array}{ccc}1&0.5&0.5 \\ 0.5&1&0.5 \\0.5&0.5&1\end{array}\right) \in \mathcal{M}_3^1 \vspace{-2mm}$$ was used.  Small ($n=100$) and moderate ($n=1,000$) sample sizes were considered. \vspace{-2mm}

For each generating process (each combination of $n=100$ or $1,000$ and 
 $(S)\in\{(A),\ldots , (I)\}$), the number of replications was set to $M=1,000$, and, for each replication,  the following estimators of $\bf L$ were computed:\vspace{-2mm}
\begin{enumerate}
\item[(a)]the preliminary estimators $\tilde{\bf L}=\tilde{\bf L}_{\text{\tiny{Fobi}}}$,  $\tilde{\bf L}_{\text{\tiny{HOPCov}}}$, $\tilde{\bf L}_{\text{\tiny{FIca}}}$, and $\tilde{\bf L}_{\text{\tiny{KIca}}}$ given in (\ref{eq:PreliminaryEstimators});\vspace{-1mm}
\item[(b)]the one-step $R$-estimators $\utLbf^{\ast}\big( \tilde{\bf L} \big)$ based on the   preliminary ones as listed under (a) and   data-driven skew-$t$ scores;\vspace{-2mm} 
\item[(c)]the one-step $R_+$-estimators $\utLbf^{\ast}_+\big( \tilde{\bf L} \big)$ based on the   preliminary ones as listed under (a) and data-driven  Student $t$ scores.\vspace{-2mm}
\end{enumerate}
For component densities 
 (B, D, F),  we moreover computed, for $n=100$  and~$1,000$,   \vspace{-2mm}
\begin{enumerate}
\item[(d)]the $T$-multistep versions of the $R$-estimators based on the   preliminary $\tilde{\bf L}_{\text{\tiny{HOPCov}}}$ and~$\tilde{\bf L}_{\text{\tiny{KIca}}}$, still with  data-driven skew-$t$ scores,  $T=1,\ldots , 5$.\vspace{-2mm}
\end{enumerate}

Many performance indices have been proposed in the literature to compare the performances of various  ICA algorithms in  simulation studies: see \citet{mm1994}, Theis et al.~(2004), \citet{d2007}, \citet{o2010}, Ilmonen et al.~(2010). The most popular one remains the so-called  \emph{Amari error} (Amari et al.~1996), which we are using here. We also considered (see  the {\it supplemental material section} for  additional tables) the  minimum distance index recently proposed by   Ilmonen et al.~(2010) which, however, essentially leads to the same conclusions. 

 The Amari error  $\text{AE}( {\bf A}, {\bf B})$ of a $k \times k$ matrix $\bf{A}$ with respect to a nonsingular $k \times k$ matrix ${\bf B}$  (it is not a matrix norm)  is defined as  \vspace{-2mm}  
\be
\label{eq:AmariError}
\text{AE}( {\bf A}, {\bf B}) = \frac{1}{2k(k-1)} \bigg( \sum_{i=1}^k \bigg( \frac{ \sum_{j=1}^k | w_{ij} |}{ \max_j | w_{ij} | } - 1\bigg) +  \sum_{j=1}^k \bigg( \frac{ \sum_{i=1}^k | w_{ij} |}{ \max_i | w_{ij} | } - 1\bigg) \bigg),\vspace{-2mm}
\en 
with ${\bf W} := {\bf B}\inv {\bf A} = \left[ w_{ij} \right]$, and takes values between $0$ and $1$,  with $\text{AE}( {\bf A}, {\bf B})$ close to~$0$ indicating higher similarity between ${\bf A}$ and ${\bf B}$ and~$\text{AE}( {\bf A}, {\bf B}) = 0$ when $ {\bf B}\inv {\bf A} = {\bf P}{\bf D}$ for some~$k \times k$ permutation matrix~${\bf P}$ and~$k \times k$ diagonal matrix~${\bf D}$, so that  $\text{AE}( {\bf A}, {\bf B}) = 0$ for   observationally equivalent mixing matrices  (such that $\Pi ({\bf A})= \Pi ({\bf B})$).  
When computed   from matrices in ${\cal M}^1_k$ (where moreover row and column permutations/rescaling are irrelevant), the Amari error  thus defines a natural measure of performance.  

Figures~\ref{fig:SimGroup1}-\ref{fig:SimGroup4} below are providing boxplots for the $M=1,000$ Amari distances associated with the various  simulation setups. Since Amari distances are intrinsically nonnegative, their extreme values  are in the right tail only, and  the relevant boxplots  (showing the first quartile, the median, the third quartile, and  a~0.95 quantile whisker) are ``one-sided".  
Figures~\ref{fig:SimGroup1}-\ref{fig:SimGroup3} are dealing with component densities~(A,~B,~C), (D,~E,~F), and (G,~H,~I), 
respectively. Figure \ref{fig:SimGroup4} shows the results for the $T$-step versions of the $R$-estimators based on~$\tilde{\bf L}_{\text{\tiny{TylHub}}}$  and~$\tilde{\bf L}_{\text{\tiny{KIca}}}$, under components densities~(B,~D,~H), as described in (d) above.

 Inspection of Figures   \ref{fig:SimGroup1}-\ref{fig:SimGroup3}  reveals that Kernel-ICA is, almost uniformly, and sometimes quite substantially, the best preliminary  under asymmetric  (A-E) or symmetric (F,~G)  setups when no contamination is involved. 
Combined with $R$-estimation (data-driven skew-$t$ scores), they yield the typical winners in these setups, even though, in principle, $R_+$-estimators should do better under symmetric densities. 
The best performances of $R$-estimators seem to take place under heavy tails (Cauchy and stable component densities)---thanks, probably, to the data-driven  selection of scores. 
Note that partial symmetry in setups~(A,~B,~C) does not really help $R_+$-estimation much. 

Setups~(H) and~(I) in Figure \ref{fig:SimGroup3} show evidence that contamination can dismantle the correlation structure of the model, with the performances of each preliminary, including Kernel-ICA,  deteriorating dramatically. The two-scatter  preliminary   constructed from the robust Tyler and Huber   estimators, though, resist better than the rest. 
The  $R$-estimators quite  significantly enhance  each preliminary when $n=1000$, and still improve them, albeit less noticeably, when $n=100$, thus partially compensating the impact of contamination on the preliminary estimators.  
Unsurprisingly, increasing the  contamination level  from $2\%$ (H) to $5\%$ (I) deteriorates the quality of the preliminaries and the $R$-estimators based on them---however, $R$-estimation still provides striking gains   when $n=1000$.

\begin{landscape}
\begin{figure}
\begin{subfigure}[h]{1\linewidth}
\captionsetup{font=footnotesize}
\caption{Sample size $n=100$}
\begin{center}
\includegraphics[scale=.44]{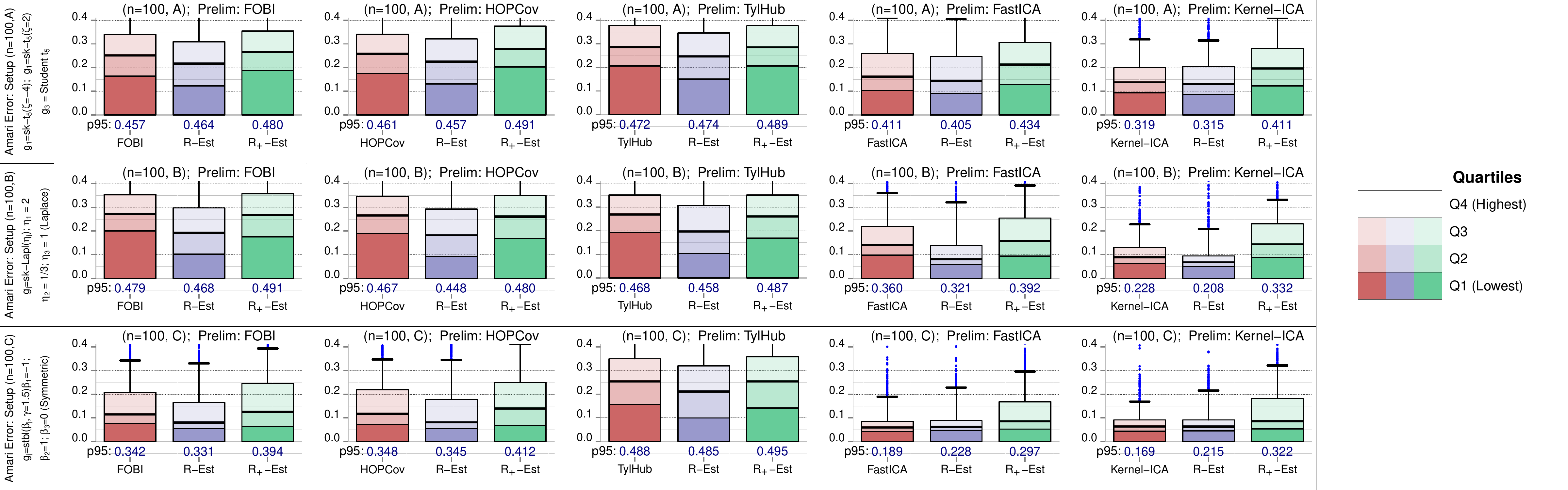}
\end{center}
\label{fig:SimGroup1-n100}
\end{subfigure}%
\\
\begin{subfigure}[h]{1\linewidth}
\captionsetup{font=footnotesize}
\caption{Sample size $n=1,000$}
\begin{center}
\includegraphics[scale=.44]{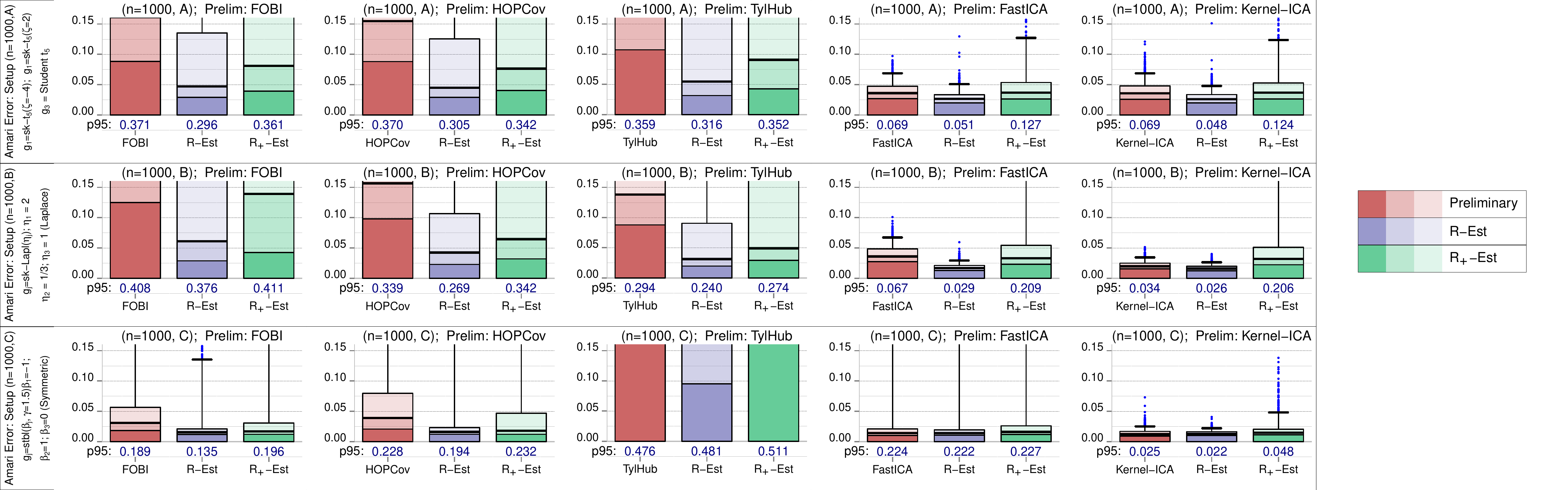}
\end{center}
\label{fig:SimGroup1-n1000}
\end{subfigure}%
\captionsetup{font=small}
\caption{
Boxplots of Amari errors obtained in $M=1,000$ replications of the setup 
 $(n, \text{S})$, $n~\! =~\! 100,\ 1,000$, $\text{S}=A,\, B,\, C,\vspace{1mm}$ for the preliminary 
$\tilde{\bf L}=\tilde{\bf L}_{\text{\tiny{Fobi}}}$,  
$\tilde{\bf L}_{\text{\tiny{HOPCov}}}$, 
$\tilde{\bf L}_{\text{\tiny{TylHub}}}$, 
$\tilde{\bf L}_{\text{\tiny{FIca}}}$, 
$\tilde{\bf L}_{\text{\tiny{KIca}}}$, the one-step $R$-estimator 
$\utLbfCaption^{\ast} ( \tilde{\bf L} )$, and  the one-step $R_+$-estimator 
$\utLbfCaption^{\ast}_+( \tilde{\bf L}) $  based on  the same preliminaries,\vspace{-1mm} with data-driven skew-$t$  and Student-$t$ scores, respectively. 
}\label{fig:SimGroup1}
\end{figure}


\begin{figure}
\begin{subfigure}[h]{1\linewidth}
\captionsetup{font=footnotesize}
\caption{Sample size $n=100$}
\begin{center}
\includegraphics[scale=0.44]{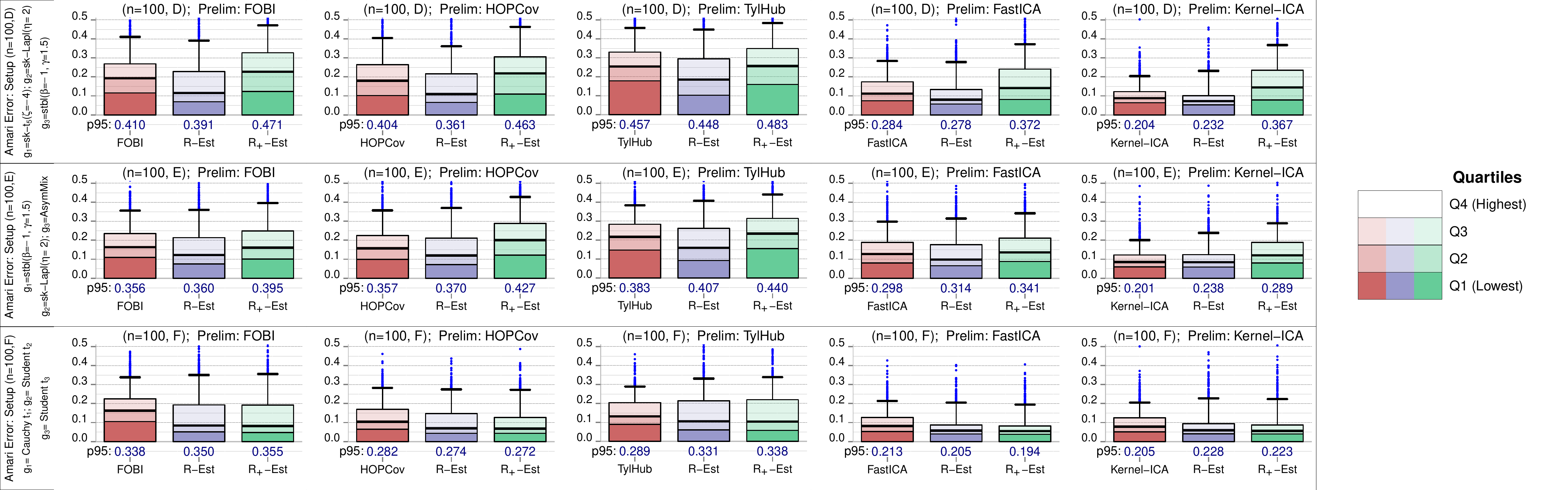}
\end{center}
\label{fig:SimGroup2-n100}
\end{subfigure}%
\\ \\
\begin{subfigure}[h]{1\linewidth}
\captionsetup{font=footnotesize}
\caption{Sample size $n=1,000$}
\begin{center}
\includegraphics[scale=0.44]{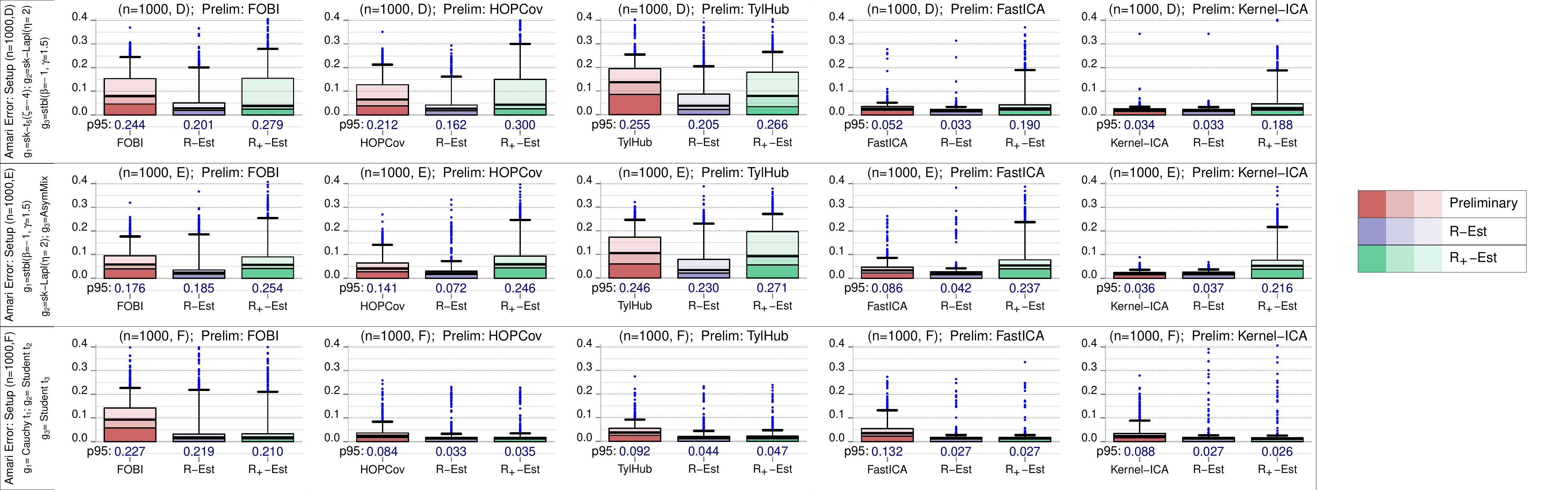}
\end{center}
\label{fig:SimGroup2-n1000}
\end{subfigure}%
\captionsetup{font=small}
\caption{
Boxplots of Amari errors obtained in $M=1,000$ replications of the setup $(n, \text{S})$, $n~\!\! =~\! \!100,\ 1,000$, $\text{S}=D,\, E,\, F,\vspace{1mm}$ for the 
    preliminary $\tilde{\bf L}=\tilde{\bf L}_{\text{\tiny{Fobi}}}$,  $\tilde{\bf L}_{\text{\tiny{HOPCov}}}$, 
    $\tilde{\bf L}_{\text{\tiny{TylHub}}}$, 
    $\tilde{\bf L}_{\text{\tiny{FIca}}}$, $\tilde{\bf L}_{\text{\tiny{KIca}}}$, the one-step $R$-estimator $\utLbfCaption^{\ast} ( \tilde{\bf L} )$, and  the one-step $R_+$-estimator  $\utLbfCaption^{\ast}_+( \tilde{\bf L}) \vspace{-1mm}$  based on  the same preliminaries, with data-driven skew-$t$  and Student-$t$ scores, respectively.} 
\label{fig:SimGroup2}
\end{figure}\vspace{-6mm}


\begin{figure}
\begin{subfigure}[h]{1\linewidth}
\captionsetup{font=footnotesize}
\caption{Sample size $n=100$}
\begin{center}
\includegraphics[scale=0.44]{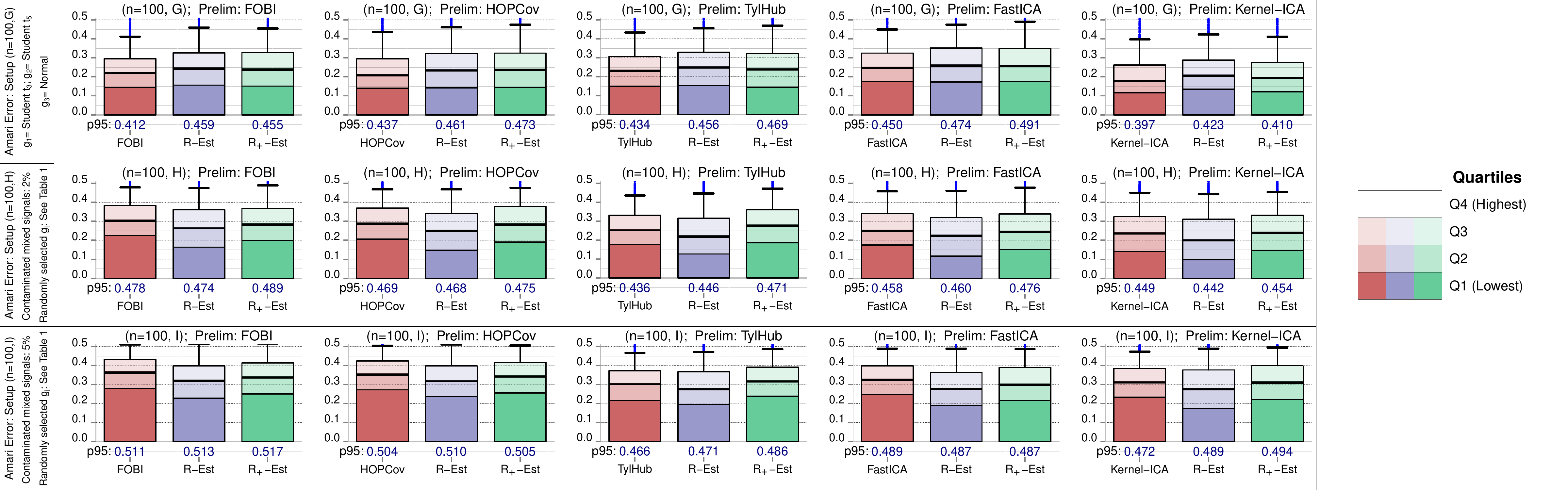}
\end{center}
\label{fig:SimGroup3-n100}
\end{subfigure}%
\\ \\
\begin{subfigure}[h]{1\linewidth}
\captionsetup{font=footnotesize}
\caption{Sample size $n=1,000$}
\begin{center}
\includegraphics[scale=0.44]{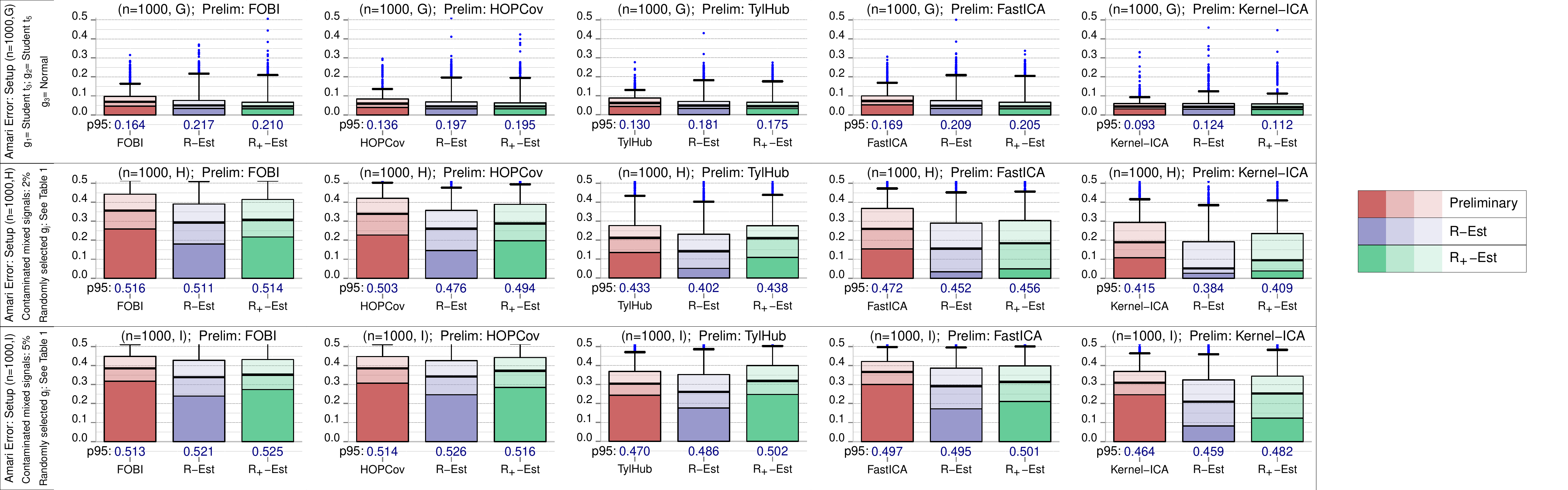}
\end{center}
\label{fig:SimGroup3-n1000}
\end{subfigure}%
\captionsetup{font=small}
\caption{Boxplots of Amari errors  
  obtained in $M=1,000$ replications of the setup $(n, \text{S})$, $n~\!\! =~\! \!100,\ 1,000$, $\text{S}=G,\, H,\, I,\vspace{1mm}$     for the 
    preliminary $\tilde{\bf L}=\tilde{\bf L}_{\text{\tiny{Fobi}}}$,  $\tilde{\bf L}_{\text{\tiny{HOPCov}}}$,
    $\tilde{\bf L}_{\text{\tiny{TylHub}}}$, 
     $\tilde{\bf L}_{\text{\tiny{FIca}}}$, $\tilde{\bf L}_{\text{\tiny{KIca}}}$, the one-step $R$-estimator $\utLbfCaption^{\ast} ( \tilde{\bf L} )$, and  the one-step $R_+$-estimator $\utLbfCaption^{\ast}_+( \tilde{\bf L}) \vspace{-1mm}$     based on  the same preliminaries, with data-driven skew-$t$ and Student-$t$ scores, respectively.} 
 \label{fig:SimGroup3}
   \end{figure}\vspace{-4mm}
  
  \begin{figure}
\begin{subfigure}[h]{0.5\linewidth}
\captionsetup{font=footnotesize}
\caption{Sample size $n=100$}
\begin{center}
\includegraphics[scale=.58]{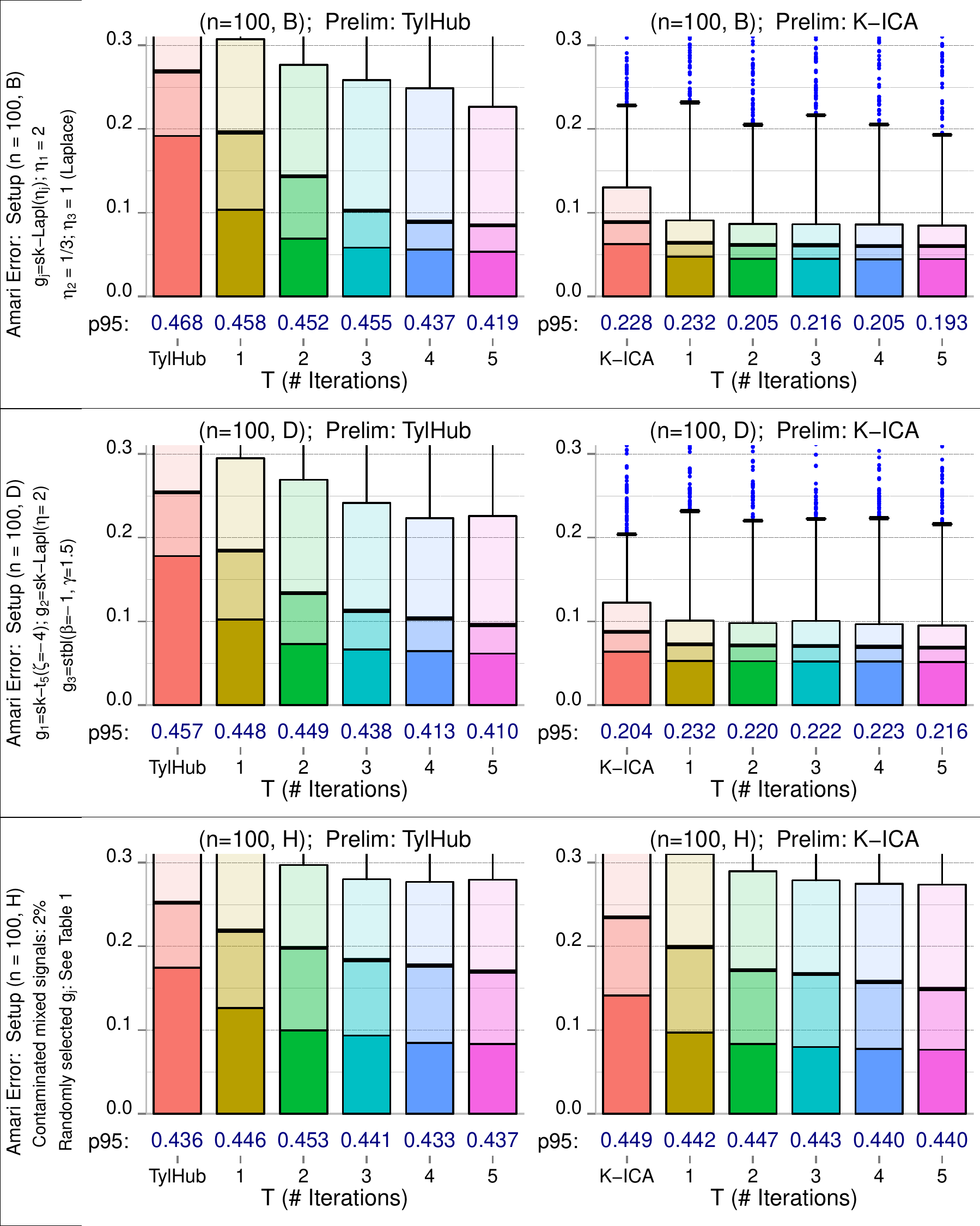}
\end{center}
\label{fig:SimGroup4-n100}
\end{subfigure}%
\begin{subfigure}[h]{0.5\linewidth}
\captionsetup{font=footnotesize}
\caption{Sample size $n=1,000$}
\begin{center}
\includegraphics[scale=.58]{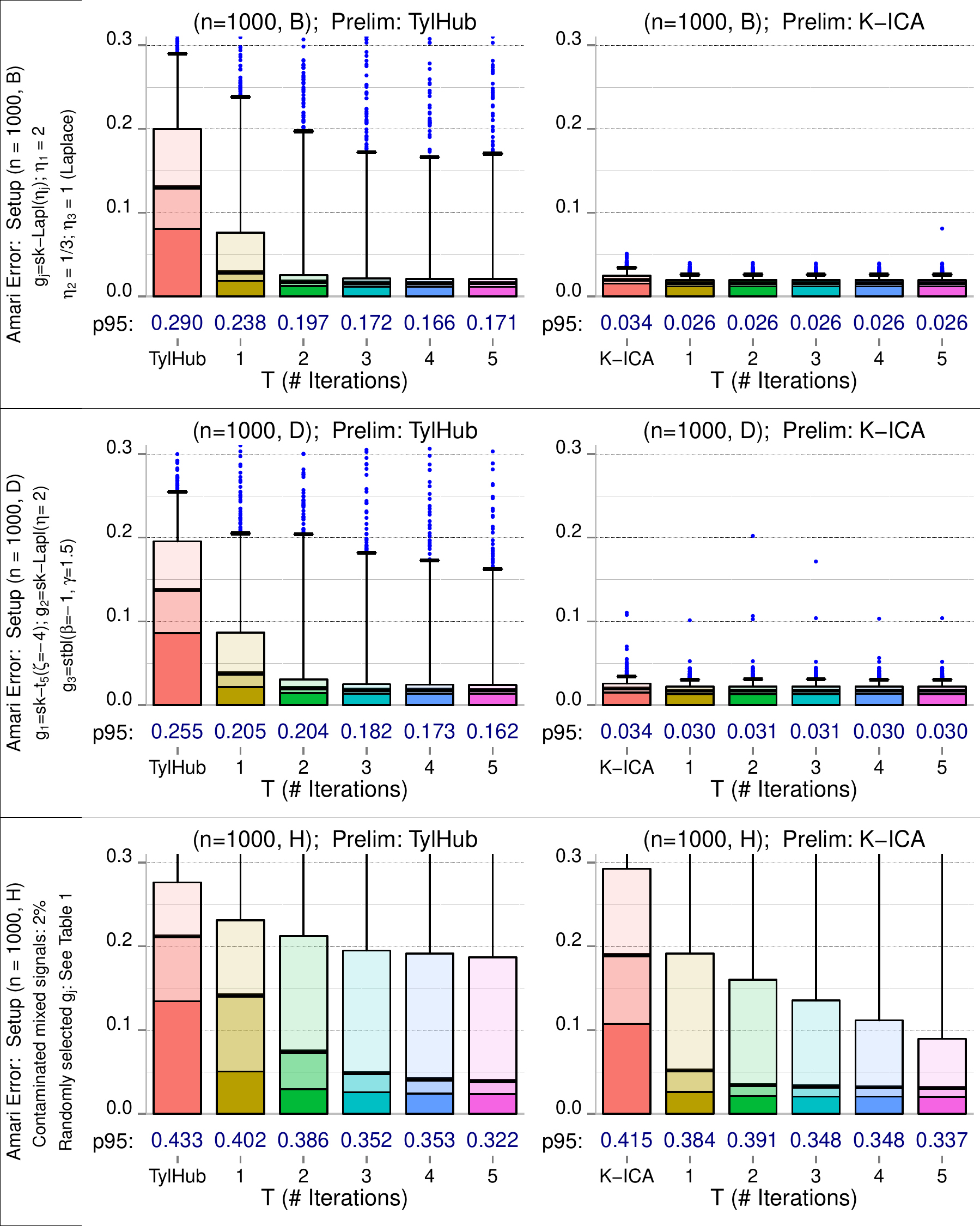}
\end{center}
\label{fig:SimGroup4-n1000}
\end{subfigure}%
\captionsetup{font=small}
\caption{
Boxplots of Amari errors 
  obtained in $M=1,000$ replications of the setup $(n, \text{S})$, $n=100, 1000$, $\text{S}=B,\, D,\, H,\vspace{1mm}$ 
   for the $T$-step $R$-estimator $\utLbfCaption^{\ast} ( \tilde{\bf L} )$   based on  preliminary  
      $\tilde{\bf L}=  \tilde{\bf L}_{\text{\tiny{TylHub}}}$ and     $\tilde{\bf L}_{\text{\tiny{KIca}}}\vspace{-3mm}$, respectively,
and  data-driven skew-$t$ scores, $T=1,\ldots, 10$}\label{fig:SimGroup4}
\end{figure}

\end{landscape}

Finally, Figure~\ref{fig:SimGroup4} shows how iterating the rank-based correction can improve a poor preliminary. The Tyler-Huber two-scatter  estimator is typically outperformed by the Kernel-ICA one,  except in setup~(H) where  contamination   leads to a drastic deterioration of Kernel-ICA. 
 After a few iterations, both the  Tyler-Huber- and Kernel-ICA-based $R$-estimators  perform quite similarly; the latter needs less iterations, though,   to reach its best performance in setups~(B) and~(D). For $n=1,000$, starting from Kernel-ICA in either of those setups, one step is essentially sufficient.
However, $R$-estimators based on either preliminary improve considerably over multiple iterations in setup~(F) with contaminated mixed samples.

\section{An application 
in  image analysis}\label{Sec:Application}
The objective of ICA in applications is typically to recover source signals from a sequence of observed mixed signals. As such, they are widely used in a variety of contexts where the fundamental assumptions (\ref{ICmod})-(\ref{eq:indf}) of ICA are unlikely to hold.  One of the merits of existing ICA such as {FastICA} and {Kernel-ICA} is that they resist reasonably well to such theoretically unwarranted applications. Such statements, of course, remain  unavoidably vague: in the absence of a formal model, indeed, pertinent benchmarks for performance evaluation  are hard to define.  Demixing acoustic signals or images, where ``readability" of the final result appears as an obvious criterion, are an exception.  Therefore, in this section, we apply various  ICA estimation methods, including the   rank-based ones, 
  to the demixing of      images that clearly do not satisfy the assumptions we have been making throughout this paper.  
The results are shown in  Figure~\ref{fig:EstimatedImages}. Their quality is best evaluated by eye-inspection, but a quantitative assessment can be made via the Amari distances provided in Table~\ref{fig:MultiStepErrors}a and b. Although traditional ICA techniques provide reasonable results, our rank-based techniques appear to bring quite significant improvements.  

A black-and-white digital image with resolution~$h \times w$ ($h,w \in \mathbb{N}$) can be represented by a \emph{pixel matrix}~${\bf Z} = (Z_{rs}) \in [0, 1]^{h \times w}$, where  
 $Z_{rs}$   represents the ``greyness" of the pixel located in the $r$th row and $s$th column; 
 if $Z_{rs} = 0 $, the pixel is pure black, and if $Z_{rs} = 1$, the pixel is pure white. In this example, we mix three source images of US currency notes, represented by the pixel matrices~${\bf Z}_j = (Z_{j;rs})$, $j=1,2,3$ ($h:=65$ and $w := 150$). These three source images are turned into three mixed ones, with  pixel matrices  ${\bf X}_j = (X_{j;rs})$, $j=1,2,3$, where~$(X_{1;rs},X_{2;s},X_{3;rs})\pr = {\bf L}^{\star}(Z_{1;rs},Z_{2;s},Z_{3;rs})\pr$, with ${\bf L}^{\star}={\bf I}_3 + 0.95({\pmb 1}_3-{\bf I}_3)\in \mathcal{M}_3^1$  (denoting by ${\pmb 1}_3$   a $3\times 3$ matrix of ones); ${\bf L}^{\star}$ thus has a diagonal of ones, all off-diagonal enties being~0.95. The source and mixed images  are displayed in Figure \ref{fig:ImageMixed}. 
 
 \begin{figure}[htbp]
\small
 \begin{subfigure}[b]{1\textwidth}
\captionsetup{font=footnotesize}
\caption{Top row:  the three source images. Bottom row:  the three mixed images.  \label{fig:ImageMixed} }
\begin{center}
\includegraphics[width=4.5in]{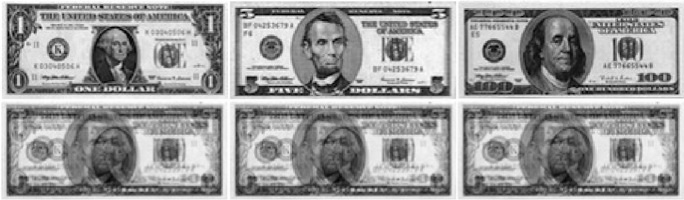}
\end{center}
\end{subfigure}%
\\ \\
\begin{subfigure}[b]{1\textwidth}
\captionsetup{font=footnotesize}
\caption{{FOBI} preliminary. Top row:  the $\tilde{\bf L}_{\text{\tiny{Fobi}}}$-demixed images.  Bottom row:   the $\utLbfCaption^{\ast}_{(20)}( \tilde{\bf L}_{\text{\tiny{Fobi}}} )$-demixed~images.\label{fig:FOBI-Image} }
\begin{center}
\includegraphics[width=4.5in]{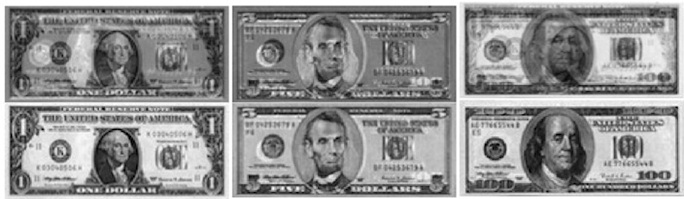}
\end{center}
\end{subfigure}%
\\ \\
\begin{subfigure}[b]{1\textwidth}
\captionsetup{font=footnotesize}
\caption{{FastICA} preliminary.  Top row:  the $\tilde{\bf L}_{\text{\tiny{FIca}}}$-demixed images. Bottom row:  the $\utLbfCaption^{\ast}_{(20)}( \tilde{\bf L}_{\text{\tiny{FIca}}} )$-demixed images. \label{fig:FICA-Image} }
\begin{center}
\includegraphics[width=4.5in]{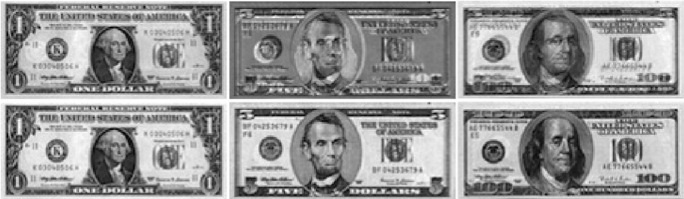}
\end{center}
\end{subfigure}%
\\ \\
\begin{subfigure}[b]{1\textwidth}
\captionsetup{font=footnotesize}
\caption{{Kernel-ICA} preliminary.  Top row:  the $\tilde{\bf L}_{\text{\tiny{KIca}}}$-demixed images. Middle row:  the $\utLbfCaption^{\ast}_{(20)}( \tilde{\bf L}_{\text{\tiny{KIca}}} )\vspace{-2mm}$-demixed~images. Bottom row:  the $\utLbfCaption^{\ast}_{+(20)}( \tilde{\bf L}_{\text{\tiny{KIca}}} )$-demixed~images. \label{fig:KICA-Image} }
\begin{center}
\includegraphics[width=4.5in]{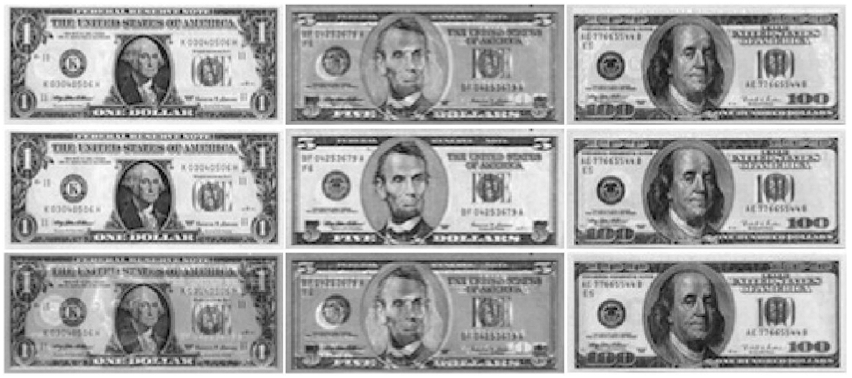}
\end{center}
\end{subfigure}%
\captionsetup{font=small}\caption{   Figure \ref{fig:ImageMixed} contains the three source images 
 and the  three mixed ones. Figures~\ref{fig:FOBI-Image}, \ref{fig:FICA-Image}, and \ref{fig:KICA-Image} show the demixed images obtained from multistep data-driven skew-$t$ score $R$-estimators, based on FOBI, FastICA, and Kernel-ICA   preliminaries, respectively. In Figure~\ref{fig:KICA-Image}, the result of a Kernel-ICA-based, data-driven Student-$t$ score multistep $R_+$-estimator method are also provided.}\label{fig:EstimatedImages}
\end{figure}

 We then performed ICA  estimation on the $n=65\times 150=9,750$ three-dimensional observations $(X_{1;rs},X_{2;s},X_{3;rs})$ by computing the multistep $R$-estimators $\utLbf^{\ast}_{(T)}( \tilde{\mathbf{L}} )$  with data-driven skew-$t$ scores  (\ref{eq:MultiStepRDataDriven}) and  preliminary estimators  $\tilde{\mathbf{L}} = \tilde{\bf L}_{\text{\tiny{Fobi}}}$, $\tilde{\bf L}_{\text{\tiny{FIca}}}$, and $\tilde{\bf L}_{\text{\tiny{KIca}}}$ as described  in~(\ref{eq:PreliminaryEstimators}), and $T=1,\ldots,20$;   the $\tilde{\bf L}_{\text{\tiny{HOPCov}}}$ and $\tilde{\bf L}_{\text{\tiny{TylHub}}}$ preliminaries were omitted because symmetrizing the \emph{HOP} and {Tyler} scatter matrices (about $10^8$  pairwise differences) was computationally too heavy.  Figures~\ref{fig:FOBI-Image}, \ref{fig:FICA-Image}, and \ref{fig:KICA-Image} contain  the resulting $\tilde{\bf L}$-   and  $\utLbf^{\ast}_{(20)}( \tilde{\mathbf{L}} )$-demixed images. Of all preliminary estimators considered, $\tilde{\bf L}_{\text{\tiny{KIca}}}$ seems to provide the best results. In Figure \ref{fig:KICA-Image}, we therefore also provide the demixed images resulting from the Ilmonen and Paindaveine   estimator~$\utLbf^{\ast}_{+(T) }\big(\tilde{\bf L}_{\text{\tiny{KIca}}} \big)$ with kernel-ICA preliminary.  
Irrespective of the preliminary, there is a clear and quite significant  visual enhancement, attributable to the use of ranks, in the $R$-estimation method. Our R-estimators, moreover, substantially outperform the signed-rank~ones. 

Those eye-inspection conclusions are confirmed and reinforced by the graphs in Figure~\ref{fig:MultiStepErrors}, which  reports  the Amari errors $\text{AE}\big({\bf L}^{\star}, \utLbf^{\ast}_{(T)}( \tilde{\mathbf{L}} )  \big)$  (\ref{eq:AmariError}) for the $R$-  and $R_+$-estimators of~${\bf L}^{\star}$ and $T=0,\ldots,20$. 
As $T$ increases, for all multistep $R$-estimators those errors  appear to  converge to some common limit  independent of the preliminary $\tilde{\mathbf{L}}$. For $\tilde{\bf L}=\tilde{\bf L}_{\text{\tiny{FIca}}}$ or~$\tilde{\bf L}_{\text{\tiny{KIca}}}$, the   decrease   is  quite significant over $T=1,\ldots,5$. The same decrease   is much slower \linebreak  for $\tilde{\bf L}=\tilde{\bf L}_{\text{\tiny{Fobi}}}$, but the final result, as $T$ gets close to 20, is the same, suggesting that  rank-based corrections eventually do compensate for a poorer performance of the preliminary. 
The same Amari errors $\text{AE}\big(  {\bf L}^{\star}, \utLbf^{\ast}_{+(T) }\big( \tilde{\mathbf{L}} \big) )  \big)$ were evaluated for the multistep (and data-driven-score) versions $\utLbf^{\ast}_{+(T) }\big( \tilde{\mathbf{L}} \big)$ of the  Ilmonen and Paindaveine $R_+$-estimators. The results,  in  Figure~\ref{fig:MultiStepErrorsSignedRank}, clearly show that   signed-ranks fail, 
which is hardly surprising, since there is little reason for ``greyness"  in the source images considered here to exhibit any symmetric behavior. \vspace{8mm}

\begin{figure}[htbp]
\begin{center}

\begin{subfigure}[b]{0.45\textwidth}
\begin{center}
\includegraphics[width=3in]{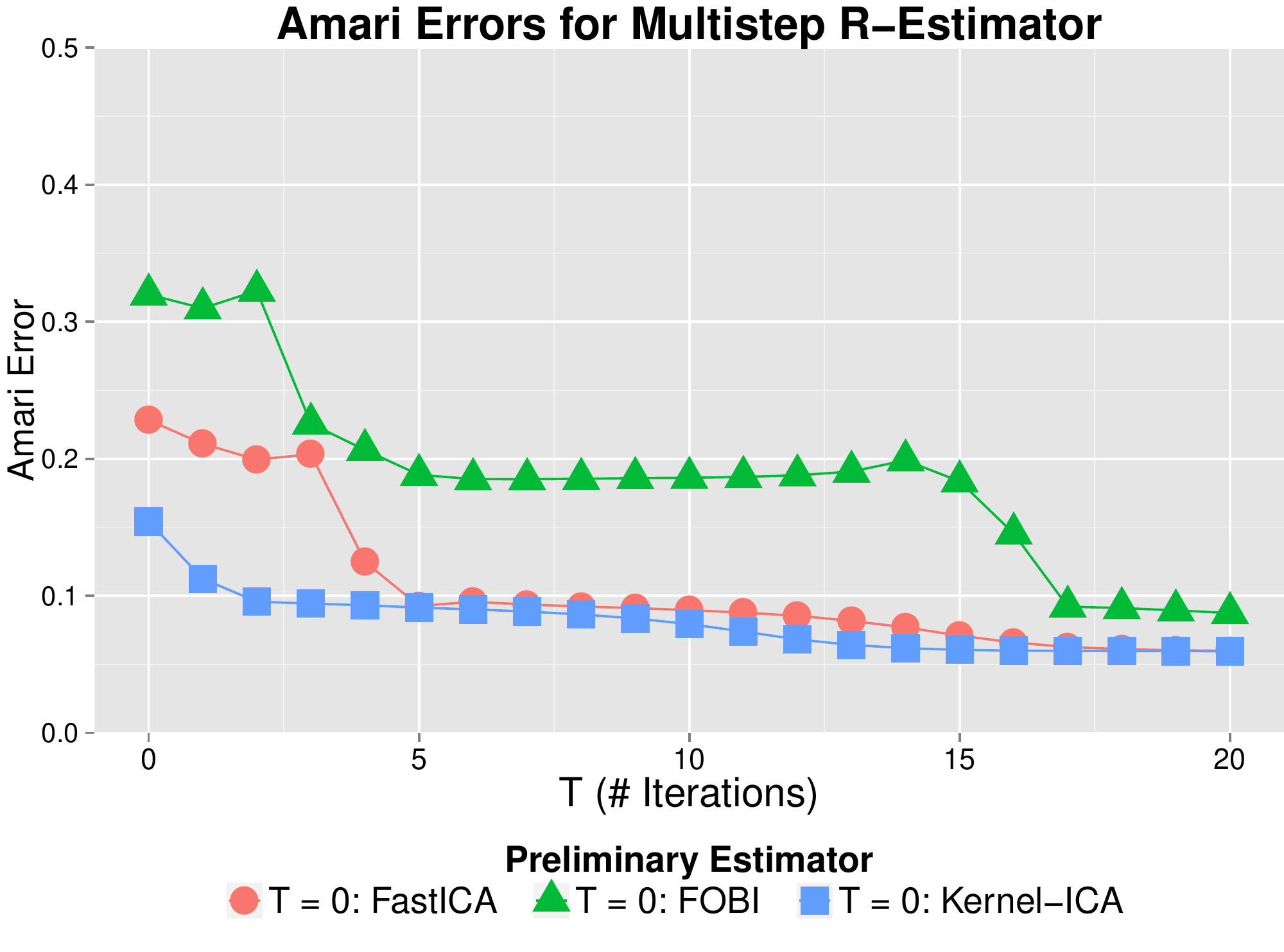}  \vspace{-4mm}
\end{center} \vspace{-4mm}
\captionsetup{font=footnotesize}
\caption{Multistep $R$-estimators. 
 \label{fig:MultiStepErrorsR} 
}
\end{subfigure}%
\qquad 
\begin{subfigure}[b]{0.45\textwidth}
\begin{center}
\includegraphics[width=3in]{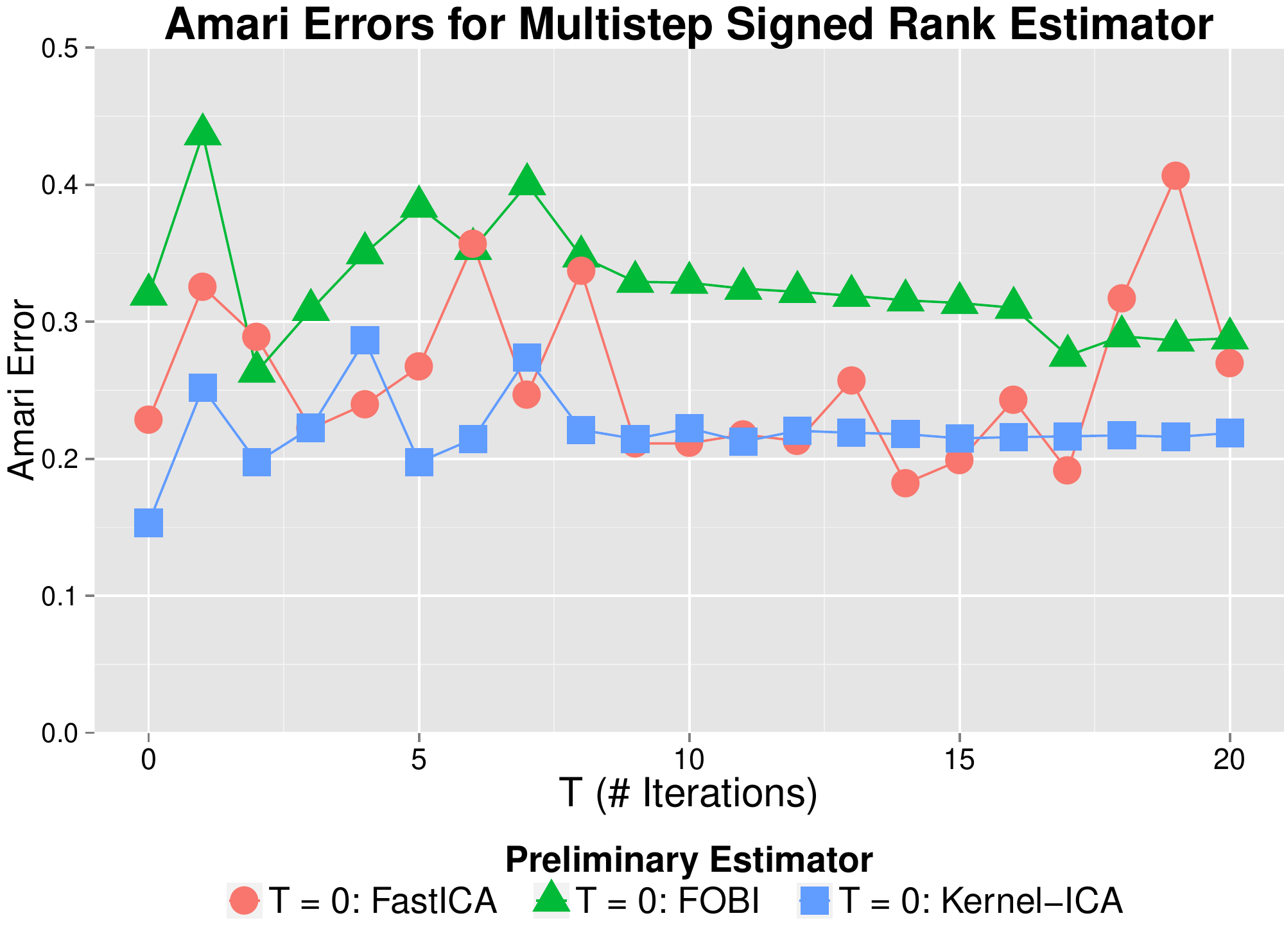} \vspace{-4mm}
\end{center} \vspace{-4mm}
\captionsetup{font=footnotesize}
\caption{ Multistep $R_+$-estimators.}
  \label{fig:MultiStepErrorsSignedRank}
\end{subfigure}%
\end{center}
\captionsetup{font=small}
\caption{ The Amari errors $\text{AE}\big(  {\bf L}^{\star}, \hat{\bf L} \big)$  
  for the multistep $R$-estimators $\utLbfCaption^{\ast}_{(T)}( \tilde{\mathbf{L}} )\vspace{-2.8mm}$ 
  and the multistep $R_+$-estimators  $\utLbfCaption^{\ast}_{+(T) }\big( \tilde{\mathbf{L}} \big)$  
  shown  in Figure~\ref{fig:EstimatedImages} and based on the preliminary  
   estimators $\tilde{\bf L} = \tilde{\bf L}_{\text{\tiny{Fobi}}}\vspace{-3mm}$,  $\tilde{\bf L}_{\text{\tiny{FIca}}}$, and $\tilde{\bf L}_{\text{\tiny{KIca}}}\vspace{-2mm}$, for $T=1, \ldots, 20$. \label{fig:MultiStepErrors}}\end{figure}

\bibliography{S6-bibliography-ICA-REstimation}
\nocite{*}
\newpage 
\appendix
\section{Supplemental material: Proofs}
\subsection{Proof of Proposition \ref{prop:ParametricULAN}}
\citet{opt2010} establish ULAN  for  ICA models under the assumption that each~$f_j$ is symmetric. Their proof consists in showing that the sufficient conditions of  Lemma 1 in Swensen~(1985) are satisfied.  
{\it Mutatis mutandis}, that proof  still goes through in the present case, with the same central sequence; only the information matrix is affected.
That  matrix depends on the covariance matrix of $\text{vec}\big({\bf T}\n_{{\bf L}; {\pmb \mu}, f}  \big)$ under  $\mathrm{P}\n_{{\pmb \mu},{\bf L}; f}$, which takes the form $$
\mathrm{E}\big[ \text{vec}\big( {\bf T}\n_{{\bf L}; {\pmb \mu}, f} \big) \text{vec}\big( {\bf T}\n_{{\bf L}; {\pmb \mu}, f} \big)\pr   \big] = \sum_{ r, s, p, q=1 }^k  \mathrm{E}\big[ \big( {\bf T}\n_{{\bf L}; {\pmb \mu}, f}  \big)_{r,p} \big( {\bf T}\n_{{\bf L}; {\pmb \mu}, f}  \big)_{s,q}   \big]  \mathbf{e}_p \mathbf{e}_q^{\prime} \otimes  \mathbf{e}_r \mathbf{e}_s^{\prime}.
$$
Because $\big({\bf T}\n_{{\pmb \mu}, {\bf L}, f}  \big)_{r,p}$ is a sum of i.i.d.\ random variables with expectation zero, 
$$
\mathrm{E}\big[ \big( {\bf T}\n_{{\bf L}; {\pmb \mu}, f}  \big)_{r,p} \big( {\bf T}\n_{{\bf L}; {\pmb \mu}, f}  \big)_{s,q}   \big] = \mathrm{E}\big[ (\varphi_{f_r}( Z_{1,r} ) Z_{1,p} - \delta_{rp})(\varphi_{f_s}(Z_{1,s}) Z_{1,q} - \delta_{sq}) \big]\quad r, s, p , q\in \{1, \ldots, k \}
$$
where the $Z_{1,j}$'s are i.i.d.\  with density $f_j$ under  $\mathrm{P}\n_{{\pmb \mu},{\bf L}; f}$ and $\delta_{rp} $ is
 the classical Kronecker index.
Evaluating those  expectations 
 yields ${\bf G}_f$ defined in (\ref{eq:Gf}). \cqfd


\subsection{Proofs for Propositions \ref{prop:RankCenSeqEquivalence} and \ref{prop:AsymptoticRepresentation}  }
Propositions \ref{prop:RankCenSeqEquivalence}(i) and \ref{prop:AsymptoticRepresentation}(i)  follow from Lemma \ref{lemma:rankCorrelationEquivalence} below,  itself adapted from Theorem~V.1.8 in \citet{hs1967}. 
Consider a triangular array $\big( U_1\n, V_1\n \big),\ldots ,\big(U_n\n, V_n\n \big)$, $n\in\mathbb{N}$ and two scores $\varphi_U,\, \varphi_V$ such that   
\begin{enumerate}
\item[(D1)]  $ U_i\n$ and  $V_i\n$,  $i=1, \ldots, n$,  are   uniform over $[0,1]$ and mutually independent, and 
\item[(D2)]  $\varphi_U,\, \varphi_V: (0,1) \to \mathbb{R}$ are  square-integrable and satisfy~(A5).
\end{enumerate}
Denote by  $R_i\n$   the rank of $U_i\n$ amongst $U_i\n, \ldots, U_i\n$, by $Q_i\n$ the rank of $V_i\n$  amongst $V_1\n, \ldots, V_n\n$, and define\vspace{-2mm}
\begin{align*}
a\n_{\mathrm{ex}}(i) &:= \mathrm{E}\big[ \varphi_U( U_1\n ) | R_1\n = i \big],  & a\n_{\mathrm{appr}}(i) &:= \varphi_U \big( \frac{i}{n+1} \big),  \\
b\n_{\mathrm{ex}}(i) &:= \mathrm{E}\big[ \varphi_V( V_1\n ) | Q_1\n = i \big], \qquad  \text{and} &  b\n_{\mathrm{appr}}(i) &:= \varphi_V \big( \frac{i}{n+1} \big).
\vspace{-2mm}\end{align*}
\normalsize\vspace{-12mm}

\noindent Assumption (D2) implies
\begin{equation}
\label{eq:NoetherCondition}
\lim_{n \to \infty} \frac{ \sum_{i=1}^n \big( a_{\mathrm{appr}}\n ( i ) - \overline{a}\n \big)^2 }{ \max_{1 \leq i \leq n} \big( a_{\mathrm{appr}}\n ( i ) - \overline{a}\n \big)^2 } = \infty \quad \text{and} \quad \lim_{n \to \infty} \frac{ \sum_{i=1}^n \big( b_{\mathrm{ex}}\n( i ) - \overline{\varphi}_V \big)^2 }{ \max_{1 \leq i \leq n} \big( b_{\mathrm{ex}}\n( i ) - \overline{\varphi}_V \big)^2 } = \infty.
\end{equation}
Let 
\begin{equation}
\label{eq:lemma:LinearRankCorrelationStatisticExact}
S_{\mathrm{ex}}\n := \frac{1}{\sqrt{n}} \sum_{i=1}^n  \Big( a_{\mathrm{ex}}\n \big( R_i\n \big) b_{\mathrm{ex}}\n \big( Q_i\n \big)  -  \bar{\varphi_U} \bar{\varphi_V}   \Big) , 
\end{equation}
where $\bar{\varphi_U} := \int_0^1 \varphi_U (u) \mathrm{d}u$ and  $\bar{\varphi_V} := \int_0^1 \varphi_V (v) \mathrm{d}v$; note that 
$$\overline{\varphi}_U = \mathrm{E}\big[ \varphi_U( U_1\n )\big]  = \frac{1}{n}\sum_{i=1}^n \mathrm{E}\big[ \varphi_U( U_1\n ) | R_1\n = i \big]  = \frac{1}{n}\sum_{i=1}^n a_{\mathrm{ex}}\n \big( i  \big)\vspace{-2mm}$$
and, similarly, $\overline{\varphi}_V =  \frac{1}{n}\sum_{i=1}^n b_{\mathrm{ex}}\n ( i )$. Also  define \vspace{-2mm}
\begin{equation}
\label{eq:lemma:LinearRankCorrelationStatisticApprox}
S_{\mathrm{appr}}\n := \frac{1}{\sqrt{n}} \sum_{i=1}^n \left( a_{\mathrm{appr}}\n \big( R_{i}\n \big) b_{\mathrm{appr}}\n \big( Q_{i}\n \big)  - \overline{a}_{\mathrm{appr}}\n \overline{b}_{\mathrm{appr}}\n \right),
\vspace{-2mm}\end{equation}
where $\overline{a}_{\mathrm{appr}}\n := \frac{1}{n} \sum_{i=1}^n a_{\mathrm{appr}}\n \big( i \big)$ and $\overline{b}_{\mathrm{appr}}\n := \frac{1}{n} \sum_{i=1}^n b_{\mathrm{appr}}\n \big( i \big)$. The following 
Lemma 
 shows that  both $S_{\mathrm{ex}}\n$ and $S_{\mathrm{appr}}\n$ admit the asymptotic representation \vspace{-2mm}
\begin{equation}
\label{eq:lemma:LinearRankCorrelationStatisticT}
T\n := \frac{1}{\sqrt{n}} \sum_{i=1}^n  \left( \varphi_U \big( U_i\n \big) \varphi_V \big( V_i\n \big)  -  \overline{\varphi}_U\n  \overline{\varphi}_V\n   \right), 
\vspace{-2mm}\end{equation}
where $\overline{\varphi}_U\n = \frac{1}{n} \sum_{i=1}^n \varphi_U \big( U_i\n \big)$ and $\overline{\varphi}_V\n = \frac{1}{n} \sum_{i=1}^n \varphi_V \big( V_i\n \big)$.
\begin{lemma}\label{lemma:rankCorrelationEquivalence}
Let   $\big( U_1\n, V_1\n \big),\ldots , \big( U_n\n, V_n\n \big)
$ and the scores
  $\varphi_U$, $\varphi_V$ 
satisfy (D1)-(D2). Then, as $n\to\infty$,\vspace{-4mm}
\begin{equation}
(i)~~S_{\mathrm{appr}}\n   = S_{\mathrm{ex}}\n + o_{L^2}( 1 ) \qquad \text{and} \qquad (ii)~~S_{\mathrm{appr}}\n = T\n + o_{L^2}( 1 ), \label{eq:lemma:rankExactAsymptoticEquivalence}
\vspace{-2mm}\end{equation}
with $S_{\mathrm{ex}}\n$, $S_{\mathrm{appr}}\n$, and $T\n$ defined in (\ref{eq:lemma:LinearRankCorrelationStatisticExact}), (\ref{eq:lemma:LinearRankCorrelationStatisticApprox}), and (\ref{eq:lemma:LinearRankCorrelationStatisticT}), respectively. 
\end{lemma}
\begin{proof}
Let us show that \vspace{-4mm}
\begin{equation}
\label{eq:lemma:rankCorrelationEquivalence1}
(i\pr ) ~~\lim_{n \to \infty} \mathrm{E}\big[ \big( S_{\mathrm{appr}}\n - S_{\mathrm{ex}}\n \big)^2 \big] = 0 \qquad \text{and} \qquad (ii\pr )~~  \lim_{n \to \infty} \mathrm{E}\big[ \big( S_{\mathrm{ex}}\n - T\n \big)^2 \big] = 0; 
\vspace{-2mm}\end{equation}
while $\ref{eq:lemma:rankExactAsymptoticEquivalence}(i)$ is the same as $(i\pr )$, $\ref{eq:lemma:rankExactAsymptoticEquivalence}(ii)$ is a consequence of $(i\pr )$, $(ii\pr )$ and   the triangle inequality. 

Defining the \emph{antirank} of $V_i\n $ with respect to $U_i\n $ by $Q_{i; \ast}\n   := \{ r : R_r\n   = i \}$ (so that~$R_{Q_{i; \ast}\n}\n  = i$), the sequence $\big( Q_{1; \ast}\n, \ldots, Q_{n; \ast}\n \big)$ is uniformly distributed over $\big\{1, \ldots, n \big\}$ in view of the  independence between the  $U_i\n$'s and the $V_i\n$'s. 
Reordering  terms, we have \vspace{-2mm}
$$
S_{\mathrm{appr}}\n := \frac{1}{\sqrt{n}} \sum_{i=1}^n \big( a_{\mathrm{appr}}\n \big( i \big)  - \overline{a}_{\mathrm{appr}}\n \big) b_{\mathrm{appr}}\n \big( Q_{i; \ast}\n \big) \quad \text{and} \quad S_{\mathrm{ex}}\n := \frac{1}{\sqrt{n}} \sum_{i=1}^n \big( a_{\mathrm{ex}}\n \big( i \big)  - \bar{\varphi_U} \big) b_{\mathrm{ex}}\n \big( Q_{i; \ast}\n \big).
\vspace{-2mm}$$
Write $S_{\mathrm{ex}}\n = S_{\ast; 1}\n + S_{\ast; 2}\n$, where\vspace{-2mm}
\small
$$
S_{\ast; 1}\n := n^{-1/2}\sum_{i=1}^n  \big( a_{\mathrm{appr}}\n \big( i \big)  - \overline{a}_{\mathrm{appr}}\n \big) b_{\mathrm{ex}}\n \big( Q_{i; \ast}\n  \big) \quad \text{and} \quad  S_{\ast; 2}\n := \frac{1}{\sqrt{n}}\sum_{i=1}^n  \big( b_{\mathrm{ex}}\n \big( i \big) - \bar{\varphi_U}  \big)\big(  a_{\mathrm{ex}}\n \big( R_{i; \ast}\n \big)  - a_{\mathrm{appr}}\n \big( R_{i; \ast}\n \big)  \big),\vspace{-2mm}
$$
\normalsize
where $R_{i; \ast}\n := \{ r : Q_r\n = i \}$ denotes the antirank of $U_i\n$ with respect to $V_i\n$. Assumption~(A5), (\ref{eq:NoetherCondition}), Lemma  V.1.6a, and Theorem  V.1.6a from \citet{hs1967} together imply
 $
\lim_{n \to \infty} \mathrm{E} \big[ \big( S_{\mathrm{appr}}\n -  S_{\ast; 1}\n \big)^2 \big] = 0$ and $ \lim_{n \to \infty} \mathrm{E} \big[ \big( S_{\ast; 2}\n \big)^2 \big] = 0,
$ 
which, along with the triangle inequality, establishes $(i\pr )$ in (\ref{eq:lemma:rankCorrelationEquivalence1}).

Let ${\bf U}\n_{(\cdot)} := \big( U_{(1)}\n, \ldots, U_{(n)}\n \big)\pr$ and ${\bf V}\n_{(\cdot)} := \big( V_{(1)}\n, \ldots, V_{(n)}\n \big)\pr$ denote the order statistics for the $n$-tuples $\{U_i\n \}_{i=1}^n$ and $\{V_i\n \}_{i=1}^n$, respectively.   
Because the antiranks $R_{1; \ast}\n$ are uniformly distributed and independent of $R_1\n, \ldots, R_n\n$, the $R_{1; \ast}\n$th order statistic  $U_{(R_{1; \ast}\n)}$ is uniformly distributed over the unit interval (the same holds true for the $Q_{1; \ast}\n$th order statistic  $V_{(Q_{1; \ast}\n)}$). Write $T\n = T_{\ast; 1}\n  + T_{\ast; 2}\n$, where\vspace{-4mm}
$$
T_{\ast; 1}\n := \frac{1}{\sqrt{n}}\sum_{i=1}^n  \big( a_{\mathrm{ex}}\n \big( i \big) -  \overline{\varphi}_U \big) \varphi_V \big( V_{( Q_{i; \ast}\n )} \big) \vspace{-5mm}$$ and   \vspace{-2mm} 
$$T_{\ast; 2}\n :=  \frac{1}{\sqrt{n}}\sum_{i=1}^n  \big( \varphi_V \big( V_i \big)  -  \overline{\varphi}_V\n \big) \big( \varphi_U( U_{(R_{i; \ast}\n) })  - a_{\mathrm{ex}}\n \big( R_{i; \ast}\n \big) \big).
\vspace{-1mm}$$
Then (\ref{eq:NoetherCondition}) and Theorem  V.1.5a from \citet{hs1967} imply that \vspace{-2mm}
$$
\lim_{n \to \infty} \mathrm{E} \big[ \big( S_{\mathrm{ex}}\n -  T_{\ast; 1}\n \big)^2 \big] = 0 \qquad \text{and} \qquad \lim_{n \to \infty} \mathrm{E} \big[ \big( T_{\ast; 2}\n \big)^2 \big] = 0,
\vspace{-2mm}$$
which establishes $(ii\pr )$ in (\ref{eq:lemma:rankCorrelationEquivalence1}).
\end{proof} 
\begin{proof}[Proof of Proposition \ref{prop:RankCenSeqEquivalence}]
All expectations  in this section are under  $\mathrm{P}^{(n)}_{\pmb {\mu}, {\bf L}, f}$, unless otherwise specified;  ${\bf R}_i\n $ stands for $ {\bf R}_i\n({\bf L})$, $i=1, \ldots, n$.  
For part (i) of the proposition to hold, it is sufficient that,     for $\utT\n_{ {\bf L}, f; \mathrm{ex}}$ and $\utT\n_{ {\bf L}, f}$  
   in (\ref{exdelta}) and (\ref{rankScoreMatrix}), \vspace{-2mm}
\begin{equation}
\label{eq:proofprop2toshow1}
\big(\utT\n_{ {\bf L}, f }\big)_{rs}  =  \big( \utT\n_{ {\bf L}, f; \mathrm{ex}} \big)_{rs} + o_{L^2}(1) \quad\text{for all $r, s \in \{1, \ldots, k \}$, as $n \to \infty$}.
\end{equation}
 First, fix $r \ne s \in \{1, \ldots, k \}$. Then,  \vspace{-2mm}
$$
\big( \utT\n_{ {\bf L}, f; \mathrm{ex}} \big)_{rs} = \frac{1}{\sqrt{n}} \sum_{i=1}^n \mathrm{E}\left[ J_{f_r}\left( U_{1r}\n \right)   \big| R_{ir}\n \right] \mathrm{E}\left[ F_s\inv \left( U_{1s} \n \right) \big|  R_{is}\n \right]\vspace{-2mm}
$$
by independence between distinct components, and\vspace{-2mm}
$$
\big(\utT\n_{ {\bf L}, f }\big)_{rs} := \frac{1}{\sqrt{n}} \sum_{i=1}^n \Big( J_{f_r} \big( \frac{R_{ir}\n}{n+1} \big) F_s\inv \big( \frac{ R_is\n }{n+1} \big) - \overline{J_{f_r}}\n  \overline{F_s\inv}\n \Big).\vspace{-2mm}
$$
Letting $\phi_U=J_{f_r}$ and $\phi_V=F_s\inv$,  (\ref{eq:proofprop2toshow1}) (for $r \ne s$) thus directly follows from Lemma~A1.
For~$r=s$, the H\'{a}jek projection theorem for linear rank statistics and the convergence rate of Riemann sums imply\vspace{-2mm}
\begin{eqnarray*}
\big( \utT\n_{ {\bf L}, f; \mathrm{ex} } \big)_{rr} \!\! &\!\! :=\!\! &\!\!   n^{-\half} \sum_{i=1}^n \left(  \mathrm{E} \left[ J_{f_r}\left( U_{ir}\n \right) F_r\inv \left( U_{ir} \n \right) \big|  R_{ir}\n \right] - 1 \right) \\
&=& n^{-\half} \sum_{i=1}^n \left( J_{f_r}\left( \frac{R_{ir}\n}{n+1} \right) F_r\inv \left( \frac{R_{ir}\n}{n+1} \right) - 1 \right)   + o_{L^2}(1) \\
&=& n^{-\half} \left( \frac{1}{n} \sum_{i=1}^n J_{f_r}\left( \frac{i}{n+1} \right) F_r\inv \left( \frac{i }{n+1} \right) \! -\!  \int_0^1\!\! J_{f_r}(u) F_r\inv(u) \mathrm{d}u  \right) + o_{L^2}(1)
=o_{L^2}(1) 
\end{eqnarray*}
as $n \to \infty$, under $\mathrm{P}^{(n)}_{\pmb {\mu}, {\bf L}, f}$. 
This establishes part (i) of Proposition~\ref{prop:RankCenSeqEquivalence}. As for part~(ii), it follows from the results in \citet{hw2003} that  $\utDelta\n_ {\mathbf{L}, \pmb \mu, f; \mathrm{ex} } = \mathbf{\Delta}^{(n) \ast}_{\mathbf{L}, \pmb \mu, f} + o_{L^2}(1) $  as~$n \to \infty$, under $\mathrm{P}^{(n)}_{\pmb {\mu}, {\bf L}, f}$. 
This, along with part (i) of the proposition   and the triangle inequality,  implies part (ii). 
\end{proof} 

\begin{proof}[Proof of Proposition \ref{prop:AsymptoticRepresentation} ]
In order to establish  part (i) of the proposition, it is sufficient to show  that,   for every $r \ne s \in \{1, \ldots, k\}$,
 $
\big( \utT_{{\bf L}; f}\n \big)_{rs}  = \big(  {\bf T}^{\diamond (n)}_{ {\bf L}, {\pmb \mu}; f, g}  \big)_{rs} + o_{L^2}(1) 
$  
as $n \to \infty$, under~$\mathrm{P}^{(n)}_{\pmb \mu, {\bf L}, g}$. Let $\mathbf{V}_i\n := \mathbf{G} \left( \mathbf{Z}_i^{(n)} \right)=:( V_{i1}\n, \ldots, V_{ik}\n)\pr$, $i=1,\ldots,n$. The rank of $V_{ij}\n$ amongst $V_{1j}\n, \ldots, V_{nj}\n$ is $R_{ij}\n ( {\bf L} )$ for each $j=1, \ldots, k$. The   claim  follows from Lemma~\ref{lemma:rankCorrelationEquivalence} by taking  score functions    $J_{f_r}$ and $F_s\inv$. 

The proof for parts (ii) and (iii) follow from that of Theorem 3.2(ii) and (iii) in \citet{ip2011}. However,  the presence of asymmetry in the independent components implies   different cross-information matrices. The result is obtained, via 
 Le~Cam's Third Lemma, from an evaluation of the covariance matrix in the asymptotically normal  joint distribution of  $\pmb{\Delta}^{\diamond (n)}_{ {\bf L}, {\pmb \mu}; f, g}$ and (\ref{eq:ParametricULANRepresentation}) under $\mathrm{P}^{(n)}_{\pmb \mu, {\bf L}, g}$. 
That covariance matrix  follows from the covariance of $\pmb{\Delta}^{\diamond (n)}_{ {\bf L}, {\pmb \mu}; f, g}$ and $\pmb{\Delta}_{\mathbf{L}, \pmb \mu, g}^{(n)}$, under $\mathrm{P}^{(n)}_{\pmb \mu, {\bf L}, g}$ which depends on  \vspace{-2mm}
$$
\mathrm{E}\big[ \text{vec}\big({\bf T}\n_{{\pmb \mu}, {\bf L}, g}  \big) \text{vec}\big( {\bf T}_{{\pmb \mu}, {\bf L}, f, g, \diamond}\n  \big)\pr   \big] = \sum_{ \substack{r, s, p, q=1 \\ r \ne s} }^k  \mathrm{E}\big[ \big({\bf T}\n_{{\pmb \mu}, {\bf L}, g}  \big)_{r,p} \big( {\bf T}_{{\pmb \mu}, {\bf L}, f, g, \diamond}\n  \big)_{s,q}   \big]  \mathbf{e}_p \mathbf{e}_q^{\prime} \otimes  \mathbf{e}_r \mathbf{e}_s^{\prime}.
\vspace{-3mm}$$
Evaluating this expression eventually 
yields the value of ${\bf G}_{f,g}$ appearing in~(\ref{eq:CrossInfMatrix}) for the cross-information matrix. 
\end{proof}

\section{Supplemental material: further simulation results}

 Figures \ref{fig:MDISimGroup1} - \ref{fig:MDISimGroup4} below are summarizing  the same simulation results as   Figures \ref{fig:SimGroup1} - \ref{fig:SimGroup4}, with Amari errors replaced with the  minimum distance index proposed by  Ilmonen et al.~(2010).  Conclusions are essentially similar.

\newpage
\begin{landscape}
\vspace{-6mm}
\begin{figure}
\begin{subfigure}[h]{1\linewidth}
\captionsetup{font=footnotesize}
\caption{Sample size $n=100$}
\begin{center}
\includegraphics[scale=0.44]{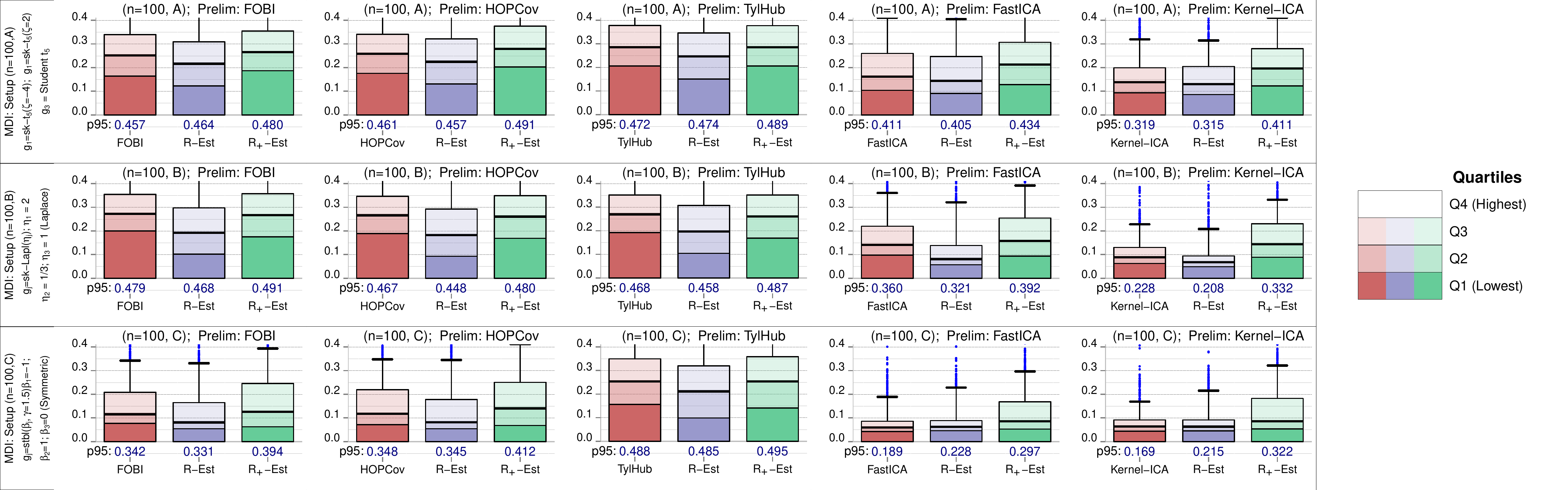}
\end{center}
\label{fig:MDISimGroup1-n100}
\end{subfigure}%
\\ \\
\begin{subfigure}[h]{1\linewidth}
\captionsetup{font=footnotesize}
\caption{Sample size $n=1,000$}
\begin{center}
\includegraphics[scale=0.44]{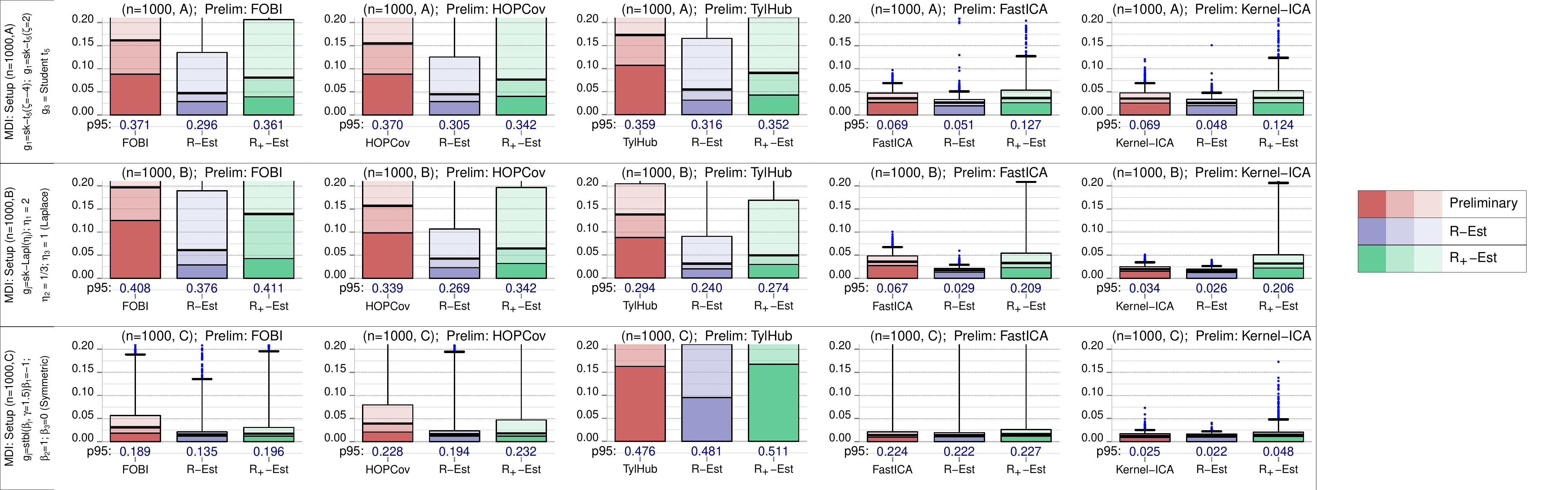}
\end{center}
\label{fig:MDISimGroup1-n1000}
\end{subfigure}%
\captionsetup{font=small}
\caption{
Boxplots of minimum distance index measurements obtained in $M=1,000$ replications of the setup 
 $(n, \text{S})$, $n~\! =~\! 100,\ 1,000$, $\text{S}=A,\, B,\, C,\vspace{1mm}$ for the preliminary 
$\tilde{\bf L}=\tilde{\bf L}_{\text{\tiny{Fobi}}}$,  
$\tilde{\bf L}_{\text{\tiny{HOPCov}}}$, 
$\tilde{\bf L}_{\text{\tiny{TylHub}}}$, 
$\tilde{\bf L}_{\text{\tiny{FIca}}}$, 
$\tilde{\bf L}_{\text{\tiny{KIca}}}\vspace{-3mm} $, the one-step $R$-estimator 
$\utLbfCaption^{\ast} ( \tilde{\bf L} )$, and  the one-step $R_+$-estimator 
$\utLbfCaption^{\ast}_+( \tilde{\bf L}) $  based on  the same preliminaries, with data-driven skew-$t$  and Student-$t$ scores, respectively. 
}\label{fig:MDISimGroup1}
\end{figure}

\begin{figure}
\begin{subfigure}[h]{1\linewidth}
\captionsetup{font=footnotesize}
\caption{Sample size $n=100$}
\begin{center}
\includegraphics[scale=0.44]{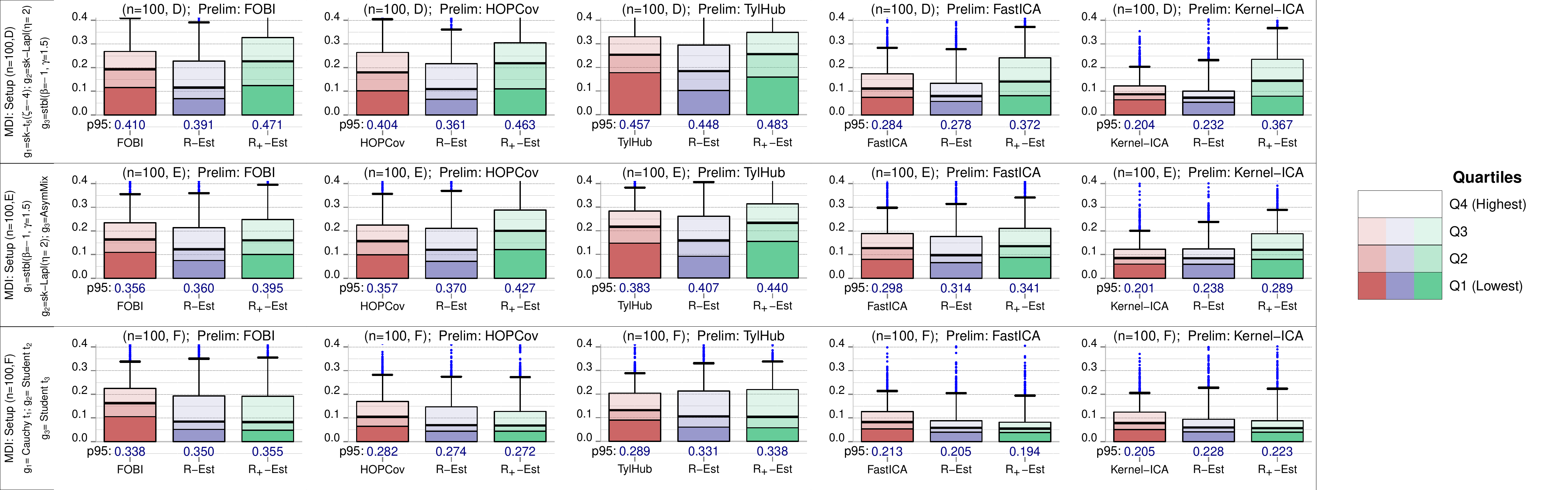}
\end{center}
\label{fig:MDISimGroup2-n100}
\end{subfigure}%
\\ \\
\begin{subfigure}[h]{1\linewidth}
\captionsetup{font=footnotesize}
\caption{Sample size $n=1,000$}
\begin{center}
\includegraphics[scale=0.44]{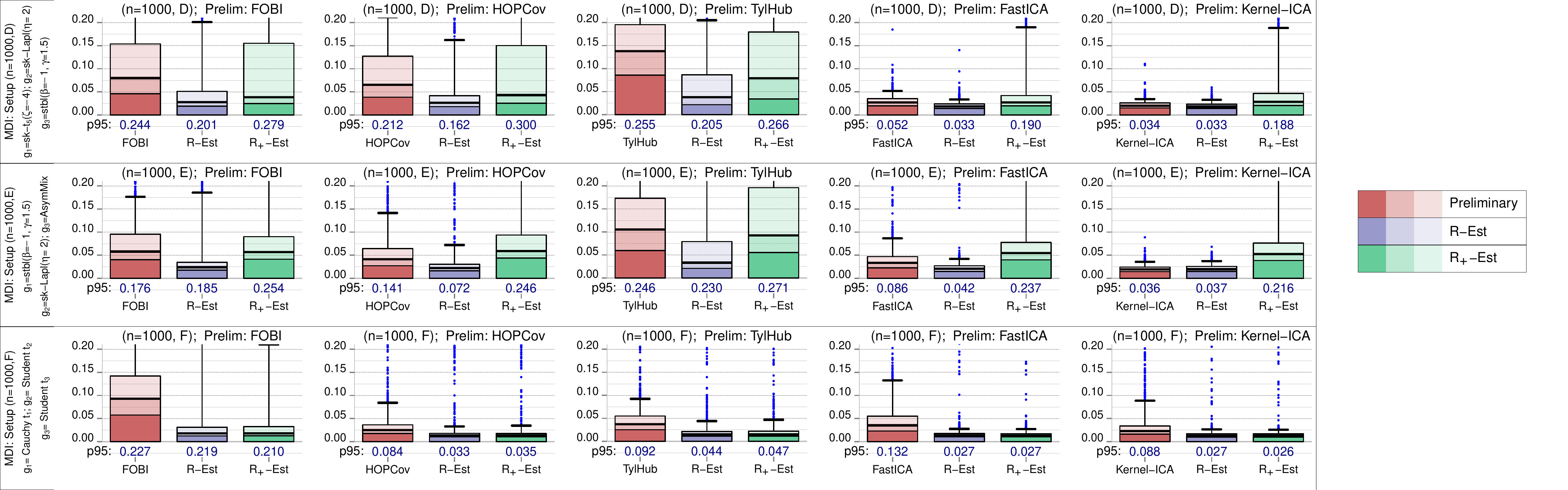}
\end{center}
\label{fig:MDISimGroup2-n1000}
\end{subfigure}%
\captionsetup{font=small}
\caption{
Boxplots of minimum distance index measurements obtained in $M=1,000$ replications of the setup $(n, \text{S})$, $n~\!\! =~\! \!100,\ 1,000$, $\text{S}=D,\, E,\, F,\vspace{1mm}$ for the 
    preliminary $\tilde{\bf L}=\tilde{\bf L}_{\text{\tiny{Fobi}}}$,  $\tilde{\bf L}_{\text{\tiny{HOPCov}}}$, 
    $\tilde{\bf L}_{\text{\tiny{TylHub}}}$, 
    $\tilde{\bf L}_{\text{\tiny{FIca}}}$, $\tilde{\bf L}_{\text{\tiny{KIca}}}\vspace{-3mm}$, the one-step $R$-estimator $\utLbfCaption^{\ast} ( \tilde{\bf L} )$, and  the one-step $R_+$-estimator  $\utLbfCaption^{\ast}_+( \tilde{\bf L}) \vspace{-1mm}$  based on  the same preliminaries with data-driven skew-$t$  and Student-$t$ scores, respectively.} 
\label{fig:MDISimGroup2}
\end{figure} 

\newpage

\begin{figure}
\begin{subfigure}[h]{1\linewidth}
\captionsetup{font=footnotesize}
\caption{Sample size $n=100$}
\begin{center}
\includegraphics[scale=0.44]{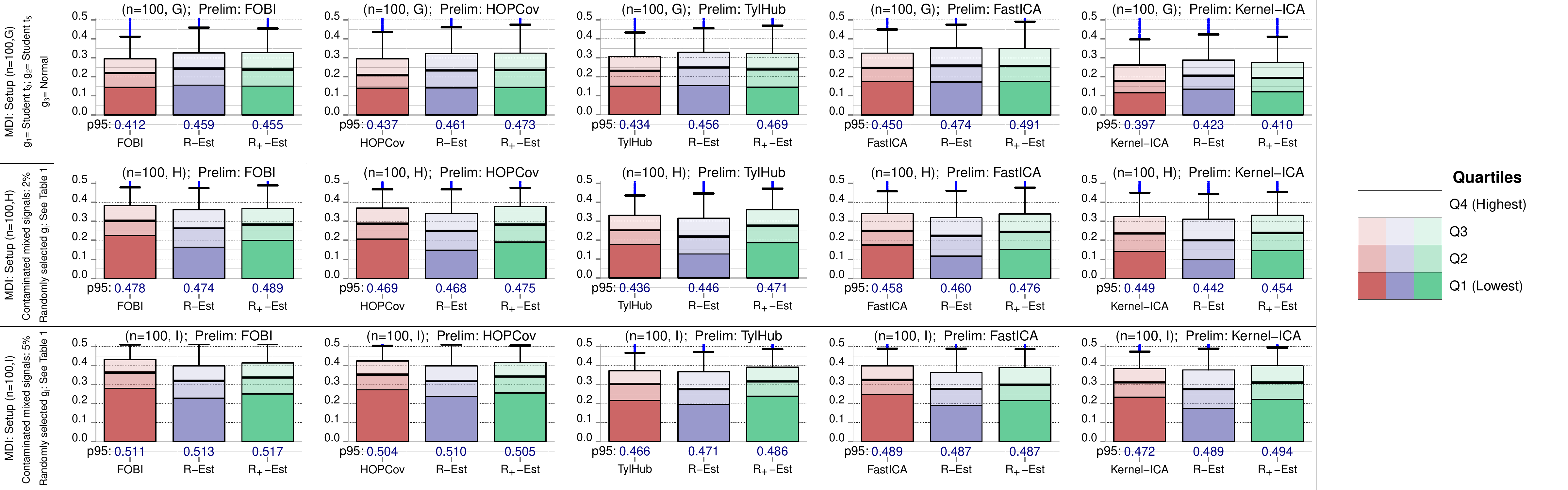}
\end{center}
\label{fig:MDISimGroup3-n100}
\end{subfigure}%
\\ \\
\begin{subfigure}[h]{1\linewidth}
\captionsetup{font=footnotesize}
\caption{Sample size $n=1,000$}
\begin{center}
\includegraphics[scale=0.44]{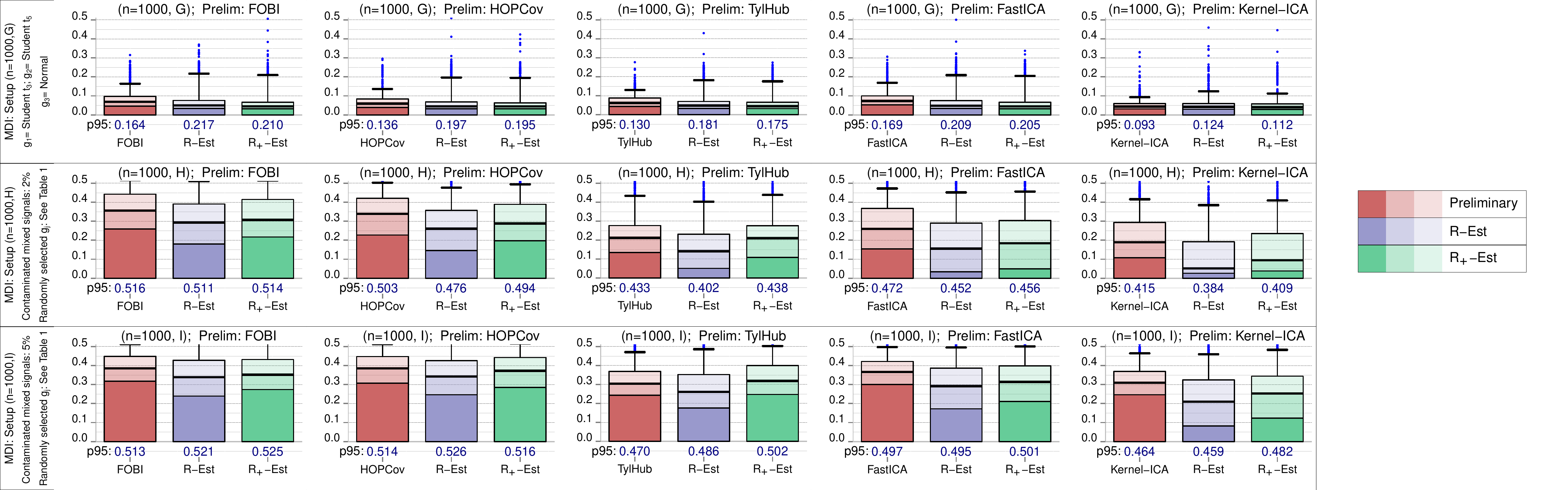}
\end{center}
\label{fig:MDISimGroup3-n1000}
\end{subfigure}%
\captionsetup{font=small}
\caption{Boxplots of minimum distance index measurements  
  obtained in $M=1,000$ replications of the setup $(n, \text{S})$, $n~\!\! =~\! \!100,\ 1,000$, $\text{S}=G,\, H,\, I,\vspace{1mm}$     for the 
    preliminary $\tilde{\bf L}=\tilde{\bf L}_{\text{\tiny{Fobi}}}$,  $\tilde{\bf L}_{\text{\tiny{HOPCov}}}$,
    $\tilde{\bf L}_{\text{\tiny{TylHub}}}$, 
     $\tilde{\bf L}_{\text{\tiny{FIca}}}$, $\tilde{\bf L}_{\text{\tiny{KIca}}}\vspace{-3mm}$, the one-step $R$-estimator $\utLbfCaption^{\ast} ( \tilde{\bf L} )$, and  the one-step $R_+$-estimator $\utLbfCaption^{\ast}_+( \tilde{\bf L})$     based on  the same preliminaries with data-driven skew-$t$ and Student-$t$ scores, respectively.} 
 \label{fig:MDISimGroup3}
   \end{figure} 
  
  \newpage
  \begin{figure}
\begin{subfigure}[h]{0.5\linewidth}
\captionsetup{font=footnotesize}
\caption{Sample size $n=100$}
\begin{center}
\includegraphics[scale=.58]{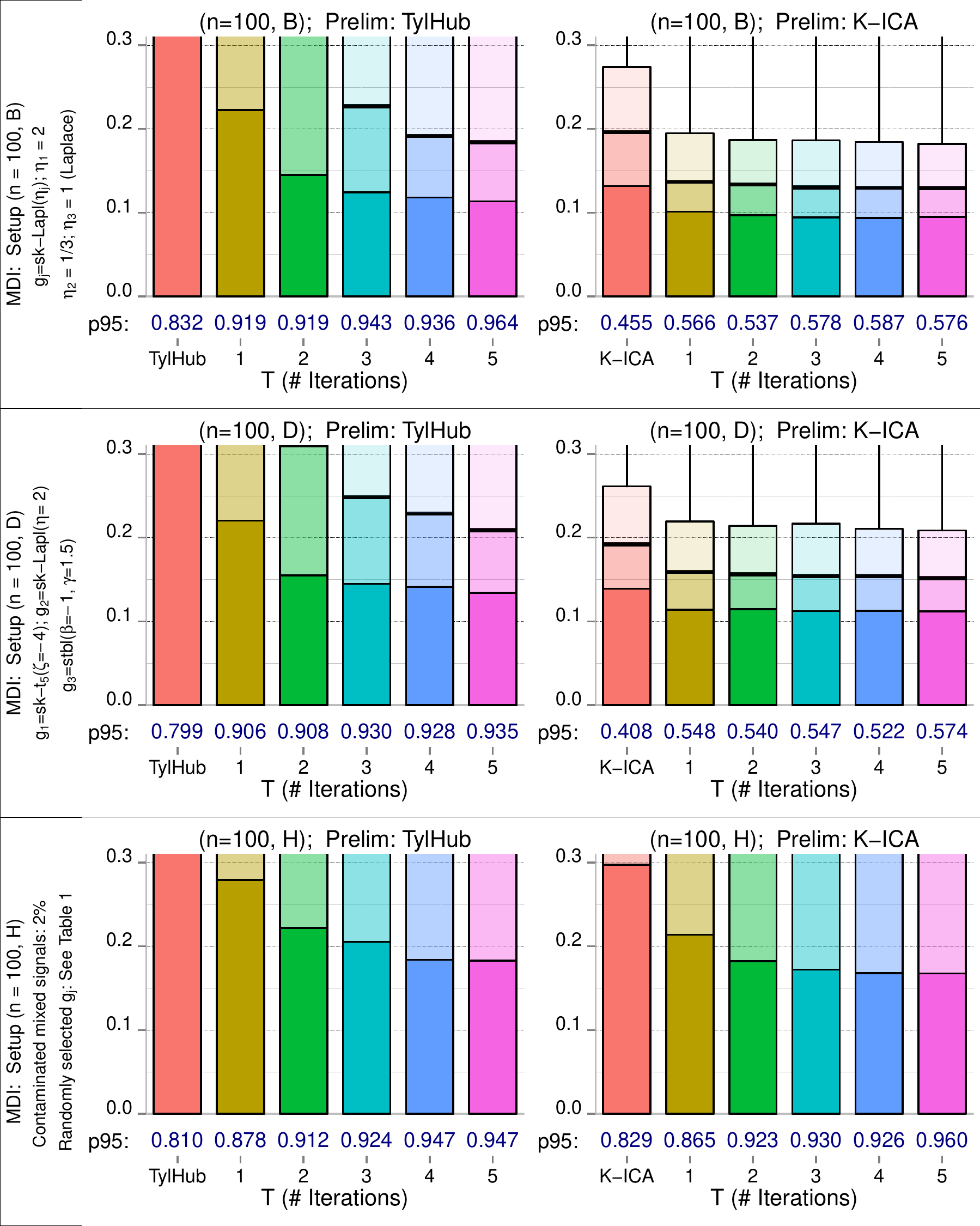}
\end{center}
\label{fig:MDISimGroup4-n100}
\end{subfigure}%
\begin{subfigure}[h]{0.5\linewidth}
\captionsetup{font=footnotesize}
\caption{Sample size $n=1,000$}
\begin{center}
\includegraphics[scale=.58]{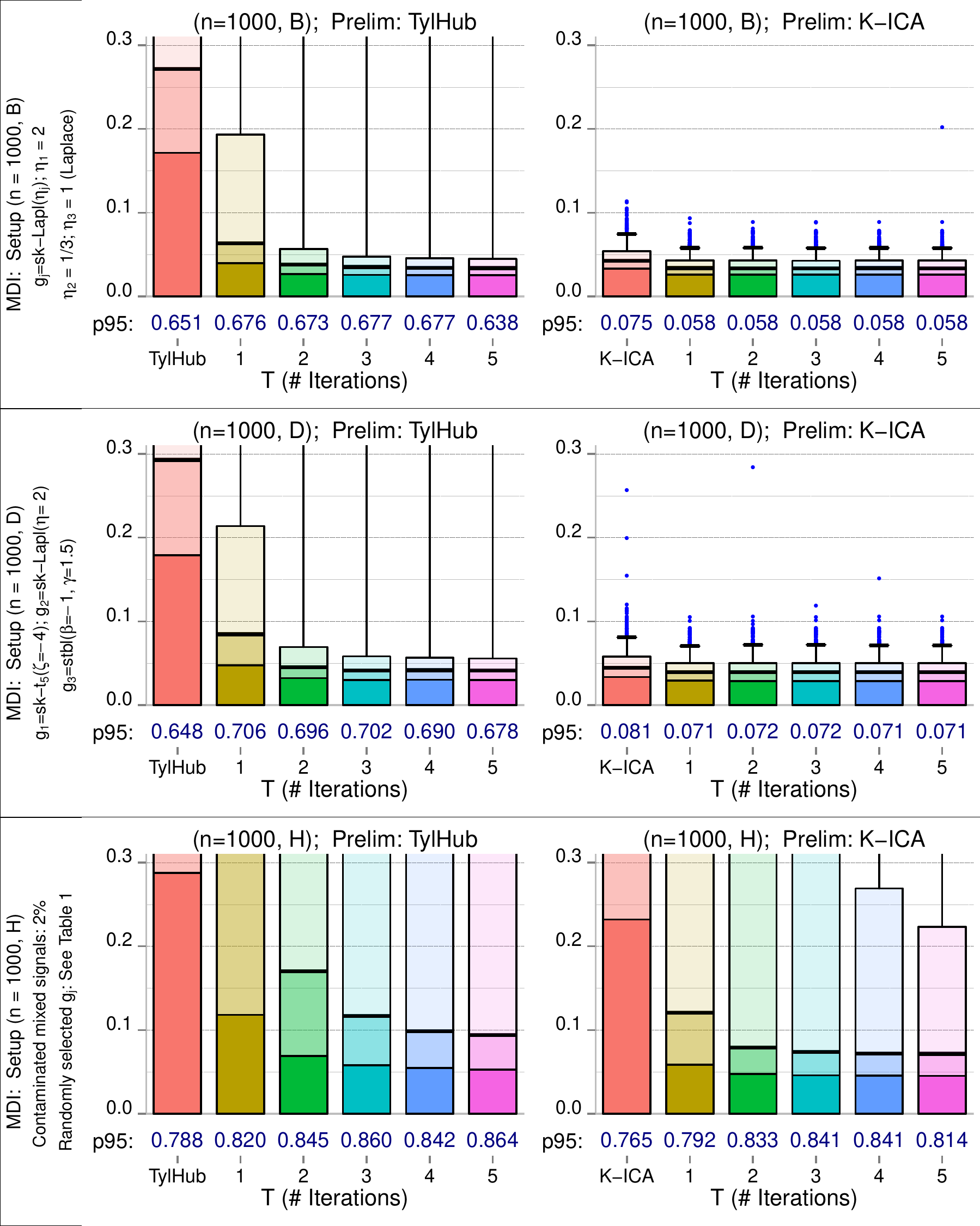}
\end{center}
\label{fig:MDISimGroup4-n1000}
\end{subfigure}%
\captionsetup{font=small}
\caption{
Boxplots of minimum distance index measurements 
  obtained in $M=1,000$ replications of the setup $(n, \text{S})$, $n=100, 1000$, $\text{S}=B,\, D,\, H,\vspace{1mm}$ 
   for the $T$-step $R$-estimator $\utLbfCaption^{\ast} ( \tilde{\bf L} )$   based on  preliminary  
      $\tilde{\bf L}=  \tilde{\bf L}_{\text{\tiny{TylHub}}}$ and     $\tilde{\bf L}_{\text{\tiny{KIca}}}\vspace{-3mm}$, respectively,
and  data-driven skew-$t$ scores, $T=1,\ldots, 10$}\label{fig:MDISimGroup4}
\end{figure}
\end{landscape}

\end{document}